\documentclass[5pt,oneside,a4paper]{article}

\usepackage{amssymb}
\usepackage{amsmath}

\usepackage{cases}
\usepackage[top=1in, bottom=1in, left=1.25in, right=1.25in]{geometry}
\usepackage{stmaryrd}

    \usepackage{cases}
    \usepackage{pgf}
    \usepackage{tikz}
    \usepackage{amsthm}
    \usepackage{framed,multicol}
    \usepackage{enumerate}
    \usepackage{makecell}
    \usepackage{listings}
    \usepackage{multicol}
    \usepackage{wrapfig}
    \usepackage{subcaption}
    \usepackage{hyperref}

\usepackage{sectsty}

\usepackage{tikz-timing}
\usepackage{tabularx}

\usepackage{mathrsfs}

\usepackage{etex} 

\usepackage{booktabs}

\usepackage{proof}

\usepackage{authblk}

\usetikzlibrary{automata}
\usetikzlibrary{shapes}
\usetikzlibrary{arrows,calc,fit}
\usetikzlibrary{snakes}
\tikzset{line/.style={draw, thick, -latex'}}

\usepackage{algorithm}
\usepackage[noend]{algpseudocode}
\makeatletter
\renewcommand{\ALG@beginalgorithmic}{\footnotesize}
\makeatother

    \newtheorem{mydef}{Definition}[section]
    \newtheorem{mytheo}{Theorem}[section]
    \newtheorem{lemma}{Lemma}[section]
    \newtheorem{prop}{Proposition}[section]

    \newcommand{\rNum}[1]{\expandafter{\romannumeral #1\relax}}
    \newcommand{\rNUM}[1]{\uppercase\expandafter{\romannumeral #1\relax}}

    \newcommand{\bff}[1]{\mathbf{#1}}
    \newcommand{\mbb}[1]{\mathbb{#1}}
    
    \newcommand{\spc}[1]{\begin{spacing}{#1}}
    \newcommand{\spce}{\end{spacing}}
    
    \newcommand{\la}[0]{\langle}
    \newcommand{\ra}[0]{\rangle}

    \DeclareCaptionFont{tiny}{\tiny}

    \newcommand{\rmn}[1]{(\romannumeral#1)}

    \newcommand{\sig}[0]{\varsigma}
    \newcommand{\act}[0]{\mathit{evt}}
    \newcommand{\noth}[0]{\mathbf{1}}
    \newcommand{\halt}[0]{\mathbf{0}}
    \newcommand{\true}[0]{\mathit{tt}}
    \newcommand{\false}[0]{\mathit{ff}}
    \newcommand{\SP}[0]{\mathbf{SP}}
    \newcommand{\Cl}[0]{\mathit{Cl}}
    \newcommand{\Var}[0]{\mathit{Var}}
    
    \newcommand{\fus}[0]{\cdot}
    \newcommand{\val}[0]{\mathit{val}}
    \newcommand{\dVar}[0]{\mathit{DV}} 
    \newcommand{\Par}[0]{\mathit{Par}}
    \newcommand{\BV}[0]{\mathit{BV}}
    \newcommand{\FV}[0]{\mathit{FV}} 
    \DeclareMathOperator{\seq}{;}
    \DeclareMathOperator{\cho}{\cup}
    \DeclareMathOperator{\para}{\cap}
    
    \newcommand{\CEvt}[0]{\mathit{CPrg}}
    \newcommand{\none}[0]{\Lambda}

    \newcommand{\Trec}[0]{\bff{Trec}}
    \newcommand{\irep}[0]{\infty}
    \newcommand{\SE}[0]{\mathit{SExe}}
    \newcommand{\Com}[0]{\mathit{rSExe}}
    \newcommand{\Must}[0]{\mathit{Must}}
    \newcommand{\Cannot}[0]{\mathit{Cant}}
    \newcommand{\ASP}[0]{\SP^{\mathit{at}}}
    \newcommand{\Chk}[0]{\mathit{Chk}}
    \newcommand{\Match}[0]{\mathit{Mat}}
    \newcommand{\Closed}[0]{\mathit{Clo}}
    \newcommand{\Signal}[0]{\mathit{Sig}}
    \newcommand{\OC}[0]{\mathit{OC}}
    
    \newcommand{\rfval}[0]{\mathit{rval}_t}

    \newcommand{\valt}[0]{\val_t}

    \newcommand{\join}[0]{\bowtie}
    \DeclareMathOperator{\sep}{|}
    \DeclareMathOperator{\nex}{.}
    \newcommand{\mval}[0]{\mathit{val}_m}
    \newcommand{\Merge}[0]{\mathit{Mer}}
    \newcommand{\rMerge}[0]{\mathit{rMer}}
    \newcommand{\append}[0]{\triangleleft}
    
    \newcommand{\place}[0]{\_}
    \newcommand{\red}[0]{\rightsquigarrow}
    \newcommand{\ToSeq}[0]{\mathit{Brz}}
    \newcommand{\strr}[0]{\mathit{str}}
    \newcommand{\trec}[0]{\mathit{trec}}
    \DeclareMathOperator{\link}{\unlhd}
    \newcommand{\comb}[0]{\mathit{comb}}
    \newcommand{\Sigs}[0]{\mathit{Sig}}
    \newcommand{\lev}[0]{\mathit{lev}}
    \newcommand{\reld}[0]{\mathit{\vdash^+}}
    
    \newcommand{\parp}[0]{\mathit{par}}
    \newcommand{\SV}[0]{\Xi}
    
    \newcommand{\seqArrow}[0]{\Rightarrow}
    \newcommand{\Can}[0]{\mathit{Can}}
    \newcommand{\getCan}[0]{\mathit{getCan}}
    \newcommand{\rgetCan}[0]{\mathit{rgetCan}}
    \newcommand{\dddef}[0]{=_{df}}
    \newcommand{\ddef}[0]{::=}

\begin{document}
\title{A Dynamic Logic for Verification of Synchronous Models based on Theorem Proving}
\author[1]{Yuanrui Zhang}
\affil[1]{School of Mathematics and Statistics, Southwest University, China}
\date{202011}

\maketitle

\tikzset{
    state/.style={
           rectangle,
           rounded corners,
           draw=black, semithick,
           minimum height=2em,
           minimum width=2em,
           inner sep=2pt,
           text centered,
           },
    state2/.style={
           rectangle,
           draw=black, semithick,
           draw=black, semithick,
           minimum height=2em,
           minimum width=2em,
           inner sep=2pt,
           text centered,
           },
    state3/.style={
           circle,
           draw=black, semithick,
           minimum height=1em,
           minimum width=1em,
           inner sep=1pt,
           text centered,
           },
    state4/.style={
           circle,
           draw=black, thick,
           minimum height=2em,
           minimum width=2em,
           inner sep=2pt,
           text centered,
           },
}

\ifx
[Basic ideas:

We have decided to propose a synchronous model following the spirit of SCCS or concurrent kleen algebra for reactive systems,
where the `synchrony' (i.e., behaviours of systems are modelled as discrete events and several events can be triggered simultaneously) has been stressed in its syntax.

The reason to take this approach is mainly that imperative synchronousl anguages, such as Esterel, are not suitable for compositional reasoning because of their `preemption operators' (this calling might not be correct and should be confirmed later).

The best is that we can propose a set of rewrite rules for reducing each Esterel statement to a structure of the synchroous model we propose.

The semantics of concurrent programs is largely inspired from \cite{???} but is quite different since the synchronous mechanism we consider.
]

[
    synchronous model is a popular model for modeling real-time systems in which physical time of system events has been abstracted away and only the notion of logical time exists.
    Synchronous programming languages such as... have been widely used in design and modelling of real-time embedded systems in industries.

]
\fi

\begin{abstract}
  Synchronous model is a type of formal models for modelling and specifying reactive systems.
  It has a great advantage over other real-time models that its modelling paradigm supports a deterministic concurrent behaviour of systems.
  Various approaches have been utilized for verification of synchronous models based on different techniques, such as model checking, SAT/SMT sovling, term rewriting, type inference and so on.
  In this paper, we propose a verification approach for synchronous models based on compositional reasoning and term rewriting.
  Specifically, we initially propose a variation of dynamic logic, called synchronous dynamic logic (SDL).
  SDL extends the regular program model of first-order dynamic logic (FODL) with necessary primitives to capture the notion of synchrony and synchronous communication between parallel programs, and
  enriches FODL formulas with temporal dynamic logical formulas to specify safety properties --- a type of properties mainly concerned in reactive systems.
  To rightly capture the synchronous communications, we define a constructive semantics for the program model of SDL.
  We build a sound and relatively complete proof system for SDL.
  Compared to previous verification approaches, SDL provides a divide and conquer way to analyze and verify synchronous models based on compositional reasoning of the syntactic structure of the programs of SDL.
  To illustrate the usefulness of SDL, we apply SDL to specify and verify a small example in the synchronous model SyncChart, which shows the potential of SDL to be used in practice.

\end{abstract}

\section{Introduction}
\label{section:Introduction}

Synchronous model is a type of formal models for modelling and specifying reactive systems~\cite{Harel85} --- a type of computer systems which maintain an on-going interaction with their environment at a speed determined by this environment.
The notion of synchronous model was firstly proposed in 1980's along with the invention and development of synchronous programming languages~\cite{Benveniste03}.
Among these synchronous programming languages,
the most famous and representative ones are Signal~\cite{Benveniste91}, Lustre~\cite{Halbwachs91} and Esterel~\cite{Berry92}, which have been widely used in academic and industrial communities (for instance, cf.~\cite{Espiau90,Berry00,Gamatie04,Qian15}).

\ifx
[
papers about application of synchronous languages in industries: 1. A synchronous approach for control sequencing in robotics applications. 2. a
formal method applied to avionic software development. 3. System level design and verification using a
synchronous language.
4. 12 years later?
5. Specifying a Cryptographical Protocol in Lustre and SCADE.
]
\fi

One crucial modelling paradigm taken in synchronous models is \emph{synchrony hypothesis}~\cite{Benveniste03},
which considers the behaviour of reactive systems as a sequence of \emph{reactions}.
At each reaction, all events of sensing from inputs, doing computations and actuating from outputs are considered to occur \emph{simultaneously}.
In such a model, time is interpreted as a discrete sequence of instants where at each instant, several events may occur in a logical order (in the sense that they still can have data dependencies).
The physical time of events no longer plays a crucial role, but the simpler \emph{logical time}, i.e., the instant at which the event is executed. 
Since events supposed to be nondeterministic in a system are now all assumed to be simultaneous,
in a synchronous model, a concurrent behaviour is \emph{deterministic}. 
Determinism has proved to be of great advantage not only in system testing (as mentioned in~\cite{Lee15,Lohstroh19}), but also in system verifications like model checking and theorem proving because it can largely decrease the size of the model of a concurrent system.

\ifx
[
program verification: while model checking is commonly taken as a main technology in verification of synchronous models, theorem-proving-based program verification has its own advantages over model checking.

Dynamic logic: A good formalism for capturing both system models and properties and recently has been used to specifying and verifying hybrid systems [Platzer's papers???].
]
\fi

In the verification of synchronous models, model checking has been adopted as the dominant technique (cf. \cite{Halbwachs93b,Pilaud88,Jagadeesan95,Andre09,Berry85,Halbwachs91}).
Although it has advantages like decidability and high efficiency, it is well-known that model checking suffers from the so-called state-space explosion problem.
Moreover, since model checking is based on the enumeration of states, it is hard for it to support \emph{modular verification} --- a ``divide and conquer'' way which proves to be important in verifying large-scale systems.

Theorem-proving-based program verification~\cite{ProgramVerifications09}, however, is a complement to model checking which can remedy these two disadvantages mentioned above.
It is not based on the enumeration of states, but the decomposition of programs according to their structures.
Dynamic logic~\cite{Pratt76} is a mathematical logic that supports theorem-proving-based program verification.
Unlike the three-triple forms $\{\phi\}p\{\psi\}$ in Hoare logic~\cite{Hoare63}, it integrates a program model $p$
directly into a logical formula so as to define a \emph{dynamic logical formula} in a form like $\phi\to [p]\psi$ (meaning the same thing as the triple $\{\phi\}p\{\psi\}$).
Such an extension on formulas allows to express the negations of program properties, e.g. $\neg [p]\psi$, which makes dynamic logic more expressive than Hoare logic.
Dynamic logical formulas can well support compositional reasoning on programs.
For example, by applying a compositional rule:
$$
\infer[]{
[p\cup q]\psi
}
{
[p]\psi \wedge [q]\psi
}
,$$ the proof of the property $[p\cup q]\psi$ of a choice program $p\cup q$, can be discomposed into the proof of the property $[p]\psi$ of the subprogram $p$ and
the proof of the property $[q]\psi$ of the subprogram $q$.

Dynamic logic has proven to be a strong method for verifying programs in a theorem-proving approach and has been applied or extended to be applied in
different types of programs and systems (e.g.~\cite{Feldman84,Rustan07,Platzer07b,Platzer07}).
In this paper, we introduce dynamic logic in the verification of reactive systems.
Specifically, we initially use dynamic logic for specifying and verifying the synchronous models of reactive systems.
However, directly using the traditional dynamic logics for our goal would yield several problems.
Firstly,  the program models of the traditional dynamic logics, called \emph{regular programs}~\cite{Harel00}, are too general to express the execution mechanism in synchronous models, i.e., the simultaneous executions of events in one reaction.
Of course, these features of synchronous models can be encoded as regular programs, but the information will be lost during the proving process when the program structures are broken down.
Secondly, regular programs only support modelling sequential programs.
It cannot specify the communications in synchronous models, which is a main characteristic of reactive systems.
Thirdly, the traditional dynamic logics can only specify state properties, i.e., properties held after the terminations of a program.
In reactive systems, however, since a system is usually non-terminating, people care more about temporal properties (especially safety properties~\cite{Halbwachs91}), i.e., properties concerning each reaction of a system.

\ifx
[
when using rogram verification in synchronous models, we face two problems:
The hard point: 1. synchronous programming languages are usually not modular for theorem proving because they contain preemptive and signal-triggered sentences;
2. In reactive systems, since the system is always non-terminated, people always care about its temporal properties (especially safety properties [???]), rather than its state properties after the termination of the system.
]

[
our approach:
our work is to remedy this short comming of program verification approach in synchronous models.
We propose a verification framework....
SDL provides a framework for verifying synchronous models in a theorem-proving approach.
1. how do we solve the hard points (discussed above);
]
\fi

In this paper, to solve these problems, we propose a variation of dynamic logic, called \emph{synchronous dynamic logic} (SDL), for the verification of synchronous models based on theorem proving.
SDL extends first-order dynamic logic (FODL)~\cite{Harel79} in two aspects: (1) the program model of SDL, called \emph{synchronous programs} (SPs), extends regular programs with the notions of reactions and signals in order to model the execution mechanism of synchronous models, and with the parallel operator to model the communications in synchronous models; (2) SDL formulas extends FODL formulas with a type of formulas from~\cite{Platzer07} called \emph{temporal dynamic logical formulas} in order to capture safety properties in reactive systems.
We build a sound and relatively complete~\cite{Cook78} proof system for SDL, which extends the proof system of FODL with a set of rules for the new primitives introduced in SPs, and with a set of rewrite rules for the parallel programs in SPs.
We illustrate how SDL can be used to specify and verify synchronous models by simple examples, where we encode one of popular synchronous models, syncCharts~\cite{Andre03}, as SPs and specify their properties as SDL formulas.
We also display how to prove a simple SDL formula in the proof system of SDL.

Our proposed SP model of SDL follows an SCCS~\cite{Milner83} style in its syntax. But instead of considering all events executed at an instant as non-ordered elements,
in SPs we also consider the logical before-and-after order between these events, as also considered in synchronous models like~\cite{Lohstroh19}.
SP model is not an imperative synchronous programming language, like Esterel and Quartz~\cite{Schneider17}, which were aimed at implementing reactive systems in a convenient and hierarchy way.
These languages contain more advanced features such as preemption and trap-exit statement~\cite{Berry92}, in order to make the programming of some high-level behaviour of circuits easier.
However, these features are not the essences of synchronous models and are redundant in terms of expressive power.
Our SPs do not provide a support of them, because they cannot be verified in a compositional way in a theorem-proving approach.
Besides, the goal of SDL is not to verify actual synchronous programs at industrial level in reality, but
to provide a verification framework for general synchronous models at early stages of design of a reactive system.
In our opinion, SDL also provides a theoretical foundation for building more complex logics with richer syntax and semantics for verifying synchronous programming languages in practice.

The rest of this paper is organized as follows.
In Sect.~\ref{section:First-order Dynamic Logic}, we give a brief introduction to FODL.
We propose SDL and build a proof system for SDL in Sect.~\ref{section:Synchronous Dynamic Logic} and ~\ref{section:Proof System of SDL} respectively, and
in Sect.~\ref{section:Soundness and Relative Completeness of SDL Calculus}, we analyze the soundness and relative completeness of SDL calculus.
In Sect.~\ref{section:SDL in Specification and Verification of SyncCharts}, we show how SDL can be used in specifying and verifying synchronous models through an example.
Sect.~\ref{section:Related Work} introduces related work, while a conclusion is made in Sect.~\ref{section:Conclusion and Future Work}.

\section{First-order Dynamic Logic}
\label{section:First-order Dynamic Logic}
\ifx
[
Give an intuitive explanation of the semantics of FODL but not the formal definition.
]
\fi

The syntax of FODL is as follows:
$$
\begin{aligned}
p\ddef &\ \psi?\ |\ x:=e\ |\ p\seq p\ |\ p\cup p\ |\ p^*, \\
\phi\ddef &\ \true\ |\ a\ |\ [p]\phi\ |\ \neg \phi\ |\ \phi\wedge\phi\ |\ \forall x.\phi.
\end{aligned}
$$
It consists of two parts, a program model called \emph{regular programs} $p$ and a set of logical formulas $\phi$.

$\psi?$ is a \emph{test}, meaning that at current state the proposition $\psi$ is true. 
$x:=e$ is an \emph{assignment}, meaning that assigning the value of the expression $e$ to the variable $x$. 
The syntax of $\psi$ and $e$ depends on the discussed domain. For example, $e$ could be an arithmetic expression and $\psi$ could be a quantifier-free first-order logical formula defined in Def.~\ref{def:Terms}, \ref{def:First-order Logical Formulas}.
$\seq$ is the \emph{sequence operator}, $p\seq q$ means that the program $p$ is first executed, after it terminates, the program $q$ is executed.
$\cup$ is the \emph{choice operator}, $p\cup q$ means that either the program $p$ or $q$ is executed.
$*$ is the \emph{star operator}, $p^*$ means that $p$ is nondeterministically executed  for a finite number of times.

$\true$ is Boolean true. (We also use $\false$ to represent Boolean false. )
$a$ is an atomic formula whose definition depends on the discussed domain.
For example, it could be the term $\theta$ defined in Def.~\ref{def:First-order Logical Formulas}.
Formulas $[p]\phi$ are called \emph{dynamic formulas}.
$[p]\phi$ means that after all executions of $p$, formula $\phi$ holds.

The semantics of FODL is based on Kripke frames~\cite{Harel00}.
A Kripke frame is a pair $(S, \val)$ where $S$ is a set of states, also called \emph{worlds}, and $\val$ is an interpretation, also called a \emph{valuation}, whose definition depends on the type of logic being discussed.
In FODL, $\val$ interprets a program as a set of pairs $(s, s')$ of states and interprets a formula as a set of states.
Intuitively, each pair $(s, s')\in \val(p)$ means that starting from the state $s$, the program $p$ is executed and terminates at the state $s'$.
For each state $s$, $s\in \val(\phi)$ means that $s$ satisfies the formula $\phi$; $s\in \val([p]\phi)$ means that for all pairs $(s,s')\in \val(p)$, $s'\in \val(\phi)$.
For a formal definition of the semantics of FODL, one can refer to~\cite{Harel00}.

The proof system of FODL is sound and relatively complete. Except for the the assignment $x:=e$ and test $\psi?$, the rules for other operators and logical connectives are also defined as a part of the proof system of SDL below in Table~\ref{table:Rules for Closed SPs} and \ref{table:Rules of FOL}. Refer to~\cite{Harel00} for more details.

\section{Synchronous Dynamic Logic}
\label{section:Synchronous Dynamic Logic}
In this section, we propose the synchronous dynamic logic (SDL).
SDL extends FODL to provide a support for specification and verification of synchronous models of reactive systems.
In Sect.~\ref{section:Syntax of SDL}, we firstly define the syntax of SDL.
Then in Sect.~\ref{section:Semantics of SDL}, we give the semantics of SDL.

\subsection{Syntax of SDL}
\label{section:Syntax of SDL}

SDL consists of a program model defined in Def.~\ref{definition: Synchronous Programs} and a set of logical formulas defined in Def.~\ref{def:sDTL Formulas}.

Before defining the program model of SDL, we first introduce the concepts of terms and first-order logical formulas.

\begin{mydef}[Terms]
\label{def:Terms}
	The syntax of a term $e$ is an arithmetical expression given as the following BNF form:
	$$\begin{aligned}
		e\ddef &\ x\ |\ c\ |\ f(e, e),
	\end{aligned}$$
	where $x\in \Var$ is a variable, $c\in \mbb{Z}$ is a constant, $f\in \{+, -, \cdot, /\}$ is a function.
\end{mydef}
The symbols $+,-,\cdot, /$ represent the usual binary functions in arithmetic theory: addition, substraction, multiplication and division respectively.

\begin{mydef}[First-order Logical Formulas]
\label{def:First-order Logical Formulas}
	The syntax of an arithmetical first-order logical formula $\psi$ is defined as follows:
	$$\begin{aligned}
	\psi \ddef &\ \true\ |\ \theta(e, e)\ |\ \neg \psi\ |\ \psi\wedge \psi\ |\ \forall x. \psi,
	\end{aligned}$$
where 
$\theta\in \{<, \le, =, >, \ge\}$ is a relation.
\end{mydef}
The symbols $<, \le, =, >, \ge$ represent the usual binary relations in arithmetic theory, for example, $<$ is the ``less-than'' relation, $=$ is equality, and so on.
As usual, the logical formulas with other connectives, such as $\vee, \exists, \to$, can be expressed by the formulas given above.

The formulas defined in Def.~\ref{def:First-order Logical Formulas} contains the terms and relations interpreted as in Peano arithmetic theory, so
in this paper, we also call them \emph{arithmetic first-order logical (AFOL) formulas}.

The program model of SDL, called \emph{synchronous program} (SP), extends the regular program of FODL with primitives to support the modelling paradigm of synchronous models.
In order to capture the notion of reactions, we introduce the \emph{macro event} in an SP to collect all events executed at the same instant.
In order to describe the communications in synchronous models, we introduce the  \emph{signals} and \emph{signal conditions} as in synchronous programming languages like Esterel~\cite{Berry92} to send and receive messages between SPs,
and we introduce a parallel operator to express that several SPs are executed concurrently.

\begin{mydef}[Synchronous Programs]
\label{definition: Synchronous Programs}
The syntax of a synchronous program $p$ is given as the following BNF form:
$$\begin{aligned}
p\ddef &\ \noth\ |\ \halt\ |\ \alpha\ |\ p\seq p\ |\ p \cho p\ |\ p^*\ |\ \para(p,...,p),
\end{aligned}$$
where $\alpha$ is defined as:
$$\begin{aligned}
\alpha \ddef&\ \epsilon\ |\ \act\nex \alpha,\\
\act \ddef &\ \psi?\ |\ \varrho?\ |\ \sig!e\ |\ x:=e, \\
\varrho\ddef&\ \hat{\sig}(x)\ |\ \bar{\sig}.
\end{aligned}$$
\end{mydef}

The set of all synchronous programs is denoted by $\SP$.
We call $\noth, \halt$ and $\alpha$ \emph{atomic programs},
and call the programs of other forms \emph{composite programs}.
We denote the set of all atomic programs as $\ASP$.

$\noth$ is a program called \emph{nothing}.
As implied by its name, the program neither does anything nor consumes time.
The role it plays is similar to the statement ``nothing'' in Esterel~\cite{Berry92}.

$\halt$ is a program called \emph{halting}. It causes a deadlock of the program.
It never proceeds and the program halts forever.
The program $\halt$ is similar to the statement ``halt'' in Esterel.

The program $\alpha$ is called a \emph{macro event}.
It is the collection of all events executed at the current instant in an SP.
A macro event consists of a sequence of events linked by dots $\nex$ one by one and it always ends up with a special event $\epsilon$ called \emph{skip}.
The event $\epsilon$ skips the current instant and forces the program move to the next instant.
$\epsilon$ is the only term that consumes time in SPs.
It plays a similar role as the statement ``skip'' in Esterel.
An event $\act$, sometimes also called a \emph{micro event}, can be either a test $\psi?$, a \emph{signal test} $\varrho?$, a \emph{signal emission} $\sig!e$ or an assignment $x := e$.
The tests $\psi?$ and assignments $x:=e$ have the same meanings as in FODL.
A signal test $\varrho?$ checks the signal condition $\varrho$ at the current instant.
$\varrho$ has two forms. $\hat{\sig}(x)$ means that the signal $\sig$ is emitted at the current instant. $x$ is a variable used for storing the value of $\sig$.
$\bar{\sig}$ means that the signal $\sig$ is absent at the current instant.
A signal emission $\sig!e$ emits a signal $\sig$ with a value expressed as a term $e$ at the current instant.
We stipulate that
a signal $\sig$ can emit with no values, and we call it a \emph{pure signal}, denoted by $\sig$.
Sometimes we also simply write $\sig!e$ as $\sig$ when the value $e$ can be neglected in the context.

The sequence programs $p\seq q$, choice programs $p\cup q$ and the star programs (or called \emph{finite loop programs}) $p^*$ have the same meanings as in FODL.
$\para$ is the parallel operator.
$\para(p_1,...,p_n)$ means that the programs $p_1,...,p_n$ are executed concurrently.
When $n=2$, we also write $\para(p_1,p_2)$ as $p_1\para p_2$.
We often call an SP without a parallel operator $\para$ a \emph{sequential program}, and call an SP that is not a sequential program a \emph{parallel program}.

In a parallel program $\para(p_1,...,p_n)$, at each instant, events from different programs $p_1,...,p_n$ are executed simultaneously based on the same environment, while events in each program $p_i$ are executed simultaneously, but in a sequence order.
SP model follows the style of SCCS~\cite{Milner83} in its syntax, but differs in the treatments to the execution order of events at one instant.
In SCCS, all events at one instant are considered as disordered elements that are commutative.

\ifx
The programs defined in Def.~\ref{definition: Synchronous Programs} is a synchronous model that follows the style of SCCS by Milner~\cite{Milner83},
but differs in the treatments to the order of events executed at one instant.
In SCCS, all events executed at one instant are linked by an commutative operator $\times$ so they are disordered.
For example, $a_1\times a_2$ and $a_2\times a_1$ are considered as the same behaviour in SCCS.
However, in SPs, events executed at the same instant still have logical orders, which means that $a_1\nex a_2$ and $a_2\nex a_1$ might be two different behaviours in SPs.
As will be seen later in this paper, logical order plays a crucial role in specifying behaviours in reactive systems.
\fi


A closed SP is a program that does not interfere with its environment.
In this paper, as the first step to propose a dynamic logic for synchronous models, we only focus on building a proof system for closed SPs.
This makes it easier for us to define the semantics for SPs because there is no need to consider the semantics for signals.
Similar approach was taken in~\cite{Peleg87} when defining a concurrent propositional dynamic logic.

\ifx
In this paper, we will only focus on building a proof system for closed systems, i.e., systems that do not interfere with their environment.
We take signals in SPs just as the channels for communications between parallel SPs so that we can avoid giving the semantics for signals.
Therefore, we need to introduce the notion of \emph{closed SPs},

As the approach taken in~\cite{concurrent dynamic logic}, in this paper, we take signals just as channels for communications, so a signal
\fi
\ifx
Our first version of SDL will only focus on the derivation of closed programs.
A closed program is a program that does not interfere with its environment.
Therefore there are not any sentences concerning signals and signal tests in the program.
\fi
\begin{mydef}[Closed SPs]
\label{def:Closed SPs}
	Closed SPs are a subset of SPs, denoted as $\Cl(\SP)$, where
	a closed SP $q\in \Cl(\SP)$ is defined by the following grammar:
	$$\begin{aligned}
	q\ddef \noth\ |\ \halt\ |\ \alpha'\ |\ q\seq q\ |\ q\cho q\ |\ q^*\ |\ \para(p,...,p),
	\end{aligned}$$
where $p\in \SP$ is an SP, $\alpha'$ is defined as:
	$$\begin{aligned}
\alpha' \ddef&\ \epsilon\ |\ \act'\nex \alpha',\\
\act' \ddef &\ \psi?\ |\ x:=e.
\end{aligned}$$
\end{mydef}

As indicated in Def.~\ref{def:Closed SPs}, a closed SP is an SP in which signals and signal tests can only appear in a parallel program.

We often call a macro event $\alpha'$ a \emph{closed macro event}, call an event $\act'$ a \emph{closed event}.
We call an SP which is not closed an \emph{open SP}.
We denote the set of all closed atomic programs as $\Cl(\ASP)$.
For convenience, in the rest of the paper, we also use $p$ to represent a closed SP and use $\alpha$, $\act$ to represent a closed macro event and a closed event respectively.

SDL formulas extend FODL formulas with temporal dynamic logical formulas of the form $[p]\Box\phi$ from~\cite{Platzer07} to capture safety properties in reactive systems.
The syntax of SDL formulas is given as follows.

\begin{mydef}[SDL Formulas]
    \label{def:sDTL Formulas}
	The syntax of an SDL formula $\phi$ is defined as follows:
$$\begin{aligned}
	\phi \ddef &\ \true\ |\ \theta(e,e)\ |\ \neg \phi\ |\ \phi\wedge \phi\ |\ \forall x.\phi\ |\ [p]\phi\ |\ [p]\Box\phi,
\end{aligned}$$
where $\theta\in \{<, \le, =, >, \ge\}$, $p\in \Cl(\SP)$.
\end{mydef}

$\phi$ is often called a \emph{state formula}, since its semantics concerns a set of states.
A property described by a state formula, hence, is called a \emph{state property}.

The formula $[p]\Box\phi$ captures that the property $\phi$ holds at each reaction of the program $p$.
Intuitively, it means that all execution traces of $p$ satisfy the temporal formula $\Box\phi$, which means that
all states of a trace satisfy $\phi$.
The dual formula of $[p]\Box\phi$ is written as $\la p\ra \Diamond\phi$, we have $\neg [p]\Box\phi = \la p\ra\Diamond\neg \phi$.
Intuitively, $\la p\ra \Diamond\phi$ means that there exists a trace of $p$ that satisfies the temporal formula $\Diamond \phi$, which means that there exists a state of a trace that satisfies $\phi$.

We assign precedence to the operators in SDL: the unary operators $*$, $\neg$, $\forall x$ and $[p]$ bind tighter than binary ones.
$\seq$ binds tighter than $\cup$.
For example, the formula $\neg \psi\wedge [\alpha\seq p^*\cup q]\phi\wedge \psi'$ should be read as $(\neg \psi)\wedge ([(\alpha\seq (p^*))\cup q]\phi)\wedge \psi'$.

\begin{mydef}[Bound Variables in SPs]
In SDL, a variable $x$ is called a bound variable if it is bound by a quantifier $\forall x$, an assignment $x := e$ or a signal test $\hat{\sig}(x)?$,
otherwise it is called a free variable.
\end{mydef}

We use $\BV(p)$ and $\FV(p)$ (resp. $\BV(\phi)$ and $\BV(\phi)$) to represent the sets of bound and free variables of an SP $p$ (resp. a formula $\phi$) respectively.

We introduce the notion of substitution in SDL.
A substitution $\phi[e/x]$ replaces each free occurrence of the variable $x$ in $\phi$ with the expression $e$.
A substitution $\phi[e/x]$ is \emph{admissible} if there is no variable $y$ that occurs freely in $e$ but is bound in $\phi[e/x]$.
In this paper, we always guarantee admissible substitutions by bound variables renaming mentioned above.

In this paper, we always assume \emph{bound variables renaming} (also known as $\alpha-$ conversion) for renaming bound variables when needed to guarantee admissible substitutions.
For example, for a substitution $([z:=x-1\nex \epsilon]x > z)[x/z]$, it equals to $[z:=y-1\nex \epsilon]y > x$ by renaming $x$ with a new variable $y$.

In SPs, we stipulate that signals are \emph{the only way} for communication between programs.
Therefore, any two programs running in parallel cannot communicate with each other through reading/writing a variable in $\Var$.
To meet this requirement, 
we put a restriction on the bound variables of SPs that are running in parallel.

\begin{mydef}[Restriction on Parallel SPs]
	\label{def:Restriction on Parallel SPs}
	In SPs, any parallel program $\para(p_1,...,p_n)$ satisfies that
	$$\BV(p_i)\cap \FV(p_j) = \BV(p_j)\cap \FV(p_i) = \emptyset $$
for any $i, j$, $1\le i<j\le n$.
\end{mydef}

We assume bound variables renaming for renaming bound variables when needed to solve the conflicts.
For example, given two programs $p = (x := 1 \seq v := x + 2)$ and $q = (\epsilon\seq y := x + 1)$, we have $p\para q = (z := 1\seq v := z + 2)\para (\epsilon\seq y := x + 1)$ by renaming the bound variable $x$ of $p$ with a new variable $z$.


\subsection{Semantics of SDL}
\label{section:Semantics of SDL}
Because of the introduction of temporal dynamic logical formulas $[p]\Box\phi$ in SDL, it is not enough to define the semantics of SPs as a pair of states like the semantics of the regular program of FODL~\cite{Harel00}, because except
recording the states where a program starts and terminates, we also need to record all states during the execution of the program.

Before defining the semantics of SDL, we first introduce the notions of states and traces.

\begin{mydef}[States]
	A state $s : \Var\to \mbb{Z}$ is a mapping from the set of variables $\Var$ to the domain of integers $\mbb{Z}$.
\end{mydef}

We denote the set of all states as $\bff{S}$.

\begin{mydef}[Traces]
	A trace $tr$ is a finite sequence of states: $$s_1s_2...s_n, \mbox{ where }n\ge 1.$$
	$n$ is the length of the trace $tr$, denoted by $l(tr)$.
\ifx
	The length of a finite trace $tr$, denoted by $l(tr)$, is defined as the length of its sequence of states. Formally,
	\begin{enumerate}
		\item $l(s) \dddef 1$ for any state $s$;
		\item $l(tr)\dddef l(tr^b) + l(tr^2)$ for any finite trace $tr$.
	\end{enumerate}
\fi
\end{mydef}

	We use $tr(i)$ ($i\ge 1$) to denote the $i$th element of $tr$.
	We use $tr^i$ ($i\ge 1$) to denote the prefix of $tr$ starting from the $i$th element of $tr$, i.e.,
	$tr^i \dddef tr(i)tr(i+1)...tr(n)$.
	We use $tr^b$ to denote the first element of $tr$, so $tr^b = tr(1)$; we use
	$tr^e$ to denote the last element of $tr$, so $tr^e = tr(l(tr))$.

\begin{mydef}[Concatenation between Traces]
	\label{def:Concatenation between Traces}
	\ifx
	The fusion of a finite trace $tr_1 = s_1s_2...s_n$ and a trace $tr_2 = u_1u_2...u_m...$, denoted by $tr_1\fus tr_2$, is defined as
	$$tr_1\fus tr_2 \dddef s_1s_2...s_nu_2...u_m...,$$
	provided that $s_n = u_1$.
	\fi
	
	Given two traces $tr_1 = s_1...s_n$ and $tr_2 = u_1u_2...u_m$ ($n\ge 1$, $m\ge 1$),
	the concatenation of $tr_1$ and $tr_2$, denoted by $tr_1\circ tr_2$, is defined as:
	$$tr_1\circ tr_2\dddef
		s_1...s_nu_2...u_m,\mbox{ provided that $s_n = u_1$}.$$
	
	The concatenation operator $\circ$ can be lifted to an operator between sets of traces.
	Given two sets of traces $T_1$ and $T_2$, $T_1\circ T_2$ is defined as:
	$$T_1\circ T_2\dddef \{tr_1\circ tr_2\ |\ tr_1\in T_1, tr_2\in T_2\}.$$
\end{mydef}

Note that if $tr^e_1\neq tr^b_{2}$, $tr_1\circ tr_2$ is undefinable. Therefore, $\circ$ is a partial function.

\ifx
\begin{mydef}[Infinite Concatenation between Traces]
	The fusion of an infinite (but countable) number of finite traces
	$$\begin{aligned}&tr_1 = s_{1,1}s_{1,2}...s_{1,k_1},\\
					&tr_2 = s_{2,1}s_{2,2}...s_{2,k_2},\\
					&...,\\
					&tr_n = s_{n,1}s_{n,2}...s_{n,k_n},\\
					&...,
	\end{aligned}$$
	denoted by $tr_1\fus tr_2\fus...\fus tr_n\fus...$, is defined as
	$$tr_1\fus tr_2\fus...\fus tr_n\fus...\dddef \underbrace{s_{1,1}s_{1,2}...s_{1,k_1}}_{tr_1}\underbrace{s_{2,2}...s_{2,k_2}}_{tr_2}...\underbrace{s_{n,2}...s_{n,k_n}}_{tr_n}...,$$
	provided that $s_{1,k_1} = s_{2,1}$, $s_{2,k_2} = s_{3,1}$,...,$s_{n,k_n} = s_{n+1,1}$,....
	
	The concatenation of an infinite (but countable) number of traces $tr_1$, $tr_2$,..., $tr_n$,..., denoted by
	$tr_1\circ tr_2\circ ...\circ tr_n\circ ...$, is defined as:
	$$tr_1\circ tr_2\circ ...\circ tr_n\circ...\dddef \left\{
		\begin{array}{ll}
		tr_1\fus tr_2\fus...\fus tr_n \fus ..., &\mbox{if all traces $tr_i$ ($i\ge 1$) are finite}\\
		tr_1\fus tr_2\fus...\fus tr_k,  &\mbox{if $tr_k$ ($k\ge 1$) is infinite, and $tr_1, ..., tr_{k-1}$ are finite}
		\end{array}
	\right..$$
	
	The infinite concatenation between traces can be lifted to the concatenation between an infinite (but countable) number of sets of traces $T_1$, $T_2$, ..., $T_n$,..., denoted as
	$T_1\circ T_2\circ ...\circ T_n\circ ...$, which is defined as:
	$$T_1\circ T_2\circ ...\circ T_n\circ ...\dddef \{tr_1\circ tr_2\circ ...tr_n\circ...\ |\ tr_1\in T_1, tr_2\in T_2,..., tr_n\in T_n, ...\}.$$
\end{mydef}
\fi

The semantics of SDL formulas is defined as a Kripke frame $(\bff{S}, \val)$ where the interpretation $\val$ maps each SP to a set of traces on $\bff{S}$ and each SDL formula to a set of states in $2^\bff{S}$.
In the following subsections, we define the valuation $\val(\phi)$ of an SDL formula $\phi$ step by step, by simultaneous induction in Def.~\ref{def:Valuation of sDTL Formulas}, Def.~\ref{def:Valuation of Closed SPs} and Def.~\ref{def:Par}.

In Sect.~\ref{section:Valuations of Terms and Closed Programs}, we first define the semantics of closed SPs in Def.~\ref{section:Valuation of Parallel Programs}, in which
we leave the detailed definition of the semantics of parallel programs in Def.~\ref{def:Par} of Sect.~\ref{def:Valuation of sDTL Formulas}.
With the semantics of closed SPs, we define the semantics of SDL formulas in Def.~\ref{def:Valuation of sDTL Formulas} of Sect.~\ref{section:Valuation of Formulas}.


\subsubsection{Valuations of Terms and Closed Programs}
\label{section:Valuations of Terms and Closed Programs}
In this subsection, we define the valuations of terms  and closed SPs in SDL.

\begin{mydef}[Valuation of Terms]
	\label{def:Valuation}
	The valuation of terms of SDL under a state $s$, denoted by $\val_s$, is defined as follows:
	\begin{enumerate}
		\item $\val_s(x) \dddef s(x)$;
		\item $\val_s(c) \dddef c$;
		\item $\val_s(f(e_1, e_2))\dddef f(\val_s(e_1), \val_s(e_2))$, where $f\in \{+, -, \cdot, /\}$;
	\end{enumerate}
\end{mydef}
Note that we assume that $+, -, \cdot, /$ are interpreted as their normal meanings in arithmetic theory.

The semantics of closed SPs is defined as a set of traces in the following definition.

\begin{mydef}[Valuation of Closed SPs]
	\label{def:Valuation of Closed SPs}
	The valuation of closed SPs is defined inductively based on the syntactic structure of SPs as follows:
	\begin{enumerate}
		\item $\val(\noth)\dddef \bff{S}$;
		\item $\val(\halt)\dddef \emptyset$;
		\item $\val(\alpha)\dddef \{ss'\ |\ tr\in \mval(\alpha), tr^b = s, tr^e = s'\}$, where $\mval(\alpha)$ is defined as:
		\begin{enumerate}[(i)]
		    \item $\mval(\epsilon)\dddef \{ss\ |\ s\in \bff{S}\}$;
		    \item $\mval(\psi?\nex \alpha')\dddef \{ss\ |\ s\in \val(\psi)\}\circ \mval(\alpha')$, where $\val(\psi)$ is defined in Def.~\ref{def:Valuation of sDTL Formulas};
		    \item $\mval(x:=e\nex \alpha')\dddef \{ss'\ |\ s' = s[x\mapsto \val_s(e)]\}\circ \mval(\alpha')$, where $s[y\mapsto v]$ is defined as
            $$
            \begin{aligned}
            s[y\mapsto v](z)\dddef \left\{
            \begin{array}{ll}
            v & \mbox{if $z = y$}\\
            s(z) & \mbox{otherwise}
            \end{array}
            \right\};
            \end{aligned}
            $$
		\end{enumerate}
		\ifx
		\item $\val(\alpha)\dddef \left\{ss'\ \left|\ \begin{aligned}&s, s'\in \bff{S}; \\
					&\mbox{for each } (x:=e)\in \alpha, s'(x)=\val_s(e);\\
					&\mbox{for all }y\in \Var\mbox{ but }y\notin \dVar(\alpha),  s'(y)=s(y)\end{aligned}\right.\right\}$;
		\fi
		\item $\val(p\seq q)\dddef \val(p)\circ \val(q)$;
		\item $\val(p\cho q)\dddef \val(p)\cup \val(q)$;
		\item $\val(p^*) \dddef \bigcup^{\infty}_{n=0}\val^n(p)$, where $\val^n(q) \dddef \underbrace{\val(q)\circ ...\circ\val(q)}_{n}$, $\val^0(q)\dddef \bff{S}$
	for any $q\in \SP$;	
		
		\item $\val(\para(p_1,...,p_n))\dddef \Par(\para(p_1,...,p_n))$, where $\Par(\para(p_1,...,p_n))$ is defined in Def.~\ref{def:Par}.
	\end{enumerate}
	
\end{mydef}

The semantics of $\noth$ is defined as the set of all traces of length $1$.
For any set $A$, we have $\bff{S}\circ A = A\circ \bff{S} = A$, which meets the intuition that $\noth$ does nothing and consumes no time.
The definitions of the semantics of the programs $\halt$, $p;q$, $p\cup q$ and $p^*$ are easy to understand.

The key point of the definition lies in the definitions of the semantics of the macro event $\alpha$ and the parallel program $\para(p_1,...,p_n)$.
The semantics of $\para(p_1,...,p_n)$ is given later in Sect.~\ref{section:Valuation of Parallel Programs}.

The semantics of $\alpha$ reflects the time model of synchronous models.
It consists of a set of pairs of states, we call them \emph{macro steps}.
Each macro step starts at the beginning of the macro event $\alpha$, and ends after the execution of the last event of $\alpha$ --- the skip $\epsilon$.
In each macro step $ss'$ of a macro event, the sequential executions of all events in the macro event form a trace $tr$ with $tr^b = s$ and $tr^e = s'$.
We call the execution of a micro event a \emph{micro step}.
The set of the execution traces of all micro events of $\alpha$ is defined by $\val_m(\alpha)$.
The definition of the semantics of $\epsilon$ means that from any state $s$, $\epsilon$ just skips the current instant and jumps to the same state without doing anything more.
By simultaneous induction, we can legally put the definition of $\val(\psi)$ in Def.~\ref{def:Valuation of sDTL Formulas} since $\psi$ is a pure first-order logical formula which does not contain any SPs.
The semantics of an assignment $x:=e$ is defined just as in FODL.

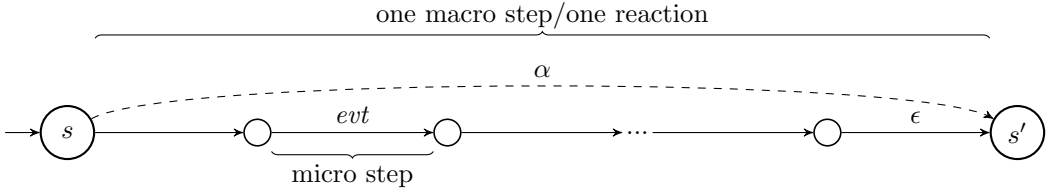
\begin{figure}[htpb]
	\centering
\scalebox{1}{
	\begin{tikzpicture}[->,>=stealth', node distance=3cm]
	\node[state4, initial, initial where = west, initial text= ] (s1) {$s$};
	\node[state3, xshift=2.5cm, yshift=0cm] at (s1) (t1) {};
    \node[state3, xshift=2.5cm] at (t1) (t2) {};
    \node[state3, xshift=2.5cm, draw=none] at (t2) (t3) {...};
	\node[state3, xshift=2.5cm] at (t3) (t4) {};
    \node[state4, xshift=2.5cm] at (t4) (s2) {$s'$};

	\path
    (s1) edge  (t1)
    (t1) edge node[above] {$\act$} (t2)
    (t2) edge  (t3)
    (t3) edge  (t4)
    (t4) edge node [above] {$\epsilon$} (s2)
    (t2) edge [snake=brace, raise snake=0.25cm, -] node [below=0.3cm] {micro step}  (t1)
    (s1) edge [snake=brace, raise snake=1.25cm, -] node[above=1.25cm] {one macro step/one reaction} (s2)
    (s1) edge [bend left, dashed, looseness=0.25] node [above] {$\alpha$} (s2)
	;
	\end{tikzpicture}
} 
\caption{Macro Steps vs. Micro Steps}
\label{figure:Macro Steps vs. Micro Steps}
\end{figure}

According to Def.~\ref{def:Valuation of Closed SPs}, it is easy to see that each trace represents one execution of an SP, and each transition between two states of a trace exactly captures the notion of one ``reaction'' in synchronous models.
During a reaction, all events of a macro event are executed in sequence in different micro steps.
Fig.~\ref{figure:Macro Steps vs. Micro Steps} gives an illustration.

We can call the valuation given in Def.~\ref{def:Valuation of Closed SPs} the \emph{macro-step valuation}.
On the contrary, we can define the \emph{micro-step valuation} for closed SPs, stated as the following definition.

\begin{mydef}[Micro-Step Valuation of Closed SPs]
    The micro-step valuation of a closed SP $p$, denoted by $\mval(p)$, is defined inductively based on the syntactic structure of $p$ as:
        $$\mbox{$\mval(p)\dddef \val(p)$ for all programs but a macro event}.$$
    For a macro event $\alpha$, $\mval(\alpha)$ is already defined in Def.~\ref{def:Valuation of Closed SPs}.
\end{mydef}

In the micro-step valuation of closed SPs, each transition between two states exactly captures a micro step.

\subsubsection{Valuation of Parallel Programs}
\label{section:Valuation of Parallel Programs}

\paragraph{Definition of the Function $\Par$}
In this subsection, we mainly define the semantics of parallel programs, which is $\Par(\para(p_1,...,p_n))$ in Def.~\ref{def:Valuation of Closed SPs}.
To this end, we need to define how the programs $p_1,...,p_n$ communicate with each other at one instant.
However, because of the choice program and the star program, an SP turns out to be non-deterministic, which makes the direct definition of the communication rather complicated.
To solve this problem, we adopt the idea proposed  in~\cite{Peleg87}.
We firstly split each program $p_i$ ($1\le i\le n$) into many deterministic programs $r_{i1}, r_{i2}, ...$, each of which captures one deterministic behaviour of $p_i$.
Then we define the communication of the parallel program $\para(r_{1k_1},...,r_{nk_n})$ (where $k_1,...,k_n\in \mbb{N}^+$) for each combination $r_{1k_1},...,r_{nk_n}$ of the deterministic programs of $p_1,..., p_n$.

We call a deterministic program a \emph{trec}, a name inherited from~\cite{Peleg87}\footnote{where trec means a ``tree computation''.}.

\begin{mydef}[Trecs]
	\label{def:Trecs}
	A trec is an SP whose syntax is given as follows:
    $$
    \begin{aligned}
    \trec\ddef&\ \halt\ |\ \noth\ |\ \strr\ |\ \parp, \\
    \parp\ddef&\ \para(\trec,...,\trec)\ |\ \parp\seq \trec\ |\ \trec\seq \parp,\\
    \strr\ddef&\ \alpha\ |\ \alpha\seq \strr.
    \end{aligned}
    $$
	
    We denote the set of all trecs as $\Trec$.
	
\end{mydef}

\ifx
Note that in Def.~\ref{def:Trecs}, trec $p^\irep$ is different from SP $p^\omega$ on the syntax level.
The former is not an SP, but a shorthand of an infinite string (, or we can call it an \emph{infinite SP}), while the latter is an SP.
\fi

A trec is an SP in which there is no choice and star programs. In a trec, a sequential program must be $\halt$, $\noth$ or in the form of a string $\strr$,
from which the macro event at the current instant is always visible; a parallel program can be combined with a trec by the sequence operator $\seq$.

Different from the trec defined in~\cite{Peleg87}, where when a parallel program $p\para q$ is combined with a program $r$ by the sequence operator $\seq$,
$r$ is always combined with $p$ and $q$ separately in their branches as: $(p\seq r) \para (q\seq r)$.
The trec defined in Def.~\ref{def:Trecs}, however, allows the form like $(p\para q)\seq r$.
Unlike the program model of concurrent propositional dynamic logic, in our proposed SP model, the behaviour of $(p\para q)\seq r$ is not equivalent to that of $(p\seq r)\para(q\seq r)$, because generally, a program $r$ is not equivalent to
a parallel program $r\para r$.
\ifx
Different from the trec defined in~\cite{Peleg87}, the behaviour of the trec defined in Def.~\ref{def:Trecs} is no longer a tree-like structure.
As an example, Fig.~\ref{???} shows the behaviour of the trec ???.

[???intuitive explanation of trecs]

A trec defined in Def.~\ref{def:Trecs} is still a tree-like structure, but is different from the structure defined in~\cite{???}, where a trec is constructed based on the equivalence of the behaviours of the programs $(p\para q)\seq r$ and $(p\seq r)\para(q\seq r)$. However, in SPs, the behaviour of $(p\para q)\seq r$ is not equivalent to that of $(p\seq r)\para(q\seq r)$, because generally, a program $r$ is not equivalent to
a parallel program $r\para r$.
\fi
For example, consider a simple SP $r = (x:=1 \seq x:= x+1)$, according to the syntax of SPs in Def.~\ref{definition: Synchronous Programs} and the restriction defined in Def.~\ref{def:Restriction on Parallel SPs}, we have $r\para r = (y:=1\seq y:=y+1)\para(x:=1\seq x:=x+1)$, whose behaviour is not equivalent to that of $r$.

Intuitively, an SP has a set of trecs, each of which exactly represents one deterministic behaviour of the SP.
The relation between SPs and their trecs is just like the relation between regular expressions and their languages.
For example, an SP $\sig_1\nex \epsilon\seq (\sig_2\nex \epsilon\cup \sig_3\nex \epsilon)$ has two trecs: $\sig_1\nex \epsilon\seq \sig_2\nex \epsilon$ and $\sig_1\nex \epsilon \seq \sig_3\nex \epsilon$.
Below, we introduce the notion of the trecs of an SP to capture its deterministic behaviours.

Before introducing the trecs of an SP, we need to introduce an operator which links two trecs with the sequence operator $\seq$ in a way so that the resulted program is still a trec.

\begin{mydef}[Operator $\link$]
    \label{def:Operator link}
    Given two trecs $p$ and $q$, the trec $p\link q$ is defined as:
    $$
    p\link q\dddef \left\{\begin{array}{ll}
    \halt & \mbox{if $p=\halt$ or $q=\halt$}\\
    q & \mbox{if $p = \noth$}\\
    p & \mbox{if $q=\noth$}\\
    \alpha\seq (p'\link q) & \mbox{if $p = \alpha\seq p'$}\\
    p\seq q & \mbox{otherwise}
    \end{array}
    \right.
    $$
\end{mydef}

Intuitively, $\link$ links $p$ and $q$ in a way that does not affect the meaning of $p\seq q$.
In the definition above, the first case assumes the equivalence of the behaviours between $\halt\seq q$ (or $p\seq \halt$) and $\halt$.
The second (resp. third) cases assumes the equivalence of the behaviours between $\noth\seq q$ and $q$ (resp. $p\seq \noth$ and $p$).
The fourth case assume the equivalence of the behaviours between $(\alpha \seq p')\seq q$ and $\alpha\seq (p'\seq q)$.
All these assumptions actually hold for closed SPs according to the definition of the operator $\circ$ given above.
Given $p, q$ are trecs, it is not hard to see that the resulted program $p\link q$ is a trec.

\begin{mydef}[Trecs of SPs]
	\label{def:Trecs of SPs}
	The set of trecs of an SP is defined inductively as follows:
	\begin{enumerate}
		\item $\tau(a) \dddef \{a\}$, where $a\in \ASP$ is an atomic program;
		\item $\tau(p\seq q)\dddef \{r_1\link r_2\ |\ r_1\in \tau(p), r_2\in \tau(q)\}$;
		\item $\tau(p\cho q)\dddef \tau(p)\cup \tau(q)$;
		\item $\tau(p^*)\dddef \bigcup^\infty_{n=0} \tau(p^n)$, where $p^0\dddef \noth$, $p^n\dddef \underbrace{p\seq p\seq ...\seq p}_{n}$ for $n\ge 1$;
		
		
		\item $\tau(\para(p_1,...,p_n))\dddef \{\para(r_1,...,r_n)\ |\ r_1\in \tau(p_1),..., r_n\in \tau(p_n)\}$.
	\end{enumerate}
\end{mydef}

\ifx 
\begin{mydef}[Current Programs of Trecs]
	\label{def:Current Events of a Trec}
	Given a trec $p\in \Trec$, function $\CEvt(p)$ is inductively defined as follows:
	
	\begin{enumerate}
		\item If $p=q\para r$, then
		$$\CEvt(p) \dddef (\CEvt(q)[1]\uplus \CEvt(r)[1], r'), $$
		where $r' = \left\{
		\begin{array}{ll}
		\CEvt(q)[2]\para \CEvt(r)[2], &\mbox{if $\CEvt(q)[2] \neq \none$ and $\CEvt(r)[2] \neq \none$}\\
		\CEvt(r)[2],  &\mbox{if $\CEvt(q)[2] = \none$}\\
		\CEvt(q)[2],  &\mbox{if $\CEvt(r)[2] = \none$}\\
		\none, &\mbox{otherwise}
		\end{array}
		\right\}$.
		
		\item If $p = a$, where $a\in \ASP$ is an atomic program, then
		$$\CEvt(p)\dddef (\{ p \}, \none), $$
		where $\none$ is a special symbol different from any trec.
	
		\item If $p=q\seq r$, then
		$$\CEvt(p)\dddef\left\{
		\begin{array}{ll}
		(\CEvt(q)[1];\CEvt(r)[1], \CEvt(r)[2]) &\mbox{if $\CEvt(q)[1]$}\\
		(\CEvt(q)[1], r'),  &\mbox{otherwise}
		\end{array}
		\right\},
		$$
		where $r' = \left\{
		\begin{array}{ll}
		\CEvt(q)[2]\seq r, &\mbox{if $\CEvt(q)[2] \neq \none$}\\
		r,  &\mbox{otherwise}
		\end{array}
		\right\}$.
		
		\item If $p = q^\omega$, then
		$$\CEvt(p)\dddef \left\{
		\begin{array}{ll}
		(\CEvt(q)[1], r), &\mbox{if $\OC(\CEvt(q)[1], \noth)$}\\
		(\{\noth\}, \none), &\mbox{otherwise}
		\end{array}
		\right\},
		$$
		where $r = \left\{
		\begin{array}{ll}
		q^\omega, &\mbox{if $\CEvt(q)[2] = \none$}\\
		\CEvt(q)[2]\seq q^\omega,  &\mbox{otherwise}
		\end{array}
		\right\}$.

	\end{enumerate}
	
\end{mydef}
\fi

With the notion of the trecs of an SP, as indicated at the beginning of Sect.~\ref{section:Valuation of Parallel Programs},
we give the definition of the function $\Par$ as follows.

\begin{mydef}[Function $\Par$]
	\label{def:Par}
	The function $\Par(\para(p_1,...,p_n))$ appeared in Def.~\ref{def:Valuation of Closed SPs} is defined as
	$$\Par(\para(p_1,...,p_n))\dddef  \bigcup_{r\in \tau(\para(p_1,...,p_n))} \valt(r), $$
where $\valt(r)$ is defined in Def.~\ref{def:Valuation of the Trecs}.
\end{mydef}

Def.~\ref{def:Par} says that the semantics of a parallel program $\para(p_1,...,p_n)$ is defined as the union of the semantics of all deterministic parallel programs of $\para(p_1,...,p_n)$.

Now it remains to give the detailed definition of $\valt(r)$ for a deterministic parallel program $r$ of the form $\para(p_1,...,p_n)$, where $p_1,...,p_n\in\Trec$ .
The central problem to define its semantics is to define how the programs $p_1,...,p_n$ communicate with each other at one instant.
Below we firstly consider the most interesting case when each program $p_i$ ($1\le i\le n$) is of the form $\alpha_i \seq q_i$, where the macro event $\alpha_i$ to be executed at the current instant is visible.
As will be shown later in Def.~\ref{def:Valuation of the Trecs}, other cases are easy to handle.

\paragraph{Definition of the Function $\Merge$}
In SPs, when $n$ programs $p_1,...,p_n$ are communicating at an instant, all events at the current instant are executed simultaneously. 
Signals emit their values by broadcasting. Once a signal is emitted, all programs observe its state immediately at the same instant.
In such a scenario, it is natural to think that during one reaction, all programs $p_1,...,p_n$ should hold a consistent view towards the state of a signal.
Such a consistency is stated as the following law, which is called \emph{logical coherence law} in Esterel~\cite{Berry99}.

\begin{enumerate}[]
\item Logical coherence law: At an instant, a signal has a unique state, i.e., it is either emitted or absent, and has a unique value if emitted.
\end{enumerate}

In Esterel, the programs satisfying this law are called \emph{logically correct}~\cite{Berry99}.
Despite the simultaneous execution of all events at an instant, in SPs, the events of each macro event are also considered to be executed in a logical order that preserves data dependencies.
In Esterel, the logically correct programs that follow the logical execution order in one reaction are called \emph{constructive}~\cite{Berry99}.
In synchronous models, it is important to ensure the constructiveness when defining the semantics of the communication behaviour.

It is possible for an SP to be non-constructive.
For example, let $p_1 = \bar{\sig_1}?\nex \sig_1!5\nex \epsilon$, $p_2 = \epsilon$, then the program $p_1\para p_2$ is logically incorrect because the signal $\sig_1$ can neither be emitted nor absent at the current instant.
If $\sig_1$ is emitted, $\bar{\sig_1}$ cannot be matched, so by the logical order $\sig_1!5$ is never executed.
If $\sig_1$ is absent, then $\bar{\sig_1}$ is matched but then $\sig_1!5$ is executed.

In this paper, we propose a similar approach as proposed in~\cite{Berry99} for defining a constructive semantics of the communication between SPs.
The communication process, defined as the function $\Merge$ in Def.~\ref{def:Function Merge} below, ensures the logical coherence law while preserving the logical execution order in each macro event at the current instant.
Our method is based on a recursive process of executing all events at the current instant step by step based on their micro steps.
At each recursion step, we analyze all events of different macro events at the current micro step and execute them according to different cases and based on the current information obtained about signals.
After the execution we update the information and use it for the next recursion step.
This process is continued until all events at the current instant are executed.
\ifx
To ensure the consistency law, we borrow the idea from the constructive semantics of Esterel~\cite{EsterelConstructiveBook}.
During recursive steps, we maintain some extra information about the signals (i.e., the parameters $\Must$, $\Cannot$ and $\SV$ in Def.~\ref{def:Function Merge} below)
so that the consistency law can be checked in the end. This will be fully introduced below.
\fi

\ifx
Before introducing the function $\Merge$, we need to define an operator $\append$ and an auxiliary functions $\Match$ in Def.~\ref{label:Operator append} and \ref{def:Function Match} respectively.

\begin{mydef}[Function $\Chk$]
	\label{def:Function Chk}
	Given a multi-set $X$ of programs of the form $X = (\varrho_1?\nex \beta_1\seq q_1\sep ...\sep \varrho_n?\nex \beta_n\seq q_n)$ and a multi-set $Y$ of signals,
    then the function $\Chk(X, Y)$ is defined as:
	$$\Chk(X, Y)\dddef (b, C, R, \SV), $$
	where
    $$
    \begin{aligned}
    b\dddef&\  \left\{
        \begin{array}{ll}
            \false & \mbox{if $\Match(\varrho_i, Y)=\false$ for all $1\le i\le n$}\\
            \true & \mbox{otherwise}
        \end{array}
    \right\},
    \\
	C\dddef&\ \{\sig\ |\ \varrho_i = \bar{\sig}, \Match(\varrho_i, Y) = \true, 1\le i\le n\},
	\\
    \SV \dddef&\ \{(\sig, A)\ |\ \Match(\varrho_i, Y)=\true, \varrho_i = \hat{\sig}(x), A = \{\sig!e\ |\ \sig!e\in Y\}, 1\le i\le n\},
    \end{aligned}
    $$
	$$
	R\dddef (r_1\sep...\sep r_n)\mbox{ with }r_i = \left\{\begin{array}{ll}
	    \varrho_i?\nex \beta_i\seq q_i &\mbox{if $\Match(\varrho_i, Y) = \false$}\\
	    (\beta_i\seq q_i)[\comb(\{e\}_{\sig!e\in Y})/x] &\mbox{if $\Match(\varrho_i, Y)=\true$ and $\varrho_i = \hat{\sig}(x)$}\\
	    \beta_i\seq q_i &\mbox{if $\Match(\varrho_i, Y)=\true$ and $\varrho_i = \bar{\sig}$}
	\end{array}\right\}
	$$
	for any $1\le i\le n$.
	\ifx
	\begin{enumerate}
		\item If there exists a guarded event $a\in X$ s.t. $\Match(\varrho_a, Y) = \none$, then
		$$\Chk(X, Y) \dddef \none;$$
		\item Otherwise 
		 $$\Chk(X, Y)\dddef (M, C, \bigcup_{a\in X}\Match(\varrho_a, Y)), $$
		 where
		 $$\begin{aligned}
		 M =&\ \Signal(X),\\
		 C =&\ \{\sig\ |\ \bar{\sig}\in \varrho_b\}.
		 \end{aligned}		
		 $$
		 \end{enumerate}
    \fi
		
\end{mydef}

The function $\Chk$ checks whether a multi-set of signals $Y$ matches the signal tests $\varrho_1,...,\varrho_n$ of the set of programs $X$,
and records the set of  resulted programs after the checking.
It returns a 4-tuple $(b, C, R, \SV)$, each of which is introduced as follows.

$b$ indicates whether $\Sigs$ can match one of the signal tests.
$C$ is the set of signals that cannot emit.
When $Y$ matches a signal test of the form $\bar{\sig}$, it means that during the current reaction $\sig$ cannot appear in the environment, or the consistency law will be violated.
So we add $\sig$ to the set $C$.
$\SV$ is a multi-set of pairs $(\sig, A)$ where $A$ is a multi-set of signals.
It records that when a signal test $\hat{\sig}(x)$ is checked, the set of the signals $\sig$ it can observe is $A$.

\ifx
$C$ is the set of signals that cannot emit obtained according to the signal tests $\varrho_1,...,\varrho_n$;
$R$ represents the resulted $n$ programs after the match, corresponding to the programs $\alpha_1\seq q_1,...,\alpha_n\seq q_n$;
$\SV'$ records the set of signals observed by each signal test of the form $\hat{\sig}(x)$.
\fi

$R$ is the multi-set of resulted programs corresponding to $X$ after the checking.
In $R$, each program $r_i$ depends on the signal test $\varrho_i$ and whether $Y$ can match $\varrho_i$.
If $Y$ does not match $\varrho_i$, $r_i$ keeps the same as in $X$.
If $Y$ matches $\varrho_i$ and $\varrho_i$ is of the form $\hat{\sig}(x)$, we substitute the variable $x$ in the program $\beta_i\seq q_i$ with a value computed by a combinational function $\comb$, which accepts the values of all signals with the same name that are emitted in $Y$ and outputs a value.
A similar approach was taken in~\cite{Berry92}.

\fi

In the definitions given below, we sometimes use a \emph{pattern} of the form $p_1\sep...\sep p_n$ to represent a finite multi-set $\{p_1,...,p_n\}$ of the programs $p_1,...,p_n$,
as it is easier to express certain changes made to the programs at certain positions separated by $\sep$.

\begin{mydef}[Function $\Merge$]
    \label{def:Function Merge}
            Given a pattern of the form $A = (p_1\sep...\sep p_n)$ ($n\ge 2$), where $p_i = \alpha_i\seq q_i$ for any $1\le i\le n$, the function $\Merge(A)$ is defined as
    $$\Merge(A)\dddef \rMerge(A, \rho_0, R_0, \Must_0, \SV_0), $$
    where $\rho_0 = \epsilon$, 
    $R_0 = \Must_0 = \SV_0 = \emptyset$.

    The recursive procedure $\rMerge(A, \rho, r, \Must, \SV)$, in which
    $\rho$ is an atomic program, $R$ is a pattern,
    $\Must$ is multi-set of signals, $\SV$ is a multi-set of signal-set pairs, is defined as follows:
    \begin{enumerate}
        \item If $\alpha_i = \epsilon$ for some $1\le i\le n$, then
        $$\rMerge\left(\begin{gathered}(...\sep p_{i-1}\sep \alpha_i\seq q_{i}\sep p_{i+1}\sep...),\\
        \rho, R, \Must, \SV
        \end{gathered}\right)
        \dddef
        \rMerge\left(\begin{gathered}(...\sep p_{i-1}\sep p_{i+1}\sep...),\\
        \rho, (R\sep q_i), \Must, \SV
        \end{gathered}\right)
        ;
        $$

        \item If $\alpha_i = a\nex \beta$ for some $1\le i\le n$, where $a\in \{\psi?, x:=e\}$, then
        $$\rMerge\left(\begin{gathered}(p_1\sep ...\sep \alpha_i\seq q_{i}\sep...\sep p_n),\\
        \rho, R,  \Must, \SV
        \end{gathered}\right)
        \dddef
        \rMerge\left(\begin{gathered}(p_1\sep ...\sep \beta\seq q_{i}\sep...\sep p_n),\\
        \rho\append a, R, \Must, \SV
        \end{gathered}\right)
        ;
        $$

        \item If $\alpha_i = \sig!e \nex \beta$ for some $1\le i\le n$, then
        $$\rMerge\left(\begin{gathered}(p_1\sep ...\sep \alpha_i\seq q_{i}\sep...\sep p_n),\\
        \rho, R, \Must, \SV
        \end{gathered}\right)
        \dddef
        \rMerge\left(\begin{gathered}(p_1\sep ...\sep \beta\seq q_{i}\sep...\sep p_n),\\
        \rho, R, \Must\uplus \{\sig!e\}, \SV
        \end{gathered}\right)
        ;
        $$
        \ifx
        where
        $
        \Sigs' = \left\{\begin{array}{ll}
            \Sigs\cup \{i\mapsto (\Sigs(i) - \{\sig!e'\})\cup \{\sig!e\}\} & \mbox{if $\sig!e'\in \Sigs(i)$}\\
            \Sigs\cup \{i\mapsto \Sigs(i) \cup \{\sig!e\}\} & \mbox{otherwise}
        \end{array}\right.
        $
        \fi
        \item If $\alpha_i = \varrho_i? \nex \beta_i$ for all $1\le i\le n$, let $\Can = \getCan(A, \Must)$, then
        \begin{enumerate}[(i)]
        \item if $\varrho_j = \hat{\sig}(x)$ and $\Match(\varrho_j, \Must)=\true$ for some $1\le j\le n$, then
        $$\rMerge\left(\begin{gathered}(p_1\sep ...\sep \varrho_j\nex \beta_j\seq q_{j}\sep...\sep p_n),\\
        \rho, R, \Must, \SV
        \end{gathered}\right)
        \dddef
        \rMerge\left(\begin{gathered}(p_1\sep ...\sep r\sep...\sep p_n),\\
        \rho, R, \Must, \SV\cup \{(\sig, \Must_\sig)\}
        \end{gathered}\right)
        ,
        $$
        where $r = (\beta_j\seq q_j)[\comb(\{e\}_{\sig!e\in \Must})/x]$, $\Must_{\sig} = \{\sig!e\ |\ \sig!e\in \Must\}$;

        \item
        if $\varrho_j = \bar{\sig}$ and $\Match(\varrho_j, \Can)=\true$ for some $1\le j\le n$,
        then
        $$\rMerge\left(\begin{gathered}(p_1\sep ...\sep \varrho_j\nex \beta_j\seq q_{j}\sep...\sep p_n),\\
        \rho, R, \Must,  \SV
        \end{gathered}\right)
        \dddef
        \rMerge\left(\begin{gathered}(p_1\sep ...\sep \beta_j\seq q_j\sep...\sep p_n),\\
        \rho, R, \Must,  \SV
        \end{gathered}\right)
        ,
        $$

        \item if $\Match(\varrho_j, \Must)=\false$ for all $1\le j\le n$, then
        $$\rMerge\left(\begin{gathered}(p_1\sep ...\sep \varrho_j\nex \beta_j\seq q_{j}\sep...\sep p_n),\\
        \rho, R, \Must,  \SV
        \end{gathered}\right)
        \dddef
        \rMerge\left(\begin{gathered}\emptyset,\\
        \halt, R, \Must,  \SV
        \end{gathered}\right)
        ,
        $$

        \item otherwise,
        $$\rMerge\left(\begin{gathered}(p_1\sep ...\sep p_n),\\
        \rho, R, \Must, \SV
        \end{gathered}\right)
        \dddef
        (\false, \rho, R)
        .
        $$

        \end{enumerate}

        \item If $A = \emptyset$, then
        $$\rMerge(
        A, \rho, R, \Must, \SV)
        \dddef
        (b, \rho, R),
        $$
        where
        $$
        b = \left\{\begin{array}{ll}
            \true , &\begin{aligned}\mbox{if for any $(\sig, A)\in \SV$, $A = \Must_\sig = \{\sig!e\ |\ \sig!e\in \Must\}$}\end{aligned}\\
            \false, &\mbox{otherwise}
        \end{array}
        \right\}.
        $$
    \end{enumerate}

\end{mydef}

The function $\Merge$ calls a recursive function $\rMerge$, which executes the events at each micro step of different macro events according to different cases (the cases 1 - 5).
$\rMerge$ returns a triple $(b, \rho, R)$, where $b$ indicates whether the communication behaviour between the programs $p_1,...,p_n$ is constructive or not.
The parameters $\rho, R, \Must, \SV$ keep updated during the recursion process of $\rMerge$.
$\rho$ is the behaviour after the communication between the macro events $\alpha_1,...,\alpha_n$.
It collects all closed events executed at the current instant in a logical order.
$R$ represents the resulted $n$ programs after the communication, corresponding to the original programs $p_1,...,p_n$.
$\Must$, $\SV$ (and $\Can$, computed from $\Must$ in the case 4,) maintain the information about signals that is necessary for analyzing the communication in order to obtain a constructive behaviour.
$\Must$ and $\Can$ record the sets of signals that \emph{must} emit and that \emph{can} emit at the current instant respectively.
$\SV$ records the set of observed signals for each signal test of the form $\hat{\sig}(x)$.
Their usefulness will be illustrated below.

\ifx
$\Sigs$ captures the set of all signals emitted at the current instant.
It helps check the signal tests at each micro step of macro events.
$\Must$, $\Cannot$ and $\SV$ are parameters for checking whether the communication behaviour at an instant follows the principle 1.
$\Must$ and $\Cannot$ record the set of signals that \emph{must} emit and the set of signals that \emph{cannot} emit at the current instant respectively.
$\SV$ records the set of observed signals for each signal tests of the form $\hat{\sig}(x)$.
\fi

The recursion process of $\rMerge$ is described as follows.

The case 1 is trivial. The skip $\epsilon$ has no effect on the communication between other programs so we simply remove the program and add the resulted part $q_i$ after the execution to the set $R$.

In the case 2, a closed event $a$ at the current micro step of a macro event $\alpha_i$ is executed.
The operator $\append$, defined later in Def.~\ref{label:Operator append}, appends $a$ to the tail of the sequence of events $\rho$.
Note that the order to put the closed events of different macro events into $\rho$ at each recursion step is irrelevant, because as indicated in Def.~\ref{def:Restriction on Parallel SPs}, all variables are local and
do not interfere with each other. For example, in a program $(x>1?\nex x := x+1\nex \epsilon)\para (y := 1\nex \epsilon)$,
the logical execution orders between $x>1?$ and $y:=1$, and between $x:=x+1$ and $y:=1$, are irrelevant.
What really matters is the order between $x>1?$ and $x:=x+1$, which is preserved in $\rho$.

In the case 3, we emit a signal at the current micro step of a macro event and put it in the set $\Must$.
Note that the order to put the signals of different macro events into $\Must$ is irrelevant.
This is because we do the signal-test matches (in the case 4) always after the emissions of all signals at the current micro step.
This stipulation makes sure that we collect \emph{as many as we can} the signals that \emph{must} emit at the current instant before checking the signal tests.
For example, let $p_1 = \sig_1!3\nex \epsilon$, $p_2 = \sig_2!5\nex \epsilon$, $p_3 = \bar{\sig_2}(x)?\nex \epsilon$,
then in the program $\para(p_1, p_2, p_3)$ the execution order of $\sig_1!3$ and $\sig_2!5$ is irrelevant, because
$\bar{\sig_2}(x)$ will be matched after the executions of $\sig_1!3$ and $\sig_2!5$.
If $\bar{\sig_2}(x)?$ does not wait till $\sig_2!5$ is executed, it would hold a wrong view towards the state of the signal $\sig_2$ because $\bar{\sig_2}(x)$ is matched.


The case 4 checks the signal tests $\varrho_1,...,\varrho_n$ after all signals at the current micro steps are emitted.
Here, based on the idea originally from~\cite{Berry99}, we propose an approach, called \emph{Must-Can approach} in this paper, to check the signal tests one by one according to two aspects of information currently obtained about signals:
the set of signals ($\Must$) that are certain to occur because they have already been executed in the case 3, and
the set of signals ($\Can$) that possibly occur because the current information about signals cannot decide their absences.
The basic idea of our approach is illustrated as the following steps, corresponding to the different cases \rmn{1} - \rmn{4} of the case 4.
\begin{enumerate}[(i)]
\item We firstly try to match all signal tests of the form $\hat{\sig}(x)$ with the current set $\Must$.
The result of a match is returned by the function $\Match$ (defined later in Def.~\ref{def:Function Match}), which checks whether a multi-set of signals satisfy the condition of a signal test.
If a signal test $\hat{\sig}(x)$ is successfully matched by $\Must$, it means that this test must be matched at the current instant since the signal $\sig$ must be emitted at the current instant.

We substitute the variable $x$ in the program $\beta_j\seq q_j$ with a value computed by a combinational function $\comb$.
When several signals with the same name are emitted at the same instant, their values are combined into a value by a function $\comb$.
For example, we can define a combinational function as $\mathit{add}(V)\dddef \Sigma_{v\in V} v$.
A similar approach was taken in~\cite{???}.

\item If all signal tests of the form $\hat{\sig}(x)$ have been tried and their matches fail, we try to match the signal tests of the form $\bar{\sig}$ with the set $\Can$, which is computed by the current set $\Must$ in the function $\getCan$ (defined later in Def.~\ref{def:Function getCan}).
The successful match of a signal test $\bar{\sig}$ means that this test must be matched at the current instant since the signal $\sig$ has no possibilities to be emitted at the current instant.

\item If all signal tests in the form of either $\hat{\sig}(x)$ or $\bar{\sig}$ fail in matching with the set $\Must$,
the communication is blocked because it means that the current emissions of signals have decided the mismatches of all signal tests at the current micro steps.

\item If none of the match conditions in the above 3 steps holds, for example, $\Match(\bar{\sig}, \Can)$ always holds for a signal test $\bar{\sig}$,
then we conclude that the communication behaviour between the programs $p_1,...,p_n$ is not constructive and end the procedure $\rMerge$ by returning $b$ as false.

\end{enumerate}

\ifx
Here we adopt the idea from~\cite{???} to consecutively check the signal tests one by one according to two sets of signals $\Must$ and $\Can$.
$\Must$ is used for checking signal tests of the form $\hat{\sig}(x)$, while $\Can$ is used for checking signal tests of the form $\bar{\sig}$.
By this approach, we can guarantee that $\rMerge$ produces a constructive result if $b=\true$,
The basic idea behind this approach is stated as following steps, where each step corresponds to the different cases \rmn{1} - \rmn{4} of the case 4.
\fi

For example, let $p_1 = \bar{\sig_1}?\nex \sig_2!5\nex \epsilon$, $p_2 = \bar{\sig_2}?\nex \sig_1!3\nex \epsilon$, consider the communication of the program $p_1\para p_2$,
in the case 4, we have $\Must = \emptyset$ and $\Can = \{\sig_2!5, \sig_1!3\}$.
We cannot proceed because $\Match(\bar{\sig_1}, \Can) = \Match(\bar{\sig_2}, \Can) = \false$.
In fact, this program is logical incorrect because it allows two possible states of signals: 1) $\sig_1$ is emitted and $\sig_2$ is absent, and 2) $\sig_1$ is absent and $\sig_2$ is emitted.
Let $q_1 = \bar{\sig_2}?\nex \sig_4\nex \epsilon$, $q_2 = \hat{\sig_1}?\nex \sig_3\nex \sig_2\nex \epsilon$ and $q_3=\hat{\sig_3}?\nex \sig_1$,
in the communication of the program $\para(q_1,q_2,q_3)$, $\Must = \emptyset$.
It is easy to see that though $\bar{\sig_2}$ is possible to be matched and so $\sig_4$ can be emitted, the tests $\hat{\sig_1}$ and $\hat{\sig_3}$ will never be matched.
So $\Can = \{\sig_4\}$ and $\Match(\bar{\sig_2},\Can)=\true$. Hence this program is constructive.

\ifx
The case 4 checks whether the current environment $\Sigs$ matches the signal tests $\varrho_1,...,\varrho_n$ at the current micro steps.
The procedure of the matching is described in the function $\Chk$ above.
If $\Sigs$ cannot match any signal test, the communication is blocked so we stop the recursion process by setting $A$ to $\emptyset$ and return $\rho$ as $\halt$.
Note that in this case, all signal tests must be matched immediately if they can be matched.
Otherwise, it violates the intuition that there is no reason to delay any matches that can be matched at the current micro step, and
it may cause non-deterministic behaviours of a communication.
For example, let $p_1 = \bar{\sig_1}?\nex \sig_2!5\nex \epsilon\seq q_1$, $p_2 = \bar{\sig_2}?\nex \sig_1!3\nex \epsilon\seq q_2$,
then in $p_1\para p_2$,
if we check $\bar{\sig_1}$ and execute $\sig_2!5$ before checking $\bar{\sig_2}$, $\sig_2$ is emitted at the current instant so $\bar{\sig_2}$ is not matched.
But if we check $\bar{\sig_2}$ and execute $\sig_1!3$ before checking $\bar{\sig_1}$, then $\sig_1$ is emitted at the current instant so $\bar{\sig_1}$ is not matched.
So we actually obtain two different semantical explanations for a single SP.
\fi

The case 5 ends the recursion procedure when all macro events $\alpha_1,...,\alpha_n$ at the current instant are executed. 
 At this end, we need to check whether all programs $p_1,...,p_n$ hold a consistent view towards the value of each signal.
Only when a signal test $\hat{\sig}(x)?$ observes the set $\Must_\sig$ of all emissions of $\sig$, we can guarantee that
  all programs agree with the value of $\sig$, i.e., $\comb(\{e\}_{\sig'!e\in \Must_\sig})$.
For example, let $p_1 = \sig_1!3\nex \hat{\sig_2}?\nex \sig_1!5\nex \epsilon$, $p_2 = \hat{\sig_1}(x)?\nex \sig_2\nex \epsilon$,
then the program $p_1\para p_2$ does not follow the logical coherence law if we choose $\comb$ as the function $\mathit{add}$ (defined above in the case 4\rmn{1}),
because the value received by the signal test $\hat{\sig_1}(x)?$ is $3$, while the value of the signal $\sig_1$ at the current instant is $\mathit{add}(3,5)$, which is $8$.

\ifx
\paragraph{Definitions of the Operator $\append$ and the Function $\Chk$}
The operator $\append$ and the function $\Chk$ appeared in the function $\Merge$ are defined as follows, where $\Chk$ relies on another auxiliary function $\Match$.

\begin{mydef}[Operator $\append$]
\label{label:Operator append}
    Given an event $\alpha$ and an event $a \in \{\psi?, x:=e\}$,
    $\alpha\append a$ is defined as:
    $$\alpha\append a\dddef \left\{
    \begin{array}{ll}
    a\nex \epsilon & \mbox{if $\alpha = \epsilon$}\\
    b\nex (\alpha'\append a) &\mbox{if $\alpha = b\nex \alpha'$}
    \end{array}
    \right\}.$$
\end{mydef}

\begin{mydef}[Function $\Match$]
	Given a signal test $\varrho$ and a set of signals $Y$, the function $\Match(\varrho, Y)$ is defined as
	$$
	\Match(\varrho, Y)\dddef \left\{
	\begin{array}{ll}
	\true & \mbox{if $\varrho=\hat{\sig}(x)$ and $\sig\in Y$}\\
	\true & \mbox{if $\varrho=\bar{\sig}$ and $\sig\notin Y$}\\
	\false & \mbox{otherwise}
	\end{array}
	\right\}.
	$$
\end{mydef}

The function $\Match$ checks whether a set of signals $Y$ matches a signal test $\varrho$.
It returns true if $Y$ matches $\varrho$, and returns false otherwise.

\begin{mydef}[Function $\Chk$]
	\label{def:Function Chk}
	Given a multi-set $X$ of programs of the form $X = (\varrho_1?\nex \beta_1\seq q_1\sep ...\sep \varrho_n?\nex \beta_n\seq q_n)$ and a multi-set $Y$ of signals,
    then the function $\Chk(X, Y)$ is defined as:
	$$\Chk(X, Y)\dddef (b, C, R, \SV), $$
	where
    $$
    \begin{aligned}
    b\dddef&\  \left\{
        \begin{array}{ll}
            \false & \mbox{if $\Match(\varrho_i, Y)=\false$ for all $1\le i\le n$}\\
            \true & \mbox{otherwise}
        \end{array}
    \right\},
    \\
	C\dddef&\ \{\sig\ |\ \varrho_i = \bar{\sig}, \Match(\varrho_i, Y) = \true, 1\le i\le n\},
	\\
    \SV \dddef&\ \{(\sig, A)\ |\ \Match(\varrho_i, Y)=\true, \varrho_i = \hat{\sig}(x), A = \{\sig!e\ |\ \sig!e\in Y\}, 1\le i\le n\},
    \end{aligned}
    $$
	$$
	R\dddef (r_1\sep...\sep r_n)\mbox{ with }r_i = \left\{\begin{array}{ll}
	    \varrho_i?\nex \beta_i\seq q_i &\mbox{if $\Match(\varrho_i, Y) = \false$}\\
	    (\beta_i\seq q_i)[\comb(\{e\}_{\sig!e\in Y})/x] &\mbox{if $\Match(\varrho_i, Y)=\true$ and $\varrho_i = \hat{\sig}(x)$}\\
	    \beta_i\seq q_i &\mbox{if $\Match(\varrho_i, Y)=\true$ and $\varrho_i = \bar{\sig}$}
	\end{array}\right\}
	$$
	for any $1\le i\le n$.
	\ifx
	\begin{enumerate}
		\item If there exists a guarded event $a\in X$ s.t. $\Match(\varrho_a, Y) = \none$, then
		$$\Chk(X, Y) \dddef \none;$$
		\item Otherwise 
		 $$\Chk(X, Y)\dddef (M, C, \bigcup_{a\in X}\Match(\varrho_a, Y)), $$
		 where
		 $$\begin{aligned}
		 M =&\ \Signal(X),\\
		 C =&\ \{\sig\ |\ \bar{\sig}\in \varrho_b\}.
		 \end{aligned}		
		 $$
		 \end{enumerate}
    \fi
		
\end{mydef}

The function $\Chk$ checks whether a multi-set of signals $Y$ satisfies the signal tests $\varrho_1,...,\varrho_n$ of the set of programs $X$,
and records the set of  resulted programs after the checking.
It returns a 4-tuple $(b, C, R, \SV)$, where $b$ has been introduced above.

$C$ is the set of signals that cannot emit.
When $Y$ matches a signal test of the form $\bar{\sig}$, it means that during the current reaction $\sig$ cannot appear in the environment, or the principle 1 will be violated.
So we add $\sig$ to the set $C$.

$\SV$ is a multi-set of pairs $(\sig, A)$ where $A$ is a multi-set of signals.
It records that when a signal test $\hat{\sig}(x)$ is checked, the set of the signals $\sig$ it can observe is $A$.

\ifx
$C$ is the set of signals that cannot emit obtained according to the signal tests $\varrho_1,...,\varrho_n$;
$R$ represents the resulted $n$ programs after the match, corresponding to the programs $\alpha_1\seq q_1,...,\alpha_n\seq q_n$;
$\SV'$ records the set of signals observed by each signal test of the form $\hat{\sig}(x)$.
\fi

$R$ is the multi-set of resulted programs corresponding to $X$ after the checking.
In $R$, each program $r_i$ depends on the signal test $\varrho_i$ and whether $Y$ can match $\varrho_i$.
If $Y$ does not match $\varrho_i$, $r_i$ keeps the same as in $X$.
If $Y$ matches $\varrho_i$ and $\varrho_i$ is of the form $\hat{\sig}(x)$, we substitute the variable $x$ in the program $\beta_i\seq q_i$ with a value computed by a combinational function $\comb$, which accepts the values of all signals with the same name that are emitted in $Y$ and outputs a value.
A similar approach was taken in~\cite{???}.
\fi

\paragraph{Definitions of $\append$, $\Match$ and $\getCan$}
The operator $\append$ and the functions $\Match$ and $\getCan$ used in the function $\Merge$ above are defined as follows.

\begin{mydef}[Operator $\append$]
\label{label:Operator append}
    Given an event $\alpha$ and an event $a \in \{\psi?, x:=e\}$,
    $\alpha\append a$ is defined as:
    $$\alpha\append a\dddef \left\{
    \begin{array}{ll}
    a\nex \epsilon & \mbox{if $\alpha = \epsilon$}\\
    b\nex (\alpha'\append a) &\mbox{if $\alpha = b\nex \alpha'$}
    \end{array}
    \right\}.$$
\end{mydef}

The operator $\append$ appends an event $a$ to the tail of a macro event $\alpha$ before $\epsilon$.

\begin{mydef}[Function $\Match$]
    \label{def:Function Match}
	Given a signal test $\varrho$ and a multi-set of signals $Y$, the function $\Match(\varrho, Y)$ is defined as
	$$
	\Match(\varrho, Y)\dddef \left\{
	\begin{array}{ll}
	\true & \mbox{if $\varrho=\hat{\sig}(x)$ and $\sig = \sig'$ for some $\sig'!e\in Y$}\\
	\true & \mbox{if $\varrho=\bar{\sig}$ and $\sig\neq \sig'$ for all $\sig'!e\in Y$}\\
	\false & \mbox{otherwise}
	\end{array}
	\right\}.
	$$
\end{mydef}

The function $\Match$ checks whether a multi-set of signals $Y$ matches a signal test $\varrho$.
It returns true if $Y$ matches $\varrho$, and returns false otherwise.

\begin{mydef}[Function $\getCan$]
\label{def:Function getCan}

    Given a pattern of the form $A = \alpha_1 \sep...\sep \alpha_n$, where $\alpha_i = \varrho_i?\nex \beta_i$ ($1\le \i\le n$), and a multi-set of signals $M$, then
    $$\getCan(A, M)\dddef \rgetCan(A, M, \Can_0),$$
    where $\Can_0=\emptyset$.
    The function $\rgetCan(A, M, \Can)$ is recursively defined as follows:
    \begin{enumerate}
    \item if $\alpha_i=\epsilon$ for some $1\le i\le n$, then
    $$
    \rgetCan((...\sep \alpha_{i-1}\sep \alpha_i\sep \alpha_{i+1}\sep...), M, \Can)\dddef \rgetCan((...\sep \alpha_{i-1}\sep \alpha_{i+1}\sep...), M, \Can);
    $$

    \item if $\alpha_i = a\nex \beta$ for some $1\le i\le n$, where $a\in \{\psi?, x:=e\}$, then
    $$
    \rgetCan((\alpha_1\sep...\sep \alpha_{i}\sep...\sep\alpha_n), M, \Can)\dddef \rgetCan((\alpha_1\sep...\sep \beta\sep...\sep\alpha_n), M, \Can);
    $$

    \item if $\alpha_i = \sig!e\nex \beta$ for some $1\le i\le n$, then
    $$
    \rgetCan((\alpha_1\sep...\sep \alpha_{i}\sep...\sep\alpha_n), M, \Can)\dddef \rgetCan((\alpha_1\sep...\sep \beta\sep...\sep\alpha_n), M, \Can\uplus\{\sig!e\});
    $$

    \item if $\alpha_i = \hat{\sig}(x)?\nex \beta$ and $\Match(\hat{\sig}(x),\Can\cup M)=\true$ for some $1\le i\le n$, then
    $$
    \rgetCan((\alpha_1\sep...\sep \alpha_{i}\sep...\sep\alpha_n), M, \Can)\dddef \rgetCan((\alpha_1\sep...\sep \beta\sep...\sep\alpha_n), M, \Can);
    $$

    \item if $\alpha_i = \bar{\sig}?\nex \beta$ and $\Match(\bar{\sig},M)=\true$ for some $1\le i\le n$, then
    $$
    \rgetCan((\alpha_1\sep...\sep \alpha_{i}\sep...\sep\alpha_n), M, \Can)\dddef \rgetCan((\alpha_1\sep...\sep \beta\sep...\sep\alpha_n), M, \Can);
    $$

    \item otherwise
    $$
    \rgetCan((\alpha_1\sep...\sep \alpha_{i}\sep...\sep\alpha_n), M, \Can)\dddef \Can;
    $$

    \end{enumerate}

\end{mydef}

The function $\getCan$ returns a set $\Can$ of signals that are possibly emitted in the communication of $n$ macro events $\alpha_1,...,\alpha_n$ at the current instant.
In $\getCan$, whether a signal test is possible to be matched is judged according to the set of signals $M$ that must emit and the current set $\Can$.
In the case 4,
if a signal $\sig$ must or can be emitted in the current environment, then we can conclude that $\hat{\sig}(x)$ can be matched.
In the case 5, a signal test $\bar{\sig}$ can possibly be matched only when we cannot decide whether the signal $\sig$ is emitted in the current environment.
Other cases are easy to understand and we omit their explanations.

\paragraph{Definition of the Function $\valt(r)$}

With the function $\Merge$, we define the valuation of the trecs $\para(p_1,...,p_n)$ as follows.

\begin{mydef}[Valuation of the Trecs $\para(p_1,...,p_n)$]
	\label{def:Valuation of the Trecs}
	The valuation of a trec $\para(p_1,...,p_n)\in \Trec$, denoted as $\valt(\para(p_1,...,p_n))$, is recursively defined as follows:

	\begin{enumerate}
	    \item If $p_i = \halt$ for some $1\le i\le n$, then
		$$\valt(\para(p_1,...,p_n))\dddef \val(\halt);$$
		
		\item If $p_i = \noth$ for some $1\le i\le n$, then
		$$\valt(\para(p_1,...p_{i-1}, p_i, p_{i+1},..., p_n))\dddef
		\valt(\para(p_1,...,p_{i-1}, p_{i+1},..., p_n));$$
		
        \item If $p_i = \para(r_{1},...,r_n)\seq q$ for some $1\le i\le n$ and $\valt(\para(r_1,...,r_n)) = \val(r)$, where
        $r$ is the program computed in the procedure $\valt(\para(r_1,...,r_n))$, which is either $\halt$ or in the form of $\alpha\seq \para(r'_1,...,r'_m)$,
        then
        $$\valt(\para(p_1,...,p_i,...,p_n))\dddef \valt(\para(p_1,...,r\link q,...,p_n));$$

		\ifx
		\item If $p_i = q\cap r$ for some $1\le i\le n$, then
		$$\rfval\left(\begin{gathered}p_1\sep...\sep p_{i-1}\sep p_i\sep p_{i+1}\sep...\sep p_n,\\
		\Must,\Cannot, \rho, r\end{gathered}\right)\dddef
		\rfval\left(\begin{gathered}p_1\sep...\sep p_{i-1}\sep q\sep r\sep p_{i+1}\sep...\sep p_n, \\
		\Must, \Cannot, \rho, r\end{gathered}\right);$$
		\fi
		
		\ifx
		\item If $p_i = \psi_i?;q$ for some $1\le i\le n$, then
		$$\rfval\left(\begin{gathered}p_1\sep...\sep p_{i-1}\sep p_i\sep p_{i+1}\sep...\sep p_n,\\
		\Must,\Cannot, \rho, r\end{gathered}\right)
		\dddef
		\val(\psi_i)\circ \rfval\left(\begin{gathered}p_1\sep...\sep p_{i-1}\sep q\sep p_{i+1}\sep...\sep p_n, \\
		\Must, \Cannot, \rho, r\end{gathered}\right);$$
		\fi
		
		\item
		If $p_i = \alpha_i;q_i$ for all $1\le i\le n$, then
		$$
        \begin{array}{ll}
        \valt(\para(p_1,...,p_n))
		\dddef
		\val(\alpha\seq \para(q'_1,...,q'_n)) &\mbox{if $b=\true$},
        \end{array}
		$$
		where $(b, \alpha, (q'_1\sep...\sep q'_n)) = \Merge(\alpha_1;q_1\sep...\sep \alpha_n;q_n)$.

		\ifx
		\item If $A = G \cup U$ and $U\neq \emptyset$, where $G$ is the multiset of all guarded events of the form $\varrho?\alpha$ or $\varrho\& \psi?\alpha$ in $A$, and $U$ is the multiset of all unguarded events of the form $\alpha$ or $\psi?\alpha$ in $A$, then
		$$\rfval(A, \Must, \Cannot, \rho, r)\dddef \rfval(A - U, \Must\uplus \Signal(U), \Cannot, \Closed(\rho\join (\underset{^{a\in U}}{\join}a)), r);$$
		\fi
		
		\ifx
		\item If $(p_1\sep...\sep p_n) = \varrho_1?;q_1\sep ...\sep \varrho_n?;q_n$, let $A = \varrho_1?;q_1\sep ...\sep \varrho_n?;q_n$, then
		$$\rfval\left(\begin{gathered}p_1\sep...\sep p_n, \\
		\Must, \Cannot, \rho, r\end{gathered}\right)
		\dddef \left\{\begin{array}{ll}
		\val(\halt), &\mbox{if $R = A$}\\
		\rfval	\left(\begin{gathered}R, \\
		\Must, \Cannot\cup C,\rho, r\end{gathered}\right),  &\mbox{otherwise}
		\end{array}
		\right.,$$
		where $(C, R) = \Chk(A, \Must)$.
		\fi
		
		\ifx
		\item If $A = G$ and $G\neq \emptyset$, where $G$ is the multiset of all guarded events of the form $\varrho?\alpha$ or $\varrho\& \psi?\alpha$ in $A$,
		then
		$$\rfval(A, \Must, \Cannot, \rho, r)\dddef \left\{\begin{array}{ll}
		\rfval	\left(\begin{gathered}A - B, \Must\uplus M, \Cannot\cup C,\\
		\rho\times \Closed(\underset{^{a\in B'}}{\times}a), r[\mathit{Sub}]\end{gathered}\right), &\mbox{if $B\neq \emptyset$}\\
		\val(\halt),  &\mbox{otherwise}
		\end{array}
		\right.,$$
		where $B\subseteq G$ is the maximum set s.t. $\Chk(B, \Must) = (M, C, \mathit{Sub})$;
		$$B' = \{\varrho?\beta\ |\ \varrho?\alpha\in B, \beta = \alpha[\mathit{Sub}]\}\cup \{\varrho\& \psi?\beta\ |\ \varrho\& \psi?\alpha\in B, \beta = \alpha[\mathit{Sub}]\};$$
		\fi
		
		\ifx
		\item If $(p_1\sep...\sep p_n) = \emptyset$, then
		$$\rfval\left(\begin{gathered}(p_1\sep...\sep p_n), \\
		\Must, \Cannot, \rho, r\end{gathered}\right)
		\dddef \left\{
		\begin{array}{ll}
		\val(\rho)\circ \rfval(r, \emptyset, \emptyset, \none, \none), &\mbox{if $\Must\cap \Cannot = \emptyset$ and $|r|>1$}\\
		\val(\rho)\circ \val(r), &\mbox{if $\Must\cap \Cannot = \emptyset$ and $|r|=1$}\\
		\val(\rho), &\mbox{if $\Must\cap \Cannot = \emptyset$ and $r = \none$}\\
		\val(\halt),  &\mbox{if $\Must\cap \Cannot \neq \emptyset$}
		\end{array}
		\right..$$
		\fi
	\end{enumerate}
\end{mydef}

The cases 1 and 2 are easy to understand.
In the case 3, when one of the trec $p_i$ contains a parallel subprogram $\para(r_1,...,r_n)$, we first compute the valuation of this subprogram, which returns the valuation of a program $r$ of the form $\halt$ or $\alpha\seq \para(r'_1,...,r'_m)$.
Then we replace $p_i$ with the trec $r\link q$, the latter has a macro event $\alpha$ visible at the current instant.

In the case 4, we merge $n$ programs $\alpha_1\seq q_1,...,\alpha_n\seq q_n$ in the function $\Merge$ as described above,
which returns the merged closed macro event $\alpha$ and the resulted $n$ programs $q'_1,...,q'_n$ after the communication at the current instant.
Note that $\valt$ is a partial function and it rejects the parallel programs that do not follow the consistency law 1 at the current instant.
Our semantics makes sure that SPs must follow the consistency law 1, while preserving the data dependencies induced by the logical order between micro events.
Our semantics is actually a \emph{constructive semantics} following~\cite{Berry99}.

We call an SP $p$ a \emph{well-defined} SP if it has a semantics $\val(p)$.
Unless specially pointing out, all SPs discussed in the rest of the paper are well defined.

Note that the function $\valt$ is well defined since a trec is always finite without the star program.

\ifx
From the case 4, it is easy to see that $\valt$ is a partial function because it only returns the parallel programs that follow the principle 1.
Hence, our semantics $\Par$ for parallel programs is a \emph{constructive semantics}~\cite{???}.
\fi

\ifx
\begin{mydef}[Synchronous Execution of Current Events]
	The synchronous execution of a multiset of current events $E$ ($E\neq \emptyset$) is described as the function $\SE(E)$ defined as:
	$$\SE(E)\dddef \Com(E, \emptyset, \emptyset, \none).$$
	The function $\Com(A, \Must, \Cannot, \rho)$, where $A$ is a multiset of events, $\Must$ and $\Cannot$ are sets of signals, $\rho\in \Cl(\ASP)\cup \{\none\}$, is inductively defined as follows:
	
	\begin{enumerate}
		\item If $I\neq \emptyset$, where $I$ is the set of all empty programs $\noth$ in $A$, then
		$$\Com(A,\Must,\Cannot, \rho)\dddef \Com(A - I, \Must, \Cannot, \rho);$$
		
		\item If $\halt\in A$, then
			$$\Com(A,\Must,\Cannot, \rho)\dddef \Com(\emptyset, \Must, \Cannot, \halt);$$
		\item If $A = G \cup U$ and $U\neq \emptyset$, where $G$ is the multiset of all guarded events of the form $\varrho?\alpha$ or $\varrho\& \psi?\alpha$ in $A$, and $U$ is the multiset of all unguarded events of the form $\alpha$ or $\psi?\alpha$ in $A$, then
			$$\Com(A, \Must, \Cannot, \rho)\dddef \Com(A - U, \Must\cup \Signal(U), \Cannot, \rho\times \Closed(\underset{^{a\in U}}{\times}a));$$
			
		\item If $A = G$ and $G\neq \emptyset$, where $G$ is the multiset of all guarded events of the form $\varrho?\alpha$ or $\varrho\& \psi?\alpha$ in $A$,
			then
			$$\Com(A, \Must, \Cannot, \rho)\dddef \left\{\begin{array}{ll}
								\Com\left(\begin{aligned}&A - B, \Must\cup Y, \\
								&\Cannot\cup C,
								\rho\times \Closed(\underset{^{a\in B}}{\times}a)\end{aligned}\right), &\mbox{if $B\neq \emptyset$}\\
								\val(\halt),  &\mbox{otherwise}
								\end{array}
								\right.,$$
			where $B\subseteq G$ is the maximum set s.t. $(b, Y, C) = \Chk(B, \Must)$ and $b = \true$;
		\item If $A = \emptyset$, then
			$$\Com(A, \Must, \Cannot, \rho)\dddef \left\{
			\begin{array}{ll}
			\val(\rho), &\mbox{if $\Must\cap \Cannot = \emptyset$}\\
			\val(\halt),  &\mbox{otherwise}
			\end{array}
			\right..$$
	\end{enumerate}
\end{mydef}

\begin{mydef}[Valuation of a Finite Trec]
	The valuation of a finite trec $p$, denoted by $\val_{\mathit{fin}}(p)$, is defined as
	$$\val_{\mathit{fin}}(p)\dddef \left\{\begin{array}{ll}
	\SE(\CEvt(p)[1]), &\mbox{if $\CEvt(p)[2] = \none$}\\
	\SE(\CEvt(p)[1])\circ \val_{\mathit{fin}}(\CEvt(p)[2]),  &\mbox{otherwise}
	\end{array}
	\right.. $$

\end{mydef}
\fi

To see the rationality of our definition of the semantics of SPs in this section, we state the next proposition, which shows that the trecs of an SP exactly capture
all behaviours of the SP.

\begin{prop}
    \label{prop:trecs property}
    For any SP $p$, $$\val(p) = \bigcup_{r\in \tau(p)}\val(r).$$
\end{prop}

The proof of Def.~\ref{prop:trecs property} is given in Appendix~\ref{section:Proofs for some Propositions and the Soundness of SDL}.

\subsubsection{Valuation of Formulas}
\label{section:Valuation of Formulas}


The semantics of SDL formulas is defined as a set of states in the following definition.

\begin{mydef}[Valuation of SDL Formulas]
    \label{def:Valuation of sDTL Formulas}
    The valuation of SDL formulas is given inductively as follows:
    \begin{enumerate}
        \item $\val(\true)\dddef \bff{S}$;
        \item $\val(\theta(e_1, e_2))\dddef \{s\ |\ s\in \bff{S}, \theta(\val_s(e_1), \val_s(e_2))\mbox{ is true}\}$;
        \item $\val(\neg \phi)\dddef \bff{S} - \val(\phi)$;
        \item $\val(\phi\wedge \psi)\dddef \val(\phi)\cap \val(\psi)$;
        \item $\val(\forall x. \phi)\dddef \{s\ |\ \mbox{for all }n\in \mbb{Z}, s\in \val(\phi[n/x])\}$;
        \item $\val([p]\phi)\dddef \{s\ |\ \mbox{for all $tr\in \val(p)$ with $tr^b = s$, $tr^e\in \val(\phi)$}\}$;
        \item $\val([p]\Box\phi)\dddef \{s\ |\ \mbox{for all $tr\in \val(p)$ with $tr^b=s$, $tr\in \val_\pi(\Box\phi)$}\}$, where the valuation $\val_\pi(\Box\phi)$ of a trace formula $\Box\phi$ is defined as
    $$\val_\pi(\Box\phi)\dddef \{tr\ |\ tr(i)\in \val(\phi)\mbox{ for all $i\ge 1$}\}.$$
    \end{enumerate}
\end{mydef}

The items 1-5 are normal definitions for first-order formulas.
The formula $[p]\phi$ describes the partial correctness of a program, since all traces of $p$ are finite, $tr^e$ always exists.
$[p]\Box\phi$ 
captures safety properties in synchronous models that hold at each instant.
The semantics of a temporal formula $\Box\phi$ is given as $\val_\pi$.

We introduce the satisfaction relation between states and SDL formulas.

\begin{mydef}[Satisfaction Relation $\models$]
\label{def:Satisfaction Relation}
    The satisfaction relation between a state $s\in \bff{S}$ and an SDL formula $\phi$, denoted as $s\models \phi$, is defined s.t.
    $$\mbox{$s\models \phi$ iff $s\in \val(\phi)$}.$$

    We say $\phi$ is valid, denoted by $\models \phi$, if for all $s\in \bff{S}$, $s\models \phi$ holds.
\end{mydef}

\section{Proof System of SDL}
\label{section:Proof System of SDL}
In this section, we propose a sound and relatively complete proof system for SDL to support verification of reactive systems based on theorem proving.
We propose compositional rules for sequential SPs which transform a dynamic formula step by step into AFOL formulas according to the syntactic structure of programs.
We propose rewrite rules for parallel SPs which transform a parallel SP into a sequential one so that the proof of a dynamic formula that contains parallel SPs can be realized by the proof of a dynamic formula that only contains sequential SPs.
Since the communication of parallel SPs is deterministic, a parallel SP can be easily rewritten into a sequential one based on the process given in the function $\Merge$ (Def.~\ref{def:Function Merge}) and
the state space of the sequential SP keeps a linear size w.r.t that of the parallel SP.

In Sect.~\ref{section:Sequent Calculus} we first introduce a logical form called \emph{sequent} which allows us to express deductions in a convenient way.
In Sect.~\ref{section:Compositional Rules for Sequential SPs} and Sect.~\ref{section:Rewrite Rules for Parallel SPs}, we propose rules for sequential and parallel SPs respectively.
In Sect.~\ref{section:Rules for FOL formulas}, we introduce other rules for first-order logic (FOL).

\subsection{Sequent Calculus}
\label{section:Sequent Calculus}

A sequent~\cite{Gentzen34} is of the form
$$\Gamma \seqArrow \Delta, $$
where $\Gamma$ and $\Delta$ are finite multisets of formulas, called \emph{contexts}.
The sequent means the formula
$$\bigwedge_{\phi\in \Gamma} \phi \to \bigvee_{\psi\in \Delta}\psi,$$
i.e., if all formulas in $\Gamma$ hold, then one of formulas in $\Delta$ holds.

When $\Gamma = \emptyset$, we denote the sequent as $\cdot \seqArrow \Delta$, meaning
the formula $\true \to \bigvee_{\phi\in \Delta}\phi$.
When $\Delta=\emptyset$, we denote it as $\Gamma\seqArrow \cdot$, meaning the formula $\bigwedge_{\phi\in\Gamma}\phi\to \false$.

A rule in sequent calculus is of the form
$$
\infer[]
{
    \Gamma\seqArrow \Delta
}
{
    \Gamma_1\seqArrow \Delta_1
    &
    ...
    &
    \Gamma_n\seqArrow \Delta_n
}
,$$
where $\Gamma_1\seqArrow \Delta_1$,...,$\Gamma_n\seqArrow \Delta_n$ are called \emph{premises}, while
$\Gamma\seqArrow \Delta$ is called a \emph{conclusion}.
The rule means that if $\Gamma_1\seqArrow \Delta_1$,...,$\Gamma_n\seqArrow \Delta_n$ are valid (in the sense of Def.~\ref{def:Satisfaction Relation}), then $\Gamma\seqArrow \Delta_n$ is valid.

We simply write rules
$\begin{aligned}\infer[]
{
    \Gamma\seqArrow \phi, \Delta
}
{
    \Gamma_1\seqArrow \psi_1, \Delta_1
    &
    ...
    &
    \Gamma_n\seqArrow \psi_n, \Delta_n
}\end{aligned}$
and
$\begin{aligned}\infer[]
{
    \Gamma, \phi_1,...,\phi_n\seqArrow \Delta
}
{
    \Gamma', \psi\seqArrow \Delta'
}\end{aligned}$
as
$$
\begin{aligned}
\infer[]
{
    \phi
}
{
    \psi_1
    &
    ...
    &
    \psi_n
}
\end{aligned}
\mbox{ and }
\begin{aligned}
\infer[]
{
    \psi
}
{
    \phi_1
    &
    ...
    &
    \phi_n
}
\end{aligned}
$$
respectively
if $\Gamma_1=...=\Gamma_n=\Gamma' = \Gamma$ and $\Delta_1=...=\Delta_n=\Delta' = \Delta$ hold.
In fact, the rule $\begin{aligned}\infer[]
{
    \phi
}
{
    \psi_1
    &
    ...
    &
    \psi_n
}\end{aligned}$
just means that the formula $\bigwedge^n_{i=1}\psi_i \to \phi$ is valid in the sense of Def.~\ref{def:Satisfaction Relation} (cf. Prop.~\ref{prop:sequent} in Appendix~\ref{section:Proofs for some Propositions and the Soundness of SDL}).

When both rules $\begin{aligned}\infer[]{\phi}{\psi}\end{aligned}$ and $\begin{aligned}\infer[]{\psi}{\phi}\end{aligned}$ exist, we simply use a single rule of the form
$$
\infer=[]
{
    \phi
}
{
    \psi
}
$$
to represent them. It means that the formula $\psi \leftrightarrow \phi$ is valid.

\subsection{Compositional Rules for Sequential SPs}
\label{section:Compositional Rules for Sequential SPs}

As shown in Table~\ref{table:Rules for Closed SPs}, the rules for sequential SPs are divided into two types.
The rules of type $(a)$ are initially proposed in this paper for special primitives in SPs, whereas the rules of type $(b)$ are inherited from FODL~\cite{Harel00} and differential temporal dynamic logic (DTDL)~\cite{Platzer07}.

Explanations of each rule in Table~\ref{table:Rules for Closed SPs} are given as follows.

The rules of the type $(a)$ are for atomic programs.
Rule $(\alpha, \Box\phi)$ says that proving that $\phi$ is satisfied at each instant of the macro event $\alpha$, is equal to prove that $\phi$ holds before and after the execution of $\alpha$.
This is because a macro event $\alpha$ only concerns two instants, i.e., the instants before and after the execution of $\alpha$. The time does not proceed in a macro event.
The rules $(\psi?)$ and $(x:=e)$ deal with the micro events in a macro event.
In rule $(\psi?)$, if the test $\psi?$ does not hold, the formula $[\psi?\nex \alpha]\phi$ is always true because there is no trace in the program $\psi?\nex \alpha$.
Rule $(x:=e)$ is similar to the rule for assignment in FODL. 
Note that the assignment $x:=e$ also affects on the micro steps after it (i.e. $\alpha$) and this reflects the data dependencies between different micro events of a macro event.
In rule $(\epsilon)$, 
since $\epsilon$ only skips the current instant, it does not affect a state property $\phi$.
Rule $(\noth)$ says that
$\noth$ neither consumes time nor does anything, so only one instant is involved. Thus $[\noth]\xi$ equals to that $\phi$ holds at the current instant.
Rule $(\halt)$ says that $[\halt]\xi$ is always true because there is no trace in $\halt$.
\ifx
[this part, wrong???]In the rules $(\epsilon)$, $(\noth)$ and $(\halt)$,
since $\epsilon$ only skips the current instant, it does not affect a state property $\phi$;
$\noth$ neither consumes time nor does anything, so only one instant is involved. Thus $[\noth]\xi$ equals to that $\phi$ holds at the current instant;
$[\halt]\xi$ is always true because there is no trace in $\halt$.
\fi

    \begin{table}[htpb]
         \begin{center}
         \noindent\makebox[\textwidth]{%
         \scalebox{1.0}{
         \begin{tabular}{|c|}
         \toprule
         \multicolumn{1}{|c|}{
                \begin{tabular}{c c c c c c}
                     $\infer=[^{(\alpha, \Box\phi)}]{
                        [\alpha]\Box\phi
                     }
                     {
                        \phi
                        \wedge
                        [\alpha]\phi
                     }
                     $
                     &
                     $\infer=[^{(\psi?)}]{
                        [\psi?\nex \alpha]\phi
                     }
                     {
                        \psi
                        \to
                        [\alpha]\phi
                     }
                     $
                     &
                     $\infer=[^{(x:=e)}]{
                        [x:=e\nex \alpha]\phi
                     }
                     {
                        ([\alpha]\phi)[e/x]
                     }
                     $
                     &
                     $\infer=[^{(\epsilon)}]{
                        [\epsilon]\phi
                     }
                     {
                        \phi
                     }
                     $
                     &
                     $
                     {
                     \infer=[^{1\ (\noth)}]
                     {[\noth]\xi}{\phi}
                     }
                     $
                     &
                     $
                     \infer=[^{(\halt)}]
                     {[\halt]\xi}{\true}
                     $
             \end{tabular}
         }
         \\
         \multicolumn{1}{|c|}{
                {\footnotesize (a) rules special in SDL}
            }
         \\
         \midrule
         \multicolumn{1}{|c|}{
            \begin{tabular}{c c c c c}
                 $
                 \infer=[^{(\seq,\phi)}]
                 {
                    [p\seq q]\phi
                 }
                 {
                    [p][q]\phi
                 }
                 $
                 &
                 \infer=[^{(\seq,\Box\phi)}]
                 {
                    [p\seq q]\Box\phi
                 }
                 {
                    [p]\Box\phi
                    \wedge
                    [p][q]\Box\phi
                 }
                 &
                 \infer=[^{(\cup)}]
                 {
                    [p\cup q]\xi
                 }
                 {
                    [p]\xi
                    \wedge
                    [q]\xi
                 }
                 &
                 $
                 \infer=[^{(*,\Box\phi)}]
                 {
                    [p^*]\Box\phi
                 }
                 {
                    [p^*][p]\Box\phi
                 }
                 $
                 &
                 $
                 \infer=[^{(*)}]
                 {
                    [p^*]\xi
                 }
                 {
                    [\noth\cup p\seq p^*]\xi
                 }
                 $
                 \\
            \end{tabular}
         }
         \\
         \\
         \multicolumn{1}{|c|}{
            \begin{tabular}{c c c c}
                 $
                 \infer[^{2\ (\mathit{ind})}]
                 {
                    \phi \to [p^*]\phi
                 }
                 {
                    {\forall (\phi \to [p]\phi)}
                 }
                 $
                 &
                 $
                 \infer[^{3\ (\mathit{con})}]
                 {
                    (\exists v\ge 0.\phi(v)) \to \la p^*\ra \phi(0)
                 }
                 {
                    \forall ((v > 0\wedge \phi(v)) \to \la p\ra\phi(v-1))
                 }
                 $
                 &
                 $
                 \infer[^{([],\mathit{gen})}]
                 {
                    [p]\phi\to [p]\psi
                 }
                 {
                    \forall(\phi\to \psi)
                 }
                $
                &
                $
                 \infer[^{(\la \ra,\mathit{gen})}]
                 {
                    \la p\ra\phi\to \la p\ra\psi
                 }
                 {
                    \forall(\phi\to \psi)
                 }
                $
            \end{tabular}
         }
         \\
         \multicolumn{1}{|c|}{
                {\footnotesize (b) rules from other dynamic logics}
            }
         \ifx
         \\
         \\
         \multicolumn{1}{|c|}{
            \begin{tabular}{c c}
                $
                 \infer[^{([],\mathit{gen})}]
                 {
                    [p]\phi\to [p]\psi
                 }
                 {
                    \forall(\phi\to \psi)
                 }
                $
                &
                $
                 \infer[^{(\la \ra,\mathit{gen})}]
                 {
                    \la p\ra\phi\to \la p\ra\psi
                 }
                 {
                    \forall(\phi\to \psi)
                 }
                $
            \end{tabular}
         }
         \fi
         \\
         \midrule
         \multicolumn{1}{|l|}{
            \begin{tabular}{l}
            $^{1}$ $\xi\in \{\phi, \Box\phi\}$\\
            $^{2}$ $\forall (\phi)\dddef \forall x_1.\forall x_2.....\forall x_n. \phi$, where $x_1,...,x_n$ are the set of all free variables in $\phi$\\
            $^{3}$ The variable $v$ does not appear in $p$
            \end{tabular}
         }
         \\
         \bottomrule
          \end{tabular}
              }
              }
          \end{center}
          \caption{Rules for Sequential SPs}
          \label{table:Rules for Closed SPs}
    \end{table}

Except for the rules $([],\mathit{gen})$ and $(\la\ra, \mathit{gen})$, all rules of type $(b)$ are for composite programs.
The rules $(\seq, \Box\phi)$ and $(*, \Box\phi)$ are inherited from DTDL~\cite{Harel00}.
The rules $(\seq,\phi)$ and $(\seq,\Box\phi)$ are due to the fact that any trace of $p\seq q$ is formed by concatenating a trace of $p$ and a trace $q$,
while rule $(\cup)$ is based on the fact that any trace of $p\cup q$ is either a trace of $p$ or a trace of $q$.
Rule $(*, \Box\phi)$ converts the proof of a temporal formula $\Box\phi$ to the proof of a state formula $[p]\Box\phi$.
This is useful because then we can apply other rules such as $(\mathit{ind})$ and $(\mathit{con})$ which are only designed for a state formula $\phi$.
Rule $(*)$ is based on the semantics of $p^*$, which means that either does not execute $p$ (i.e. $\noth$), or executes $p$ for 1 or more than 1 times (i.e. $p;p^*$).
The rules $([], \mathit{gen})$ and $(\la\ra, \mathit{gen})$ are for eliminating the dynamic parts $[p]$, $\la p\ra$ of a formula during deductions.
They are used in deriving the rules $([*])$ and $(\la *\ra)$ below and are necessary for the relatively completeness of the whole proof system.

\begin{equation}
\infer[^{([*])}]
{
    [p^*]\phi
}
{
    \psi
    &
    \forall(\psi \to [p]\psi)
    &
    \forall(\psi\to \phi)
}
\end{equation}

\begin{equation}
\infer[^{(\la *\ra)}]
{
    \la p^*\ra \phi
}
{
    \exists v\ge 0. \psi(v)
    &
    \forall((v>0\wedge \psi(v))\to \la p\ra \psi(v-1))
    &
    \forall(\psi(0) \to \phi)
}
\end{equation}

The rules $(\mathit{ind})$ and $(\mathit{con})$ are mathematical inductions for proving properties of star programs $p^*$.
Rule $(\mathit{ind})$ means that to prove that $\phi\to [p^*]\phi$ holds at a state, we need to prove that $\phi\to [p]\phi$ holds at any state.
Rule $(\mathit{con})$ has a similar meaning as $(\mathit{ind})$.
The main difference is that in $(\mathit{con})$ a variable $v$ is introduced as an indication of the termination of $p^*$.
The rules $(\mathit{ind})$ and $(\mathit{con})$ are mainly used in theories.
In practical verification, the rules $([*])$ and $(\la *\ra)$ above are applied for eliminating the star operator $*$, where $\psi$ is often known as a \emph{loop invariant} of $p^*$.
$([*])$ and $(\la *\ra)$ can be derived from $(\mathit{ind})$ and $(\mathit{con})$, refer to~\cite{Harel00} for more details.

\subsection{Rewrite Rules for Parallel SPs}
\label{section:Rewrite Rules for Parallel SPs}

Table~\ref{table:Rules for Parallel SPs} gives the rewrite rules for parallel SPs, which are divided into three types and one single rule.
The rules of the types $(a), (b), (c)$ perform rewrites for one step, whereas rule $(\para, \mathit{seq})$ rewrites a parallel program as a whole into a sequential one through a series of steps in the algorithm $\ToSeq$.

\begin{table}[htpb]
    \noindent\makebox[\textwidth]{%
    \scalebox{1.0}{
        \begin{tabular}{|c|}
        \toprule
            \multicolumn{1}{|c|}{
                \begin{tabular}{l l}
                    $\begin{aligned}
                        \infer=[]
                        {\phi\{p\}}
                        {\phi\{q\}}
                    \end{aligned}\mbox{ if } p\red q
                    \
                    ^{1}
                    $
                    \ {\footnotesize$(r1)$}
                    &
                    $
                    p\{q\}\red p\{r\}
                    \mbox{ if } q\red r
                    \
                    ^{2}
                    $
                    \ {\footnotesize$(r2)$}
                \end{tabular}
            }
            \\
            \multicolumn{1}{|c|}{
                {\footnotesize (a) structural rewrite rules}
            }
            \\
            \midrule
            \multicolumn{1}{|c|}{
                \begin{tabular}{l l l}
                $\begin{aligned}
                        \noth\seq p\red p
                \end{aligned}
                $
                    \ {\footnotesize$(\noth,\seq)$}
                &
                $\begin{aligned}
                        \noth^*\red \noth
                    \end{aligned}
                $
                    \ {\footnotesize$(\noth,*)$}
                &
                $\begin{aligned}
                        \halt\seq p\red \halt
                    \end{aligned}
                $
                \ {\footnotesize $(\halt, \seq)$}
                \\
                $\begin{aligned}
                        \halt^*\red \noth
                    \end{aligned}
                    $
                    \ {\footnotesize$(\halt,*)$}
                &
                $\begin{aligned}
                        (p\seq q)\seq r\red p\seq (q\seq r)
                    \end{aligned}
                $
                    \ {\footnotesize$(\seq,\mathit{ass})$}
                &
                $\begin{aligned}
                        p\seq (q\cup r)\red (p\seq q)\cup (p\seq r)
                    \end{aligned}
                    $
                    \ {\footnotesize$(\seq,\mathit{dis}1)$}
                \\
                $\begin{aligned}
                        (q\cup r)\seq p\red (q\seq p)\cup (r\seq p)
                    \end{aligned}
                    $
                    \ {\footnotesize$(\seq,\mathit{dis}2)$}
                &
                $\begin{aligned}
                        (p\cup q)\cup r\red p\cup (q\cup r)
                    \end{aligned}
                    $
                    \ {\footnotesize$(\cup,\mathit{ass})$}
                &
                $\begin{aligned}
                        p^*\red \noth\cup p;p^*
                    \end{aligned}
                    $
                    \ {\footnotesize$(*,\mathit{exp})$}
                \end{tabular}
            }
            \\
            \multicolumn{1}{|c|}{
                {\footnotesize (b) rewrite rules for sequential programs}
            }
            \\
            \midrule

            \multicolumn{1}{|c|}{
                \begin{tabular}{l l l}
                    $\begin{aligned}
                        \para(...,p,\noth,q,...)\red \para(...,p,q,...)
                    \end{aligned}
                    $
                    \ {\footnotesize$(\para,\noth)$}
                    &
                    $\begin{aligned}
                        \para(...,p, \halt, q,...)\red \halt
                    \end{aligned}
                    $
                    \ {\footnotesize$(\para,\halt)$}
                    &
                    $\begin{aligned}
                        \para(...,p\cup q,...)\red \para(...,p,...)\cup \para(...,q,...)
                    \end{aligned}
                    $
                    \ {\footnotesize$(\para,\mathit{dis})$}
                \end{tabular}
            }
            \\
            \multicolumn{1}{|c|}{
                \begin{tabular}{l}
                    $\begin{aligned}
                        \para(\alpha_1\seq q_1,...,\alpha_n\seq q_n)\red \alpha\seq \para(q'_1,...,q'_n)
                    \end{aligned}
                    \mbox{ if $(b, \alpha, (q'_1\sep...\sep q'_n)) = \Merge(\alpha_1\seq q_1\sep...\sep \alpha_n\seq q_n)$ and $b = \true$}
                    $
                    \ {\footnotesize$(\para,\mathit{mer})$}
                    \\

                \end{tabular}
            }
            \\
        \multicolumn{1}{|c|}{
                {\footnotesize (c) rewrite rules for parallel programs}
            }
            \\
        \midrule
        \multicolumn{1}{|c|}{
                \begin{tabular}{l}
                    $\begin{aligned}
                        \para(p_1,...,p_n)\red \ToSeq(\para(p_1,...,p_n))
                    \end{aligned}
                    \mbox{ if $\para(p_1,...,p_n)$ is well defined}
                    $
                    \ {\footnotesize$(\para,\mathit{seq})$}
                \end{tabular}
            }
        \\
        \midrule
        \multicolumn{1}{|l|}{
            \begin{tabular}{l}
            $^{1}$ $p$ and $q$ are closed programs; $\phi\{\place\}$ is a program hole of the formula $\phi$\\
            $^{2}$ $p\{\place\}$ is a program hole of the program $p$; the reduction $q\red r$ is from the rules of the types $(b)$ and $(c)$
            \end{tabular}
         }
         \\
        \bottomrule
        \end{tabular}
    }
    }
    \caption{Rewrite Rules for Parallel SPs}
    \label{table:Rules for Parallel SPs}
\end{table}

Two rules of type $(a)$ are structural rules for rewriting an SDL formula or an SP according to the rewrites of its parts.
They rely on the definition of \emph{program holes} given below.
\begin{mydef}[Program Holes]
    Given a formula $\phi$, a program hole of $\phi$, denoted as $\phi\{\place\}$, is defined in the following grammar:
    $$\begin{aligned}
    \phi\{\place\}\ddef&\ \neg \phi\{\place\}\ |\ \phi\{\place\}\wedge \phi\ |\ \phi\wedge \phi\{\place\}\ |\ \forall x.\phi\{\place\}\ |\ [p\{\place\}]\phi\ |\ [p\{\place\}]\Box\phi,
    \end{aligned}
    $$
    where a program hole $p\{\place\}$ of an SP $p$ is defined as:
    $$
    \begin{aligned}
    p\{\place\}\ddef&\ \place\ |\ p\{\place\}\seq p\ |\ p\seq p\{\place\}\ |\ p\{\place\}\cup p\ |\ p\cup p\{\place\}\ |\ (p\{\place\})^*\ |\ \para(p,...,p\{\place\},...,p).
    \end{aligned}$$
    We call $\place$ a place which can be filled by a program.

    Given a program hole $\phi\{\place\}$ (resp. $p\{\place\}$) and an SP $q$, $\phi\{q\}$ (resp. $p\{q\}$) is the formula (resp. the program) obtained by filling the place $\place$ of $\phi\{\place\}$ (resp. $p\{\place\}$) with $q$.

\end{mydef}
The rules $(\mathit{r1})$ and $(\mathit{r2})$ say that the rewrites between SPs preserves the semantics of SDL formulas and SPs.
Rule $(\mathit{r1})$ means that if $p$ can be rewritten by $q$, then we can replace $p$ by $q$ anywhere of $\phi$ without changing the semantics of $\phi$.
Rule $(\mathit{r2})$ has a similar meaning.
Note that $p\{q\}$, $p\{r\}$, $q, r$ can also be open SPs, and in this case we do not care about their semantics (because they do not have one).

The rewrites in the rules of type $(a)$ come from the rules of the types $(b)$ and $(c)$.
The rules of type $(b)$ are for sequential programs.
They are all based on the semantics of sequential SPs.
For example, rule $(\noth, \seq)$ is based on the fact that $\val(\noth\seq p) = \val(\noth)\circ \val(p) = \val(p)$.
The rules of type $(c)$ are for parallel programs.
They are based on the semantics of parallel programs given in Sect.~\ref{section:Valuation of Parallel Programs}.
See Appendix~\ref{section:Proofs for some Propositions and the Soundness of SDL} for the proofs of their soundness.

\ifx
To make the proof of the soundness of the rule $(r2)$ easier later in Appendix~\ref{???}, we stipulate that in this rule the reduction $q\red r$ cannot come from the rule $(\para, \mathit{seq})$,
in other words, the rule $(\para, \mathit{seq})$ can be applied to only a program as a whole, not the part of a program.
This restriction prevents that the rules $(r2)$ and $(\para, \mathit{seq})$ appear in each other's body.
However, it does not affect the relative completeness of SDL because the relative completeness of SDL only relates to the rewrite rule $(\para, \mathit{seq})$ where
$\ToSeq$ does not require to use the rule $(\para, \mathit{seq})$ in its procedure according to Algorithm~\ref{alg:ToSeq} introduced below.
\fi

In a parallel program $\para(p_1,...,p_n)$, if the program $p_i$ ($1\le i\le n$) is in the form $q^*$ or $q^*\seq r$, consecutively applying the rules of the types (a), (b) and (c) may generate an infinite proof tree.
An example is shown in Fig.~\ref{figure:An example of inifinite proof trees}, where we merge several derivation steps into one by listing all the rules applied during these steps.
In the proof tree of this example, we can see that the node $(1)$ is the same as the root node so the proof procedure will never end.

\begin{figure}[htpb]
	\centering
\scalebox{1}{
	\infer[^{(*, \mathit{exp}), (\mathit{r1})}]
    {
        (\sig\nex \epsilon)^*\para(\hat{\sig}?\nex \epsilon)^*
    }
    {
        \infer[^{(\cap,\mathit{dis}), (\mathit{r1}), (\cup)}]
        {
            (\noth\cup (\sig\nex \epsilon)\seq (\sig\nex \epsilon)^*)
            \para
            (\noth\cup (\hat{\sig}?\nex \epsilon)\seq (\hat{\sig}?\nex \epsilon)^*)
        }
        {
            \infer[]
            {
                \noth\para \noth
            }
            {
                ...
            }
            &
            \infer[]
            {
                \noth\para (\hat{\sig}?\nex \epsilon)\seq (\hat{\sig}?\nex \epsilon)^*
            }
            {
                ...
            }
            &
            \infer[]
            {
                (\sig\nex \epsilon)\seq (\sig\nex \epsilon)^*\cap \noth
            }
            {
                ...
            }
            &
            \infer[^{(\para, \mathit{mer})}]
            {
                (\sig\nex \epsilon)\seq (\sig\nex \epsilon)^* \para (\hat{\sig}?\nex \epsilon)\seq (\hat{\sig}?\nex \epsilon)^*
            }
            {
                \infer[^{(\seq,\phi), (\epsilon)}]
                {
                    \epsilon\seq ((\sig\nex \epsilon)^* \para (\hat{\sig}?\nex \epsilon)^*)
                }
                {
                    (1):\ \ \ (\sig\nex \epsilon)^* \para (\hat{\sig}?\nex \epsilon)^*
                }
            }
        }
    }
} 
\caption{An Example of Inifinite Proof Trees}
\label{figure:An example of inifinite proof trees}
\end{figure}

To avoid this situation, we propose rule $(\para, \mathit{seq})$.
Rule $(\para, \mathit{seq})$ reduces a parallel program into a sequential one by the procedure $\ToSeq$ (Algorithm~\ref{alg:ToSeq}), which
follows the Brzozowski's method~\cite{Brzozowski64} that transforms an NFA (non-deterministic finite automaton), expressed as a set of equations, into a regular expression.
This process  relies on Arden's rule~\cite{Arden61} to solve the equations of regular expressions.
Before introducing the procedure $\ToSeq$, we firstly explain why Arden's rule can apply to SPs.

In fact, it is easily seen that a sequential SP is a regular expression, whose semantics is exactly the set of sequences the regular expression represents.
The sequence operator $\seq$ represents the \emph{concatenation} between sequences of two sets.
The choice operator $\cup$ represents the \emph{union} of two sets of sequences.
The loop operator $*$ represents the \emph{star} operator that is applied to a set of sequences.
An event $\alpha$ is a \emph{word} of a regular expression, representing a set of sequences with the minimal length (here in SP the minimal length is $2$).
The empty program $\noth$ is the \emph{empty string} of a regular expression, representing a set of sequences with ``zero length'' , i.e., sequences that do not change other sequences when are concatenated to them (here in SP, the ``zero length'' is 1 due to the definition of the operator $\circ$). 
The halt program $\halt$ is the \emph{empty set} of a regular expression, representing an empty set of sequences.

With this fact, Arden's rule obviously holds for sequential SPs, as stated as the following proposition.
In the below of this paper, we sometimes use $p \equiv q$ to mean that two programs $p$ and $q$ are semantically equivalent, i.e., $\val(p) = \val(q)$.

\begin{prop}[Arden's Rule in SPs]
    \label{prop:Arden's Rule in SPs}
    Given any sequential SPs $p$ and $q$ with $q\not\equiv \noth$, $X \equiv q^*\seq p$ is the unique solution of the equation $X\equiv p\cup q\seq X$.
\end{prop}

Prop.~\ref{prop:Arden's Rule in SPs} can be proved according to the semantics of SPs, we omit it here.

\ifx
It is easy to see that a sequential SP is in fact a regular expression, with $\alpha$ being a \emph{word}, $\noth$ being the \emph{empty string}, $\halt$ being the \emph{empty language}, $\seq$ being the \emph{concatenation} operator,
$\cup$ being the \emph{union} operator and $*$ being the \emph{Kleene star} operator.
The semantic equivalence in SPs is exactly the language equivalence in a regular expression, where the language equivalence means that two regular expressions represent a same set of strings.
\fi

    \begin{algorithm}[ptb]
    \begin{algorithmic}[1]
    \Procedure{Brz}{$\para(p_1,...,p_n)$} /*$\para(p_1,...,p_n)$ is well defined*/
    \State let $l_1 = \para(p_1,...,p_n)$, by consecutively applying the rules $(r2)$, and the rules of the types (b) and (c) in Table~\ref{table:Rules for Parallel SPs},
    we can reduce $l_1$ as the following form:
\begin{equation}
\label{equ:ToSeq}
    l_1\red b_1\cup \alpha_{11}\seq l_1\cup ...\cup \alpha_{1n}\seq l_n,
\end{equation}
    where $n\ge 0$ ($l_1\red b_1$ when $n = 0$), $b_1$ is a sequential program, $l_2,...,l_n$ are parallel programs.
    From the reduction above we can build an equation, namely
    $$
    l_1\equiv b_1\cup \alpha_{11}\seq l_1\cup ...\cup \alpha_{1n}\seq l_n.
    $$

    \State
    continuing this reduction procedure for the programs $l_2,..., l_n$, we can then build $n$ equations:
    $$
    \begin{aligned}
    l_1\equiv&\ b_1\cup \alpha_{11}\seq l_1\cup ...\cup \alpha_{1n}\seq l_n, &(1)\\
    ...\\
    l_n\equiv&\ b_n\cup \alpha_{n1}\seq l_1\cup...\cup \alpha_{nn}\seq l_n, &(n)
    \end{aligned}
    $$
    where $b_1,...,b_n$ are sequential programs, $l_1,...,l_n$ are taken as $n$ variables.

    \For{each $k$, $k = n, n-1, ..., 2, 1$}
        \State transform the equation $(k)$ into the form $l_k\equiv p\cup q\seq l_k$
        \State by Prop.~\ref{prop:Arden's Rule in SPs}, obtain $l_k\equiv q^*\seq p$ from $l_k\equiv p\cup q\seq l_k$
        \State substitute $l_k$ on the right of the other equations $(k-1)$,..., $(1)$ with $q^*\seq p$
    \EndFor
    \State \textbf{return} $l_1$
    \EndProcedure
    \end{algorithmic}
    \captionsetup{font=footnotesize}
    \caption{\footnotesize{Procedure $\ToSeq$}}
    \label{alg:ToSeq}
    \end{algorithm}

    The procedure $\ToSeq$ transforms a parallel program into a sequential one.
    Algorithm~\ref{alg:ToSeq} explains how it works.
    At line 2, it is easy to see that by using the rules $(r2)$ and the rules of the types (b) and (c) in Table~\ref{table:Rules for Parallel SPs},
    a parallel program can always be reduced into a form as the righthand side of the reduction (\ref{equ:ToSeq}).
    We can actually prove it by induction on the syntactic structure of the parallel program.
    At line 3, we can always build a finite number of $n$ equations because whichever rule we use for reducing a parallel program $\para(p_1,...,p_n)$, e.g. $\para(p_1,...,p_n)\red q$,
    the reduced form $q$ only contains the symbols in the original form $\para(p_1,...,p_n)$.
    Therefore, the total number of reduced expressions from a parallel program $\para(p_1,...,p_n)$ by consecutively applying those rules must be finite because $\para(p_1,...,p_n)$ only contains a finite number of symbols.

    By taking all parallel programs $l_1,...,l_n$ as variables, the process of solving the $n$ equations is attributed to Brzozowski's standard process of solving $n$ regular equations,
    as stated from line 4 - line 8.
    On each iteration, Arden's rule is applied to eliminate the $k$th variable $l_k$.
    Finally, all $n-1$ variables $l_{n}$,...,$l_2$ are eliminated and we return $l_1$ as the result.

\subsection{Rules for FOL Formulas}
\label{section:Rules for FOL formulas}

Table~\ref{table:Rules of FOL} shows the rules for FOL formulas.
Since they are all common in FOL, we omit the discussion of them here.
For convenience of analysis in the rest of the paper, we also give the rules for $\vee, \to$ and $\exists$ in Table~\ref{table:Other Rules of FOL}, although they can be derived by the rules in Table~\ref{table:Rules of FOL}.

    \begin{table}[htpb]
         \begin{center}
         \scalebox{1.0}{
         \begin{tabular}{|c|}
         \toprule
         \multicolumn{1}{|c|}{
            \begin{tabular}{c c c c}
                $
                \infer[^{(\mathit{ax})}]
                {
                    \Gamma, \phi\seqArrow \phi, \Delta
                }
                {}
                $
                &
                $
                \infer[^{(\mathit{cut})}]
                {
                    \Gamma\seqArrow \Delta
                }
                {
                    \Gamma\seqArrow \phi, \Delta
                    &
                    \Gamma, \phi\seqArrow\Delta
                }
                $
                &
                $
                \infer[^{(\neg r)}]
                {
                    \Gamma\seqArrow \phi,\Delta
                }
                {
                    \Gamma, \neg\phi\seqArrow \Delta
                }
                $
                &
                $
                \infer[^{(\neg l)}]
                {
                    \Gamma, \phi\seqArrow \Delta
                }
                {
                    \Gamma \seqArrow \neg\phi, \Delta
                }
                $
            \end{tabular}
            }
            \\
            \\
            \multicolumn{1}{|c|}{
                \begin{tabular}{c c c c}
                    $
                    \infer[^{(\wedge r)}]
                    {
                        \Gamma \seqArrow \phi\wedge \psi, \Delta
                    }
                    {
                        \Gamma \seqArrow \phi, \Delta
                        &
                        \Gamma \seqArrow \psi, \Delta
                    }
                    $
                    &
                    $
                    \infer[^{(\wedge l)}]
                    {
                        \Gamma, \phi\wedge \psi\seqArrow \Delta
                    }
                    {
                        \Gamma, \phi, \psi\seqArrow \Delta
                    }
                    $
                    &
                    $
                    \infer[^{1\ (\forall r)}]
                    {
                        \Gamma \seqArrow \forall x.\phi, \Delta
                    }
                    {
                        \Gamma \seqArrow \phi[y/x], \Delta
                    }
                    $
                    &
                    \infer[^{(\forall l)}]
                    {
                        \Gamma, \forall x.\phi\seqArrow \Delta
                    }
                    {
                        \Gamma, \phi[e/x]\seqArrow \Delta
                    }
                \end{tabular}
            }
            \\
            \midrule
            \multicolumn{1}{|l|}{
                \begin{tabular}{l}
                    $^{1}$ $y$ is a new variable with respect to $\Gamma$, $\phi$ and $\Delta$
                \end{tabular}
            }
            \\
            \bottomrule
         \end{tabular}
         }
         \end{center}
         \caption{Rules of FOL}
         \label{table:Rules of FOL}
    \end{table}

    \begin{table}[htpb]
         \begin{center}
         \scalebox{1.0}{
         \begin{tabular}{|c|}
         \toprule
         \multicolumn{1}{|c|}{
            \begin{tabular}{c c c c}
                $
                \infer[^{(\vee r)}]
                {
                    \Gamma \seqArrow \phi\vee \psi, \Delta
                }
                {
                    \Gamma \seqArrow \phi, \psi, \Delta
                }
                $
                &
                $
                \infer[^{(\vee l)}]
                {
                    \Gamma, \phi\vee \psi\seqArrow \Delta
                }
                {
                    \Gamma, \phi\seqArrow \Delta
                    &
                    \Gamma, \psi\seqArrow \Delta
                }
                $
                &
                $
                \infer[^{(\to r)}]
                {
                    \Gamma\seqArrow \phi\to \psi,\Delta
                }
                {
                    \Gamma, \phi\seqArrow \psi\Delta
                }
                $
                &
                $
                \infer[^{(\to l)}]
                {
                    \Gamma, \phi\to \psi\seqArrow \Delta
                }
                {
                    \Gamma \seqArrow \phi, \Delta
                    &
                    \Gamma, \psi\seqArrow \Delta
                }
                $
            \end{tabular}
            }
            \\
            \\
            \multicolumn{1}{|c|}{
                \begin{tabular}{c c}
                    $
                    \infer[^{(\exists r)}]
                    {
                        \Gamma \seqArrow \exists x.\phi, \Delta
                    }
                    {
                        \Gamma \seqArrow \phi[e/x], \Delta
                    }
                    $
                    &
                    \infer[^{1\ (\exists l)}]
                    {
                        \Gamma, \exists x.\phi\seqArrow \Delta
                    }
                    {
                        \Gamma, \phi[y/x]\seqArrow \Delta
                    }
                \end{tabular}
            }
            \\
            \midrule
            \multicolumn{1}{|l|}{
                \begin{tabular}{l}
                    $^{1}$ $y$ is a new variable with respect to $\Gamma$, $\phi$ and $\Delta$
                \end{tabular}
            }
            \\
            \bottomrule
         \end{tabular}
         }
         \end{center}
         \caption{Other Rules of FOL}
         \label{table:Other Rules of FOL}
    \end{table}

At the end of this section, we introduce the notion of \emph{deductive relation} in SDL.
\begin{mydef}[Deductive Relation $\vdash$]
    Given a SDL formula $\phi$ and a multi-set of formulas $\Phi$,
    we say $\phi$ is derived by $\Phi$, denoted as $\Phi\vdash \phi$, iff the sequent $\Phi\seqArrow \phi$ can be derived by the rules in Table~\ref{table:Rules for Closed SPs}, \ref{table:Rules for Parallel SPs}, \ref{table:Rules of FOL}.

    If $\Phi$ is an empty set, we simply denote $\Phi\vdash \phi$ as $\vdash \phi$.
\end{mydef}

\section{Soundness and Relative Completeness of SDL Calculus}
\label{section:Soundness and Relative Completeness of SDL Calculus}

\subsection{Soundness of SDL Calculus}
The soundness of the proof system of SDL is stated as the following theorem.

\begin{mytheo}[Soundness of SDL]
\label{theorem:Soundness of SDL}
Given an SDL formula $\phi$, if $\vdash \phi$, then $\models \phi$.
\end{mytheo}

To prove Theorem~\ref{theorem:Soundness of SDL}, it is equivalent to prove that each rule in Table~\ref{table:Rules for Closed SPs}, \ref{table:Rules for Parallel SPs}, \ref{table:Rules of FOL} is sound.

For the rules in Table~\ref{table:Rules for Closed SPs}, we only need to prove the soundness of the rules of type $(a)$, since all rules of type $(b)$ are inherited from FODL and DTDL.
Based on the semantics of SDL formulas, the soundness of the rules of type $(a)$ is given in Appendix~\ref{section:Proofs for some Propositions and the Soundness of SDL}.
For the proofs of the soundness of the rules $(\seq, \Box\phi)$ and $(*, \Box\phi)$, one can refer to~\cite{Platzer07}.
For the proofs of the soundness of other rules of type $(b)$, one can refer to~\cite{Harel00}.

The soundness of a rewrite rule means that given two closed SPs $p$ and $q$, if $p\red q$, then $\val(p) = \val(q)$.
In Appendix~\ref{section:Proofs for some Propositions and the Soundness of SDL}, we prove the soundness of all rewrite rules of Table.~\ref{table:Rules for Parallel SPs}.
We firstly prove the soundness of the rules of the types $(b)$ and $(c)$ (see Prop.~\ref{prop:Soundness of the Rewrite Rules for Closed Sequential Programs} and \ref{prop:Soundness of the Rewrite Rules for Parallel Programs}) based on their semantics,
then we prove the soundness of the rules $(\mathit{r1})$ and $(\mathit{r2})$ (see Prop.~\ref{prop:Soundness of the Rewrite Rule r1} and \ref{prop:Soundness of the Rewrite Rule r2}) by induction on the structures of the program holes and the level of a rewrite relation (as defined in Def.~\ref{def:Level of a Rewrite Relation}).
The soundness of rule $(\para, \mathit{seq})$ is direct according to Arden's rule (see Prop.~\ref{prop:Soundness of the Rewrite Rule of (para, seq)}).

The soundness of the FOL rules in Table~\ref{table:Rules of FOL} and \ref{table:Other Rules of FOL} are directly from FOL.

\subsection{Relative Completeness of SDL Calculus}
Since SDL includes AFOL in itself, due to G\"odel's incompleteness theorem~\cite{Godel31}, SDL is not complete.
In the following, we consider the relative completeness~\cite{Cook78} of SDL.

Given an SDL formula $\phi$, we use $\reld \phi$ to represent that $\phi$ is derivable in the proof system consisting of all rules of Table~\ref{table:Rules for Closed SPs}, \ref{table:Rules for Parallel SPs}, \ref{table:Rules of FOL} and
all tautologies in AFOL as axioms. Intuitively, $\reld \phi$ means that $\phi$ can be transformed into a set of pure AFOL formulas in the proof system of SDL.

The relative completeness of SDL is stated as the following theorem.

\begin{mytheo}[Relative Completeness of SDL]
\label{theorem:Relative Completeness of SDL}
    Given an SDL formula $\phi$, if $\models \phi$, then $\reld\phi$.
\end{mytheo}

The main idea behind the proof of Theorem~\ref{theorem:Relative Completeness of SDL} follows the proof of FODL initially proposed in~\cite{Harel79}, where a theorem called ``the main theorem'' (Theorem 3.1 of \cite{Harel79}) was proposed as the skeleton of the whole proof.
We give our proof by augmenting that main theorem with one more condition for the temporal dynamic formulas $[p]\Box\phi$, which are new in SDL.
Our main theorem is stated as follows.

\begin{mytheo}[The ``Main Theorem'']
\label{theorem:The ``Main Theorem''}
    SDL is relatively complete if the following conditions hold:
    \begin{enumerate}[(i)]
        \item SDL formulas are expressible in AFOL.
        \item For any AFOL formulas $\phi^\flat$ and $\psi^\flat$, if $\models \phi^\flat \to \mathop{op}\psi^\flat$, then $\reld \phi^\flat\to \mathop{op}\psi^\flat$.
        \item For any SDL formulas $\phi$ and $\psi$, if $\reld \phi\to \psi$, then $\reld \mathop{op} \phi\to \mathop{op} \psi$.
        \item For any AFOL formulas $\phi^\flat$ and $\psi^\flat$, if $\models \phi^\flat\to [p]\Box\psi^\flat$, then $\reld \phi^\flat\to [p]\Box\psi^\flat$; and
        if $\models \phi^\flat\to \la p\ra \Diamond\psi^\flat$, then $\reld\phi^\flat\to \la p\ra\Diamond\psi^\flat$.
    \end{enumerate}
    In the above conditions,  $\mathop{op}\in \{[p], \la p\ra, \forall x, \exists x\}$.
    We use $\phi^\flat$ to stress that $\phi$ is an AFOL formula.
\end{mytheo}

The only differences between Theorem~\ref{theorem:Relative Completeness of SDL} and the main theorem in~\cite{Harel79} are:
1) we replace all ``FODL formulas'' with ``SDL formulas'' in the context;
2) we add a new condition (the condition \rmn{4}) to our theorem for temporal dynamic formula $[p]\Box\phi$ and its dual form $\la p\ra \Diamond\phi$.
With Theorem~\ref{theorem:The ``Main Theorem''}, we can prove the relative completeness of SDL based on these 4 conditions (See Appendix~\ref{section:Proofs for the Relative Completeness of SDL}).
During the process of the proof the new-adding condition \rmn{4} plays a central role when proving the relative completeness of SDL formulas of the forms $[p]\Box\phi$ and $\la p\ra\Diamond\phi$.

The rest remains to prove the 4 conditions of Theorem~\ref{theorem:The ``Main Theorem''}.
For the condition \rmn{1}, the expressibility of SDL formulas in AFOL is directly from that of FODL formulas because, intuitively, a sequential program in SDL can be seen as a regular program in FODL if we ignore the differences between macro steps and micro steps which only play their roles in parallel SPs.
The proofs of the conditions \rmn{2} and \rmn{3} mainly follow the corresponding proofs of FODL in~\cite{Harel79}, where the main difference is that in the proof steps we also need to consider the cases when an SP is a special primitive (such as the macro event) and when an SP is a parallel program.
The condition \rmn{4} mainly follows the idea behind the proof of the condition $\rmn{2}$ but is adapted to fit the proofs of the temporal dynamic formulas $[p]\Box\phi$ and $\la p\ra\Diamond\phi$.
Appendix~\ref{section:Proofs for the Relative Completeness of SDL} gives the proofs of these 4 conditions.

\section{SDL in Specification and Verification of SyncCharts}
\label{section:SDL in Specification and Verification of SyncCharts}
In this section, we show that SDL is a useful framework in specification and verification of synchronous models.
We show that synchronous models and their properties can be specified by SDL in a natural way and verified in the proof system of SDL in a compositional way. 
Among many synchronous models, in this paper, we choose SyncChart~\cite{Andre03} as an example.
SyncChart has the same semantics as Esterel and has been embedded into the famous industrial tool SCADE.
SyncChart captures the most essential features of synchronous models and is in the form of automata, which is easy to be transformed into SPs.


The example proposed here is not intended to show the applicability of SDL in practical verification of reactive systems, but to show the connection between SP and the synchronous models used in practice and the potential of SDL as a verification framework.

\subsection{Encoding Basic SyncCharts}
\label{section:Encoding Basic SyncCharts}
In this subsection, we use SPs to specify basic syncCharts.
As we will see, the process of encoding a basic syncChart as an SP is quite straightforward as SP models can rightly capture the features of synchronous models.
In the following, we give two toy examples to show SPs can encode syncCharts.
The first example is about a sequential system, while the second one is about a concurrent system.
Note that we only give examples and do not intend to give a general algorithm for the automatic encoding, which is out of the scope of this paper.

Currently, SPs only support encoding the basic syncCharts, which do not include advanced synchronous features such as preemption, hierarchy and so on (cf.~\cite{Andre03}).
However, in terms of expressiveness, considering the basic syncCharts is enough because those advanced features essentially do not inhence the expressive power of syncCharts.
They are just for engineers to model in a more convenient and neater way.


For the first example, we consider a simple circuit as considered in~\cite{Andre03}, called a ``frequency divider'', which is modelled as an syncChart, named ``FDIV2'', shown in Fig.~\ref{figure:A Frequency Divider}.
The frequency dividor waits for a first occurrence of a signal $T$, and then emits a signal $C$ at every other occurrence of $T$.
Its syncChart FDIV2 consists of two states, with the $\mathit{off}$ state as the initial state.
Each transition in a syncChart represents one reaction of a reactive system, i.e., an instant at which all events occur simultaneously in a logical order.
It exactly corresponds to a macroevent in SPs.
The label on a transition is called the ``trigger and effect'', which is in the form of $\mathit{trigger}/\mathit{effect}$, representing the actions of reading input signals and sending output signals respectively at each instant.
For example, in FDIV2, the label ``$T/C$'' means that at an instant, if the signal $T$ is triggered, the signal $C$ is emitted.
A label corresponds to an event in SPs.

\begin{figure}[htpb]
        \centering
        \includegraphics[width=0.8\linewidth]{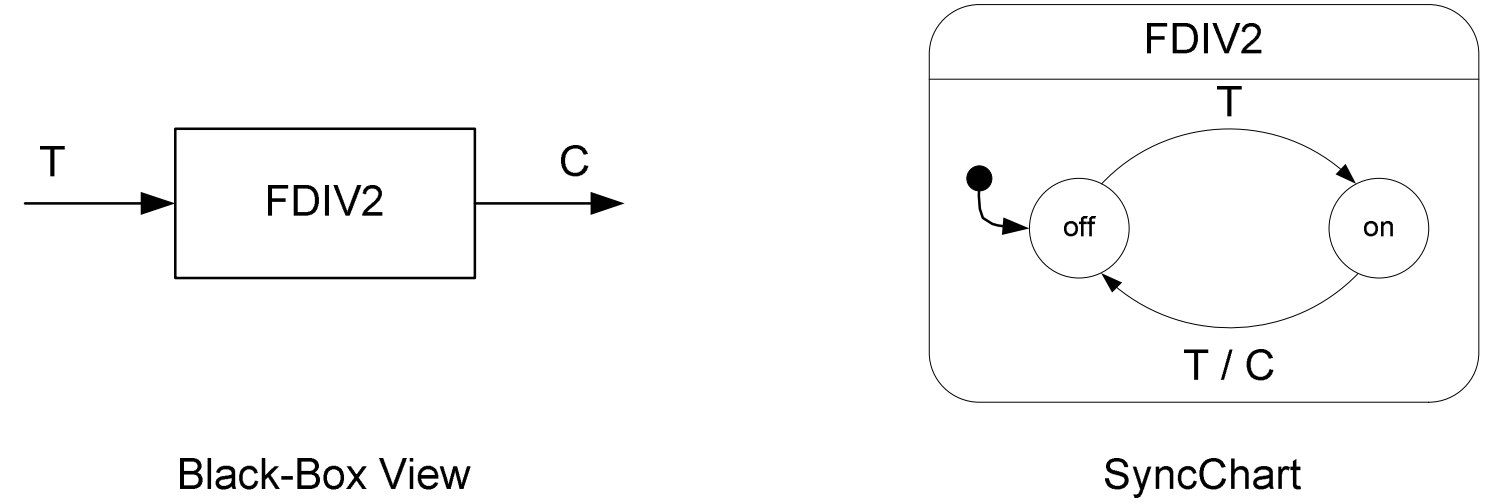}
        \caption{A Frequency Divider}
        \label{figure:A Frequency Divider}
    \end{figure}

The behaviour of FDChart is that
\begin{enumerate}[(1)]
    \item At the state $\mathit{off}$, the syncChart waits for a signal $T$ and moves to the state $\mathit{on}$;
    \item at the state $\mathit{on}$, it waits for a signal $T$ and emits a signal $C$ at the same instant, and moves to the state $\mathit{off}$.
\end{enumerate}
In SDL, let $\sig_t$, $\sig_c$ represent the signals $T$ and $C$ respectively, we can encode FDChart as an open SP
$$
P_{\mathit{FD}} = ((\bar{\sig_t}?\nex \epsilon)^*\seq \hat{\sig_t}?\nex \epsilon\seq (\bar{\sig_t}?\nex \epsilon)^*\seq \hat{\sig_t}?\nex \sig_c\nex \epsilon)^*,
$$
where the event $\hat{\sig_t}?\nex \epsilon$ models the transition from the  state $\mathit{off}$ to the  state $\mathit{on}$, on which the signal $T$ is triggered,
while the event $\hat{\sig_t}?\nex \sig_c\nex \epsilon$ models the transition from the state $\mathit{on}$ to the state $\mathit{off}$, on which both the signals $T$ and $C$ are triggered.
The logical order between $\hat{\sig_t}$ and $\sig_c$ in the event $\hat{\sig_t}?\nex \sig_c\nex \epsilon$ indicates the ``trigger - effect'' relation between the two signals $T$ and $C$.
The program $(\bar{\sig_t}?\nex \epsilon)^*$ means to wait the signal $\sig_t$ without doing anything.

FDChart looks simpler than the program $P_{\mathit{FD}}$ because it omits the behaviour of ``waiting the signal'' on its graph (, which should be a self-loop added on the state $\mathit{off}$ or $\mathit{on}$).
In SPs, we can actually define a syntactic sugar for this behaviour: for any signal $\sig$ and event $\alpha$,
\begin{equation}
\label{equ:sugar}
\hat{\sig}(x)?? \alpha \dddef (\bar{\sig}\nex \epsilon)^*\seq \hat{\sig}(x)?\nex \alpha,
\end{equation}
which means that the program waits $\sig$ until it it emitted, and then proceeds as $\alpha$ (at the same instant).
With this shorthand, $P_{\mathit{FD}}$ can be rewritten as
$$
P_{\mathit{FD}} = (\hat{\sig_t}?? \epsilon \seq \hat{\sig_t}?? (\sig_c\nex \epsilon))^*.
$$

In the second example, we consider a simple circuit from~\cite{Andre03}, called a ``binary counter'',  which is modelled as a syncChart in Fig.~\ref{figure:A Binary Counter}.
The binary counter reads every signal $T$ and counts the number of occurrences of $T$ by outputting the signals $B0$ and $B1$ that represent the bit:
the present of a signal represents $1$, while the absent of a signal represents $0$.
The syncChart of the binary counter, called ``Cnt2'', is a parallel syncChart and is obtained by a parallel composition of two syncCharts that model circuits called ``flip-flops''~\cite{Andre03}.
A dashed line is to separate the two syncCharts running in parallel.
The execution mechanism of parallel syncCharts is the same as that of parallel SPs.
At each reaction, transitions in different syncCharts can be triggered simultaneously in a logical order.
Therefore, the behaviour of the parallel syncChart is deterministic.

\begin{figure}[htpb]
        \centering
        \includegraphics[width=0.5\linewidth]{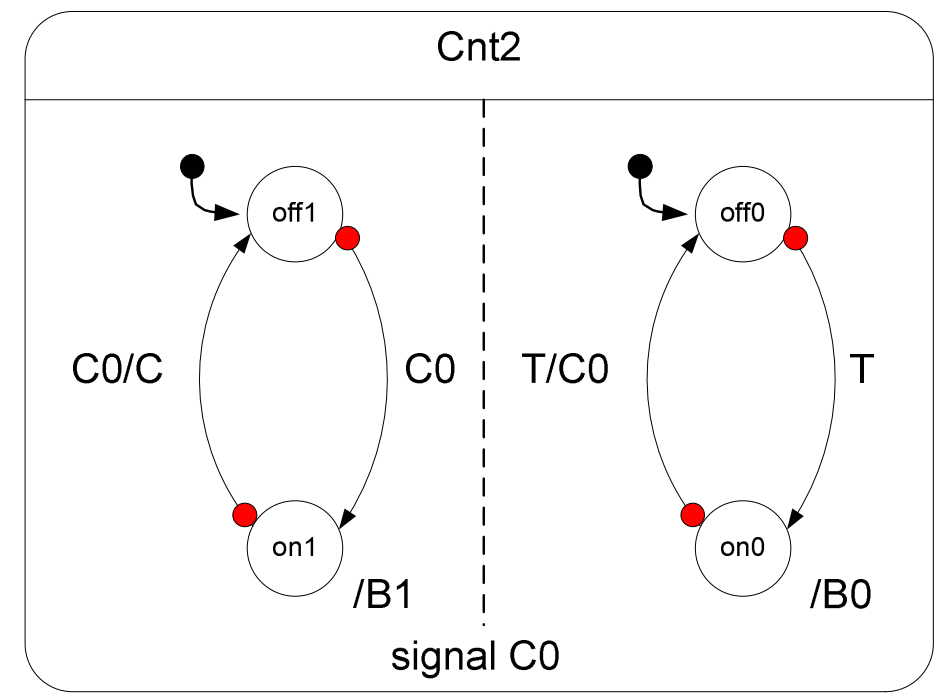}
        \caption{A Binary Counter}
        \label{figure:A Binary Counter}
    \end{figure}

In syncCharts, a state can be tagged with a $\mathit{trigger}/\mathit{effect}$ label and
special types of transitions were introduced to decide when the label on a state is triggered.
This does not increase the expressiveness of syncCharts but can reduce the number of states and transitions a syncChart has.
The special transitions (with a red circle at the tail of an arrow) appeared in Cnt2 are called ``strong abortion transitions''~\cite{Andre03}.
When a strong abortion transition is triggered, the label of the state that the transition is entering is triggered simultaneously, while the label of the state the transition is exiting cannot be triggered.

The behaviour of Cnt2 is given as follows:
\begin{itemize}
    \item The right syncChart:
        \begin{enumerate}[(1)]
            \item At the state $\mathit{off0}$, the syncChart waits for the signal $T$ and moves to the state $\mathit{on0}$, at the same instant, the signal $B0$ is emitted.
            \item At the state $\mathit{on0}$, the syncChart waits for the signal $T$ and if it is emitted, the signal $C0$ is emitted, and then the syncChart moves to the state $\mathit{off0}$.
        \end{enumerate}
    \item The left syncChart:
        \begin{enumerate}[(1)]
            \item At the state $\mathit{off1}$, the syncChart waits for the signal $C0$ from the right syncChart and moves to the state $\mathit{on1}$, at the same instant, the signal $B1$ is emitted.
            \item At the state $\mathit{on1}$, the syncChart waits for the signal $C0$ from the right syncChart and if it is emitted, the signal $C0$ is emitted, and then the syncChart moves to the state $\mathit{off0}$.
        \end{enumerate}
\end{itemize}
Let $\sig_t$, $\sig_{c0}$, $\sig_{b0}$, $\sig_{b1}$, $\sig_{c}$ represent the signals $T$, $C0$, $B0$, $B1$  and $C$ respectively, Cnt2 can be encoded as an SP as follows:
$$
\begin{aligned}
P_{\mathit{BC}} =&\ P_l\para P_r, \\
P_r =&\ (\hat{\sig_{t}}?? (\sig_{b0}\nex \epsilon) \seq \hat{\sig_{t}}?? (\sig_{c0}\nex \epsilon))^*,\\
P_l =&\ (\hat{\sig_{c0}}?? (\sig_{b1}\nex \epsilon) \seq \hat{\sig_{c0}}?? (\sig_{c}\nex \epsilon))^*.
\end{aligned}
$$
$P_l$, $P_r$ encode the left and right syncCharts respectively.

\subsection{Specifying and Proving Properties in SyncCharts}
\label{section:Specifying and Proving Properties in SyncCharts}

We consider a simple safety property for the first example discussed above, which says
\begin{center}``whenever the signal $\sig_c$ is emitted, the signal $\sig_t$ is emitted''.\end{center}
Since $P_{\mathit{FD}}$ is an open SP, we assume an environment $E$ for $P_{\mathit{FD}}$ which generally emits $\sig_t$ or does nothing at each instant.
To collect the information about the signals $\sig_t$ and $\sig_c$, we define two observers $O_t$ and $O_c$, which listen to these two signals and record their states within the local variables
$x$ and $y$ at each instant.

The property is thus specified in an SDL formula as follows:
$$
\phi_{\mathit{FD}} = (x=0\wedge y=0) \to [\para(P_{\mathit{FD}}, E, O_t, O_c)]\Box (y=1\to x=1),
$$
where
$$
\begin{aligned}
E =&\ (\sig_t\nex \epsilon\cup \epsilon)^*, \\
O_t =&\ (\hat{\sig_t}?\nex x:=1\nex \epsilon\cup \bar{\sig_t}?\nex x:=0\nex \epsilon)^*, \\
O_c =&\ (\hat{\sig_c}?\nex y:=1\nex \epsilon\cup \bar{\sig_c}?\nex y:=0\nex \epsilon)^*.
\end{aligned}
$$

\begin{figure}[htpb]
        \noindent\makebox[\textwidth]{%
        \scalebox{0.8}{
        \begin{tikzpicture}[->,>=stealth', node distance=3cm]
        \node[draw=none] (txt2) {
            $
            \infer[]
            {
                (1):\
                \phi_1, \psi_1\seqArrow [p^*_2][p_2]\Box\phi_2
            }
            {
                \infer
                {
                    \phi_1,\psi_1\seqArrow \psi_2
                }
                {
                    \surd
                }
                &
                \infer[^{(\forall r), (\to r), (x:=e)}]
                {
                    \phi_1,\psi_1\seqArrow \forall x. \forall y. (\psi_2\to [p_2]\psi_2)
                }
                {
                    \infer
                    {
                        \phi_1,\psi_1, \psi_2\seqArrow \psi_2
                    }
                    {
                        \surd
                    }
                }
                &
                \infer
                {
                    \phi_1, \psi_1\seqArrow \forall x.\forall y.(\psi_2\to [p_2]\Box\phi_2)
                }
                {
                    ...
                }
            }
            $
        };

        \node[draw=none, yshift=-1.75cm] at (txt2.south) (txt1) {
        $
        \infer[^{(\to r), (\para,\mathit{seq}), (\mathit{r1})}]
        {
            \cdot \seqArrow\phi_{\mathit{FD}}
        }
        {
            \infer[^{(*,\Box\phi)}]
            {
                \phi_1\seqArrow [P]\Box \phi_2
            }
            {
                \infer[^{([*])}]
                {
                    \phi_1\seqArrow [p^*_1][p_1]\Box\phi_2
                }
                {
                    \infer
                    {
                        \phi_1\seqArrow \psi_1
                    }
                    {
                        \surd
                    }
                    &
                    \infer
                    {
                        \phi_1\seqArrow\forall x.\forall y.(\psi_1\to [p_1]\psi_1)
                    }
                    {
                        ...
                    }
                    &
                    \infer[^{(\forall r), (\to r)}]
                    {
                        \phi_1\seqArrow\forall x.\forall y.(\psi_1\to [p_1]\Box\phi_2)
                    }
                    {
                        \infer[^{(\seq, \Box\phi), (\wedge r)}]
                        {
                            \phi_1, \psi_1\seqArrow [p_1]\Box\phi_2
                        }
                        {
                            \infer[^{(*, \Box\phi)}]
                            {
                                \phi_1, \psi_1\seqArrow [p^*_2]\Box\phi_2
                            }
                            {
                                (1)
                            }
                            &
                            \infer[]
                            {
                                \phi_1, \psi_1\seqArrow [p^*_2][p_3]\Box\phi_2
                            }
                            {
                                ...
                            }
                        }
                    }
                }
            }
        }
        $
        };

        \path
        ;

        \draw ([xshift=-0.37cm, yshift=-0.37cm]txt1.north) -- ([xshift=-0.37cm-1.5cm, yshift=-0.37cm]txt1.north) -- ([xshift=-0.37cm-1.5cm, yshift=0.12cm]txt2.south -| txt1.north);

        \end{tikzpicture}
        } 
        }

        \noindent\makebox[\textwidth]{%
        \scalebox{0.8}{
        \begin{tabular}{|l|}
        \toprule
        $\phi_1 = (x=0 \wedge y=0)$
        \ \
        $\phi_2 = (y=1\to x=1)$
        \ \
        $\psi_1 = \true$
        \ \
        $\psi_2 = \true$\\
        \hline
            $P = ((x:=0\nex y:=0\nex \epsilon)^*\seq (x:=1\nex y:=0\nex \epsilon)\seq (x:=0\nex y:=0\nex \epsilon)^*\seq (x:=1\nex y:=1\nex \epsilon))^*$
        \\
            $p_1 = (x:=0\nex y:=0\nex \epsilon)^*\seq (x:=1\nex y:=0\nex \epsilon)\seq (x:=0\nex y:=0\nex \epsilon)^*\seq (x:=1\nex y:=1\nex \epsilon)$
        \\
            $p_2 = (x:=0\nex y:=0\nex \epsilon)$
            \ \
            $p_3 = (x:=1\nex y:=0\nex \epsilon)\seq (x:=0\nex y:=0\nex \epsilon)^*\seq (x:=1\nex y:=1\nex \epsilon)$
        \\
        \bottomrule
        \end{tabular}
        }
        }
        \caption{The Derivation Procedure of $\phi_{\mathit{FD}}$}
        \label{figure:The Derivation Procedure of phiFD}
\end{figure}

Fig.~\ref{figure:The Derivation Procedure of phiFD} shows the process of proving $\phi_{\mathit{FD}}$.
By rule $(\para, \mathit{seq})$, the program $\para(P_{\mathit{FD}}, E, O_t, O_c)$ can be rewritten into a closed sequential program $P$
through a Brzozoski's procedure.
The proof tree is constructed by applying the rules of the SDL calculus reversely step by step.
To save spaces, we omit most of the branches of the proof tree by ``$...$'', and we merge several derivations into one by listing all the rules applied during these derivations on the right side of the derivation.
$\phi_{\mathit{FD}}$ is proved to be true if the AFOL formula at each leaf node is valid (indicated by a $\surd$ at each leaf node).
The AFOL formula at each leaf node can then be checked through an SAT/SMT solving procedure.

\ifx
We take the first example in Sect.~\ref{section:Encoding Basic SyncCharts} as a simple example to explain how to specify and prove properties of syncCharts.

In Sect.~\ref{section:Encoding Basic SyncCharts}, we see that $P_{\mathit{FD}}$ specifies FDChart as an open system. However, two problems arise when we want to specify properties of FDChart.
One is that we cannot specify the properties directly for $P_{\mathit{FD}}$ since $P_{\mathit{FD}}$ is an open program and in SDL
an open program does not have a semantics.
The other problem is that we cannot specify the properties concerning signals in syncCharts because in SDL, a signal is just a channel for communication between open programs, but in a syncChart, a signal is essentially a variable that carries values.

\ifx
When encoding this syncChart, we need to deal with two problems arising from the differences between syncCharts and SPs.
One problem is that in SDL we can only describe the properties of a closed SP, but the syncChart of the FD is an open system.
The other problem is that in SDL a signal is just a channel for communication between open programs and has no concrete semantics, while in a syncChart a signal is essentially a variable that carries values.
This means that in SDL we cannot directly specify the properties of syncCharts that concern signals.
\fi

To solve these two problems, we introduce auxiliary local variables and expressions to represent the information carried by the signals of a syncChart, so we can turn an open SP into a closed SP by eliminating the signals.
What expressions should be introduced depends on what kinds of properties we are about to specify.
In this example,
we use $x_t$, $x_c$ to represent the current numbers of the occurrences of the signals $T$ and $C$ respectively.
The program $P_{\mathit{FD}}$ can then be turned into an SP as:
$$
P'_{\mathit{FD}} = (\epsilon^*\seq x_t:=x_t + 1\nex \epsilon\seq \epsilon^*\seq x_t:=x_t + 1\nex x_c := x_c + 1\nex \epsilon)^*,
$$
where we remove the signal test $\hat{\sig_t}$ and replace the signals $\sig_t$, $\sig_c$ with two expressions $x_t:=x_t + 1$ and $x_c := x_c + 1$ respectively.
The program $P'_{\mathit{FD}}$ just captures the behaviour of FDChart, but from the perspective that FDChart is a closed system where the signals are invisible from outside.

\ifx
which contains two reactions in a loop.
The first reaction models the transition from the $\mathit{off}$ state to the $\mathit{on}$ state, on which the signal $T$ is triggered.
The second reaction models the transition from the $\mathit{on}$ state to the $\mathit{off}$ state, on which both the signals $T$ and $C$ are triggered.
The order in the second event indicates the ``trigger - effect'' relation between the two signals.
\fi

With $P'_{\mathit{FD}}$, we specify a property for FDChart, saying ``On the state $\mathit{off}$, the number of the occurrences of the signal $T$ is twice as much as the number of the occurrences of the signal $C$'', which is expressed as an SDL formula
$$
\phi_{\mathit{FD}} = (x_t=0\wedge x_c = 0) \to [P'_{\mathit{FD}}](x_t = 2\cdot x_c).
$$

Fig.~\ref{???} gives the process of proving $\phi_{\mathit{FD}}$.
Starting from the root node, the proof tree is formed by applying the rules of the SDL calculus reversely step by step.
$\phi_{\mathit{FD}}$ is proved to be true if the AFOL formula at each leaf node is valid (indicated by a $\surd$ at each leaf node).

\fi

\ifx
\subsection{Encoding Preemption in SyncCharts}
\label{section:Encoding Preemption in SyncCharts}

In this section, we show how SPs can express the advanced features of syncCharts through an example, which is, to use SPs to express the preemption in syncCharts.
Preemption is one of the most important feature in synchronous modelling languages such as SyncChart, Esterel, Quartz~\cite{???} and so on.
It is the possibility for the environment to interrupt the current execution of a syncChart at any reaction.
Such a mechanism greatly simplifies the form of modelling languages so as to reduce the burden for engineers.
Despite of this, as mentioned in Sect.~\ref{section:Introduction}, preemption does not support modular verification in terms of theorem proving.
In SPs, we do not provide a primitive for expressing preemption.
However, in the following, we show that by defining a suitable sugar, we can capture this mechanism in SPs.

Consider adding a transition labelled as ``$R$'' from BCChart (Fig.~\ref{???}) to a state named as ``$\mathit{ex}$'', as shown in Fig.~\ref{???},
the signal $R$ enables the preemption that interrupts any transitions in BCChart and triggers a transition to the state $\mathit{ex}$.
Here we in fact use the mechanism ``hierarchy'' in syncCharts to take the whole BCChart as a \emph{macrostate}~\cite{???}.
The transition from the macrostate to the state $\mathit{ex}$ is also a strong abortion transition, at any reaction, once the signal $R$ is triggered,
it is immediately triggered and all other transitions in the macrostate at the same reaction are forbidden.
Fig.~\ref{???} shows a Mealy machine as the behavioural semantics of this syncChart, where $\bar{X}$ represents that a signal $X$ is absent.
\fi

\section{Related Work}
\label{section:Related Work}
\subsection{Verification Techniques for Synchronous Models}

\ifx
Previous verification approaches for synchronous models are mainly based on model checking.
The compilers of synchronous programming languages, such as Esterel, Lustre and Quartz, have integrated the model checking procedure as one part of themselves~\cite{???,???,???}.
When model checking synchronous models, different specifications languages were adopted, such as synchronous observers~\cite{synchronous observers and the verification of reactive systems,From a synchronous declarative language to a temporal logic dealing with multiform time}, LTL~\cite{???}, clock constraints~\cite{???,???}, and so on.
Another fully automatic approach is based on SAT/SMT checking.

To tackle the state-explosion problem in model checking, automated theorem proving is often used as an alternate or complement technique.

SAT/SMT checking is also used to tackle the state-explosion problem and

To tackle the state-explosion problem, SAT/SMT checking is often adopted
\fi

Previous verification approaches for synchronous models are mainly based on model checking.
Different specification languages, such as synchronous observers~\cite{Halbwachs93b,Pilaud88}, LTL~\cite{Jagadeesan95} and clock constraints~\cite{Andre09}, were used to capture the safety properties. They were transformed into target models where reachability analysis was made.
For synchronous programming languages, the process of reachability analysis is often embedded into their compilers, for instance, cf.~\cite{Berry85,Halbwachs91}.

Despite for its decidability and efficiency on small size of state spaces, model checking suffers from the notorious state-explosion problem.
Recent years automated/interactive theorem proving, as a complement technique to model checking, has been gradually applied for analysis and verification of synchronous models in different aspects.
One hot research work is to use theorem provers like Coq~\cite{Bertot04} to mechanize and verify the compiling processes of synchronous programming languages, so that
the equivalence between compiled code and source code can be guaranteed~\cite{Bourke17,Berry19}.
SAT/SMT solving, as a fully automatic verification technique, was used for checking the time constraints in synchronous programming languages such as Lustre~\cite{Hagen08} and Signal~\cite{Ngo14}.
In~\cite{Talpin15}, a type theory was proposed to provide a compile-time framework for analyzing and checking time properties of Signal, through the inference of refinement types.

Rather than targeting on explicit synchronous languages, our proposed formalism focuses on a more general synchronous model SPs, extended from regular programs that are suitable for compositional reasoning.
As indicated in Sect.~\ref{section:Introduction}, SPs capture the essential features of synchronous models and ignore those which do not support compositional reasoning.
Different from synchronous languages that are totally deterministic, SPs have the extra expressive power to support program refinement with the non-deterministic operator $\cup$.
Similar to the type-theory approach~\cite{Talpin15}, instead of directly using SAT/SMT solving, SDL provides compositional rules for decomposing SPs according to their syntactic structures, so as to divide a big verification problem into small SMT-solving problems in derivation processes.

\cite{Florence19} proposed an equation theory for pure Esterel.
There, term rewrite rules were built for describing the constructive semantics of Esterel so that two different Esterel programs can be formally compared and their equivalences can be formally reasoned about.
\cite{Song21} proposed a so-called \emph{synchronous effects logic} for verifying temporal properties of pure Esterel programs.
A Hoare-style forward proving process was developed to compute the behaviours of Esterel programs as synchronous effects.
Then a term rewriting system was proposed to verify the temporal properties, which are also expressed as synchronous effects, of Esterel programs.

Compared to~\cite{Florence19,Song21}, the verification of SDL is not solely based on term rewriting, but also based on a Hoare-style program verification~\cite{ProgramVerifications09} process.
Instead of verifying by checking the equivalences or refinement relations between two programs,
we reason about the satisfaction relation between a program and a logic formula, in a form $[p]\phi$ or $[p]\Box\phi$.
In SPs, the rewrite rules built in Table~\ref{table:Rules for Parallel SPs} for reducing parallel SPs play a similar role as the rewrite rules defined in~\cite{Florence19,Song21} for symbolically executing parallel Esterel programs.
The synchronous effects used to capture the behaviours of Esterel programs in~\cite{Song21} was similar to the form of SPs.

In \cite{Gesell12}, a Hoare logic calculus was proposed for the synchronous programming language Quartz.
In that work, the authors manage to prove a Quartz program in a simple form in which there is no parallel compositions and all events in a macro step  are collected together in a single form called \emph{synchronous tuple assignment} (STA).
Such a simple form can be obtained either by manual encoding or by the compiler of Quartz in an automatic way.
The STA there actually corresponds to the macro event in SPs.
Compared to~\cite{Gesell12}, SP is not an synchronous language but a more general synchronous model.
Except for state properties, SDL also supports verifying safety properties with the advantage brought by dynamic logic to support formulas of the form $[p]\Box\phi$.


\ifx
During model checking procedures, SAT/SMT checking was often adopted for efficiency issues.

Such a procedure has even been embedded into the compilers of synchronous programming languages, such as the compilers for Esterel~\cite{???}, Lustre~\cite{???} and Quartz~\cite{???}.
When model checking synchronous models, different specifications languages were adopted, such as synchronous observers~\cite{synchronous observers and the verification of reactive systems,From a synchronous declarative language to a temporal logic dealing with multiform time}, LTL~\cite{???}, clock constraints~\cite{???,???} and so on.
Another fully automatic approach is based on SAT/SMT checking.
In this approach, synchronous languages, like Lustre~\cite{???} and Signal~\cite{???}, are directly translated into AFOL formulas, which can then be checked by an SMT solver, like Z3.

In recent years, as the development of general theorem provers such as Coq~\cite{Coq} and Isabelle, automated/interactive theorem proving has been gradually applied for analysis and verification of synchronous models in different aspects.
One application is to use theorem provers to mechanize and verify the compiling processes of synchronous programming languages so that
the equivalence between compiled code and source code can be guaranteed~\cite{A formally verified compiler for Coq (PLDI 2017),Towards Coq-verified Esterel Semantics and Compiling (2019)}.
In~\cite{talpin's refinement type for dataflow models}, a type theory was proposed to provide a compile-time framework for analyzing and checking time properties of Signal, through the inference of refinement types~\cite{Liquid Type}.

[verified semantics]
[A formally verified compiler for Coq (PLDI 2017), Towards Coq-verified Esterel Semantics and Compiling (2019)]

[???]
[A calculus for esterel (2019)]

[type inference]
[Talpin's work]

[contracts, hoare logic]
[contracts for synchronous interface, ]

[theorem proving approach]
\fi

\subsection{Dynamic Logic}
Dynamic logic was firstly proposed by V.R. Pratt~\cite{Pratt76} as a formalism for modelling and reasoning about program specifications.
The syntax of SDL is largely based on and extends that of FODL~\cite{Harel79} and that of its extension to concurrency~\cite{Peleg87}.
Temporal dynamic logical formulas of the form $[p]\Box\phi$ were firstly proposed in Process logic~\cite{Harel82}.
\cite{Beckert01} studied a first-order dynamic logic containing formulas of this form (there, it was written as $[|p|]\phi$) and initially proposed
a relatively complete proof system for it.
Inspired from~\cite{Beckert01}, in \cite{Platzer07}, the author introduces the form $[p]\Box\phi$ into his DTDL and proposed
the compositional rules for $[p]\Box\phi$.
\ifx
Process logic~\cite{???} initially allows dynamic formulas of the form $[p]\eta$, where $\eta$ is a general temporal formula, like $\phi\bff{U} \psi$ (read as ``$\phi$ until $\psi$''), to express temporal properties of program models.
However, by now only the proof system of proposition process logic has been fully investigated.
On the other side, \cite{A sequent calculus for first-order dynamic logic with trace modalities} studied a type of dynamic logic that only contains formulas of the form $[|p|]\phi$, which has the same meaning as $[p]\Box\phi$ in process logic.
The idea in \cite{A sequent calculus for first-order dynamic logic with trace modalities} was further developed by A. Platzer in his DTDL~\cite{???}, where the compositional rules for the formulas $[p]\Box\phi$ were proposed and a relatively complete proof system was given.
\fi
The semantics of SDL mainly follows the trace semantics of process logic.
In SDL, we inherit formulas of the form $[p]\Box\phi$ from DDTL to express safety properties of synchronous models, and
adopt the compositional rules from DTDL (i.e. the rules $(\seq, \Box\phi)$ and $(*, \Box\phi)$) in order to build a relatively complete proof system for SDL.

Many variations of dynamic logic have been proposed for modelling and verifying different types of programs and system models.
For instance, Y.A. Feldman proposed a probabilistic dynamic logic PrDL for reasoning about probabilistic programs~\cite{Feldman84};
\cite{Rustan07} proposed the Java Card Dynamic Logic for verifying Java programs;
In~\cite{Platzer07b} and \cite{Platzer07}, the Differential Dynamic Logic (DDL) and DTDL were proposed respectively for specifying and verifying hybrid systems.
DDL introduced differential equations in the regular programs of FODL to capture physical dynamics in hybrid systems.
The time model of DDL is continuous.
In DDL, a discrete event (e.g. $x:=e$) does not consume time, while a continuous event (i.e. a differential equation) continuously evolves until some given conditions hold.
Compared to DDL, the time model of SDL is discrete and is reflected by the macro events.
SDL mainly focuses on capturing the features of synchronous models and preserving them in program models during theorem-proving processes.

An attempt to build a dynamic logic for synchronous models was made in~\cite{Zhang21}, where a clock-based dynamic logic (CDL) was proposed to specify and verify specific clock specifications in synchronous models.
SDL differs from CDL in the following two main points.
Firstly, in CDL, all events occurring at an instant are disordered, while in SDL, all micro events of a macro event are executed in a logical order.
So SDL is able to capture data dependencies, which are an important feature of synchronous models.
Secondly, in CDL, we propose formulas of the form $[p]\xi$, where $\xi$ is a clock constraint such as $c_1\prec c_2$, to express the clock specifications of a program $p$.
While in SDL, we introduce a more general form $[p]\Box\phi$, which can express not only clock constraints, but also other safety properties.
Compared to CDL, SDL is a logic that is more general and more expressive.

\ifx
\subsection{Synchronous Models}
[SCCS...]

[Deterministic Actor]

[Different Temporal Dynamic logic]
where introduced $[p]\Box\phi$ ,which is quite suitable for capturing safety properties in our proposed SP model.
\fi

\section{Conclusion and Future Work}
\label{section:Conclusion and Future Work}

In this paper, we mainly propose a dynamic logic --- SDL --- for specifying and verifying synchronous models based a theorem proving technique.
We define the syntax and semantics of SDL, build a constructive semantics for parallel SPs, and propose a sound and relatively complete proof system for SDL.
We show the potential of SDL to be used in specifying and verifying synchronous models through an example.

As for future work, we mainly focus on mechanizing SDL in the theorem prover Coq and apply SDL in specifying and verifying more interesting examples rather than the toy examples in this paper.

\bibliographystyle{elsarticle-num}
\bibliography{SDL-arxiv-20210408}

\begin{thebibliography}{10}
\bibitem{Florence19}
S.~P. Florence, S.-H. You, J.~A. Tov, R.~B. Findler, A calculus for esterel: If
  can, can. if no can, no can., Proc. ACM Program. Lang. 3~(POPL).
\end{thebibliography}

\clearpage
\appendix
\section{Proofs for some Propositions and the Soundness of SDL}
\label{section:Proofs for some Propositions and the Soundness of SDL}

\begin{proof}[Proof of Prop.~\ref{prop:trecs property}]
    We prove by induction on the syntactic structure of an SP $p$.

    The base cases are trivial.
    We need to consider 3 basic cases: when $p$ is $\halt$, $\noth$ or $\alpha$.
    We only take $\alpha$ as an example, the cases for $\halt$ and $\noth$ are similar.
    In fact, we immediately obtain the result since $\tau(\alpha) = \{\alpha\}$.

    If $p = p_1\seq p_2$, by induction hypothesis, we have $\val(p) = \val(p_1)\circ \val(p_2) = \bigcup_{r\in \tau(p_1)}\val(r) \circ \bigcup_{r\in \tau(p_2)}\val(r)=\bigcup_{r_1\in \tau(p_1), r_2\in \tau(p_2)}\val(r_1)\circ \val(r_2)$.
    According to the definition of the operator $\link$ (Def.~\ref{def:Operator link}), it is easy to see that for any trecs $q_1$, $q_2$, $\val(q_1\link q_2)= \val(q_1\seq q_2) = \val(q_1)\circ \val(q_2)$ holds.
    This is because $\link$ just concatenate two trecs with the operator $\seq$ in a special way which does not affect the semantics of the composed program, as analyzed in Section~\ref{section:Valuations of Terms and Closed Programs}.
    Therefore, we have $\bigcup_{r_1\in \tau(p_1), r_2\in \tau(p_2)}\val(r_1)\circ \val(r_2) = \bigcup_{r_1\in \tau(p_1), r_2\in \tau(p_2)}\val(r_1\link r_2) = \bigcup_{r\in \tau(p_1\seq p_2)}\val(r)$.

    If $p=p_1\cup p_2$, by induction hypothesis, we immediately get $\val(p) = \val(p_1)\cup \val(p_2) = \bigcup_{r\in \tau(p_1)}\val(r)\cup \bigcup_{r\in \tau(p_2)}\val(r) = \bigcup_{r\in \tau(p_1)\cup \tau(p_2)}\val(r) = \bigcup_{r\in \tau(p_1\cup p_2)}\val(r)$.

    If $p=q^*$, by induction hypothesis, we have $\val(q) = \bigcup_{r\in \tau(q)}\val(r)$. So $\val^n(q) = \underbrace{\bigcup_{r\in \tau(q)}\val(r)\circ...\circ \bigcup_{r\in \tau(q)}\val(r)}_{n} = \bigcup_{r\in\tau(q)}\val^n(r) = \bigcup_{r\in\tau(q)}\val(\underbrace{r\link...\link r}_{n}) = \bigcup_{r\in \tau(q^n)}\val(r)$. Hence, we have $\val(p) = \bigcup^\infty_{n=0}\val^n(q) = \bigcup^\infty_{n=0}\bigcup_{r\in \tau(q^n)}\val(r) = \bigcup_{r\in \tau(q^*)}\val(r)$.

    If $p=\para(q_1,...,q_n)$, we immediately obtain the result by Def.~\ref{def:Par}.
\end{proof}

\begin{prop}
    \label{prop:sequent}
    The following two propositions hold:
    \begin{enumerate}
        \item $\begin{aligned}\infer[^{(\mathit{seq1})}]
                    {
                        \Gamma\seqArrow \phi, \Delta
                    }
                    {
                        \Gamma\seqArrow \psi_1, \Delta
                        &
                        ...
                        &
                        \Gamma\seqArrow \psi_n, \Delta
                    }\end{aligned}$
              is sound for all contexts $\Gamma$ and $\Delta$
              iff
              the formula $\bigwedge^n_{i=1} \psi_i \to \phi$ is valid.

        \item $\begin{aligned}\infer[^{(\mathit{seq2})}]
                    {
                        \Gamma, \phi_1,...,\phi_n\seqArrow \Delta
                    }
                    {
                        \Gamma, \psi\seqArrow \Delta
                    }\end{aligned}$
            is sound for all contexts $\Gamma$ and $\Delta$
            iff
            the formula $\bigwedge^n_{i=1}\phi_i \to \psi$ is valid.
    \end{enumerate}
\end{prop}

\begin{proof}
We only prove the proposition \textit{1}.
The proposition \textit{2} follows a similar idea.

\textit{1}:
For the direction $\rightarrow$, if rule $(\mathit{seq1})$ is sound, then we know that if
$\bigwedge_{\varphi\in \Gamma}\varphi \to \psi_1\vee \bigvee_{\varphi\in \Delta}\varphi$,...,
$\bigwedge_{\varphi\in \Gamma}\varphi \to \psi_n\vee \bigvee_{\varphi\in \Delta}\varphi$ are valid, then
$\bigwedge_{\varphi\in \Gamma}\varphi \to \phi\vee \bigvee_{\varphi\in \Delta}\varphi$ is valid.
Now we show that for any $s\in \bff{S}$, if $s\models \bigwedge^n_{i=1} \psi_i$ then $s\models \phi$.
Let $\Gamma = \{\psi_1,...,\psi_n\}$, $\Delta = \{\false\}$, it is easy to see that
since $(\psi_1\wedge...\wedge\psi_n) \to (\psi_i\vee \false)$ is valid for any $1\le i\le n$, from the soundness of rule $(\mathit{seq1})$,
$(\psi_1\wedge...\wedge\psi_n)\to (\phi \vee \false)$ is valid, which means that for any $u\in \bff{S}$, if $u\models \psi_1\wedge...\wedge\psi_n$, then $u\models \phi$.
Therefore we immediately have $s\models \phi$.

For the other direction $\leftarrow$, if the formula $\bigwedge^n_{i=1} \psi \to \phi$ is valid, we now show that if
$\bigwedge_{\varphi\in \Gamma}\varphi \to \psi_1\vee \bigvee_{\varphi\in \Delta}\varphi$,...,
$\bigwedge_{\varphi\in \Gamma}\varphi \to \psi_n\vee \bigvee_{\varphi\in \Delta}\varphi$ are valid, then
$\bigwedge_{\varphi\in \Gamma}\varphi \to \phi\vee \bigvee_{\varphi\in \Delta}\varphi$ is valid.
For any $s\in \bff{S}$,
the situation when $s\not \models \bigwedge_{\varphi\in \Gamma}\varphi$ or when $s\models \bigvee_{\varphi\in \Delta}\varphi$ is trivial.
So we assume that $s\models \bigwedge_{\varphi\in \Gamma}\varphi$ and $s\not\models \bigvee_{\varphi\in \Delta}\varphi$.
Hence we have $s\models \psi_1$,..., $s\models \psi_n$. So $s\models \bigwedge^n_{i=1}\psi_i$.
Since $\bigwedge^n_{i=1} \psi_i \to \phi$ is valid, $s\models \phi$ holds.
So $s\models \bigwedge_{\varphi\in \Gamma}\varphi \to \phi\vee \bigvee_{\varphi\in \Delta}\varphi$.
Since $s$ is an arbitrary state, $s\models \bigwedge_{\varphi\in \Gamma}\varphi \to \phi\vee \bigvee_{\varphi\in \Delta}\varphi$ is valid.

\textit{2}: The direction $\rightarrow$ follows a similar idea as the proof above by setting $\Gamma = \{\true\}$ and $\Delta = \{\psi\}$.
The other direction $\leftarrow$ is also similar.
\end{proof}

\begin{prop}[Soundness of the Rules of Type (a)]
    The rules $(\alpha, \Box\phi), (\psi?), (x:=e), (\epsilon), (\noth)$ and $(\halt)$ of Table~\ref{table:Rules for Closed SPs} are sound.
\end{prop}

\begin{proof}
By Prop.~\ref{prop:sequent}, for a rule $\begin{aligned}\infer=[]{\phi}{\psi}\end{aligned}$, it is sufficient to prove that $\phi\leftrightarrow \psi$ is valid.

In the following proofs for different rules, let $s$ be any state of $\bff{S}$.

For rule $(\alpha, \Box\phi)$, $s\models [\alpha]\Box\phi$ iff $ss'\in \val_\pi(\Box\phi)$ for all traces $ss'\in \val(\alpha)$, iff $s\in \val(\phi)$ and $s'\in \val(\phi)$ for all traces $ss'\in \val(\alpha)$,
iff $s\models \phi$ and $s'\in \val(\phi)$ for all traces $ss'\in \val(\alpha)$, iff $s\models \phi$ and $s\models [\alpha]\phi$, iff $s\models \phi\wedge [\alpha]\phi$.

For rule $(\psi?)$, $s\models [\psi?\nex \alpha]\phi$ iff $s'\in \val(\phi)$ for all traces $ss'\in \val(\psi?\nex \alpha)$,
iff $s'\in \val(\phi)$ for all traces $ss'\in \mval(\psi?\nex \alpha)$,
iff $s'\in \val(\phi)$ for all traces $ss'=ss\circ ss'$ with $s\in \val(\psi)$ and $ss'\in \val(\alpha)$,
iff if $s\in \val(\psi)$, then $s'\in \val(\phi)$ for all traces $ss'=ss\circ ss'$ with $ss'\in \val(\alpha)$,
iff if $s\in \val(\psi)$, then $s'\in \val(\phi)$ for all traces $ss'\in \val(\alpha)$,
iff $s\models \psi \to [\alpha]\phi$.

For rule $(x:=e)$, $s\models [x:=e\nex \alpha]\phi$ iff $s'\in \val(\phi)$ for all traces $ss'\in \val(x:=e\nex \alpha)$,
iff $s'\in \val(\phi)$ for all traces $ss'\in \mval(x:=e\nex \alpha)$,
iff $s'\in \val(\phi)$ for all traces $ss' = ss''\circ s''s'$ with $s''=s[x\mapsto \val_s(e)]$ and $s''s'\in \mval(\alpha)$,
iff $s'\in \val(\phi)$ for all traces $s''s'\in \val(\alpha)$ with $s''=s[x\mapsto \val_s(e)]$,
iff $s''\models [\alpha]\phi$ and $s''=s[x\mapsto \val_s(e)]$,
iff $s\models ([\alpha]\phi)[e/x]$.

For rule $(\epsilon)$, $s\models [\epsilon]\phi$ iff $s'\in \val(\phi)$ for all traces $ss'\in \val(\epsilon)$,
iff $s\in \val(\phi)$ for the trace $ss\in \mval(\epsilon) = \{tt\ |\ t\in \bff{S}\}$,
iff $s\in \val(\phi)$, iff $s\models \phi$.

For rule $(\noth)$, $s\models [\noth]\phi$ (resp. $s\models [\noth]\Box\phi$) iff $s'\in \val(\phi)$ (resp. $ss'\in \val_\pi(\Box\phi)$) for all traces $ss'\in \val(\noth)$,
iff $s\in \val(\phi)$ for the trace (of length 1) $s\in \val(\noth)$,
iff $s\in \val(\phi)$, iff $s\models \phi$.

For rule $(\halt)$, $s\models [\halt]\phi$ (resp. $s\models [\halt]\Box\phi$) iff $s'\in \val(\phi)$ (resp. $ss'\in \val_\pi(\Box\phi)$) for all traces $ss'\in \val(\halt)$,
iff $\true$ since there is no traces in $\val(\halt)$.

\end{proof}

\begin{prop}[Soundness of the Rewrite Rules of Type (b)]
    \label{prop:Soundness of the Rewrite Rules for Closed Sequential Programs}
    The rewrite rules $(\noth,\seq)$, $(\noth,*)$, $(\halt,\seq)$, $(\halt,*)$, $(\seq, \mathit{ass})$, $(\seq,\mathit{dis}1)$, $(\seq, \mathit{dis}2)$, $(\cup, \mathit{ass})$ and $(*,\mathit{exp})$ of Table~\ref{table:Rules for Parallel SPs} are sound.
\end{prop}

\begin{proof}
    We only show the soundness of the rules $(\seq,\mathit{ass})$, $(\seq, \mathit{dis}1)$ and $(*, \mathit{exp})$, the soundness of other rules is either trivial or can be proved in a similar way. We omit them here.

    For rule $(\seq,\mathit{ass})$, we have that for any programs $p, q$ and $r$, $\val((p;q);r) = \val(p;q)\circ \val(r) =(\val(p)\circ \val(q))\circ \val(r) = \val(p)\circ (\val(q)\circ \val(r)) = \val(p)\circ \val(q;r) = \val(p;(q;r))$.

    For rule $(\seq, \mathit{dis}1)$, we have that for any programs $p,q$ and $r$, $\val(p;(q\cup r)) = \val(p)\circ (\val(q\cup r)) = \val(p)\circ (\val(q)\cup \val(r)) = (\val(p)\circ \val(q))\cup (\val(p)\circ \val(r)) = \val(p;q)\cup \val(p;r) = \val(p;q\cup p;r)$.

    For rule $(*,\mathit{exp})$, we have that for any program $p$,
    $\val(p^*) = \bigcup^{\infty}_{n=0}\val^n(p) = \val^0(p) \cup \bigcup^{\infty}_{n=1}\val^n(p) = \val^0(p)\cup (\val(p)\circ (\val^0(p)\cup \val^1(p)\cup \val^2(p) \cup ... )) = \val^0(p)\cup (\val(p)\circ \bigcup^{\infty}_{n=0}\val^n(p)) = \val(\noth)\cup \val(p)\circ \val(p^*) = \val(\noth\cup p\seq p^*)$.
\end{proof}

\begin{prop}[Soundness of the Rewrite Rules of Type (c)]
    \label{prop:Soundness of the Rewrite Rules for Parallel Programs}
    The rewrite rules $(\para, \noth)$, $(\para, \halt)$, $(\para, \mathit{dis})$ and $(\para, \mathit{mer})$ of Table~\ref{table:Rules for Parallel SPs} are sound.
\end{prop}

\begin{proof}
    The soundness of the rules $(\para, \noth)$, $(\para, \halt)$ and $(\para, \mathit{mer})$ can be easily proved according to the cases $1$, $2$ and $4$ of Def.~\ref{def:Valuation of the Trecs} respectively.
    Here we only show how to prove the soundness of rule  $(\para, \mathit{mer})$, the proofs for other two rules are similar.

    By Def.~\ref{def:Valuation of Closed SPs} and \ref{def:Par} we know that $\val(\para(\alpha_1\seq q_1,...,\alpha_n\seq q_n)) = \bigcup_{r\in \tau(\para(\alpha_1\seq q_1,...,\alpha_n\seq q_n))}\valt(r)$,
    where each trec $r$ of $\para(\alpha_1\seq q_1,...,\alpha_n\seq q_n)$ must be of the form $\para(\alpha_1\seq r_1,...,\alpha_n\seq r_n)$ with $r_i\in \tau(q_i)$ for all $1\le i\le n$.
    According to the case $4$ of Def.~\ref{def:Valuation of the Trecs}, we have that for each $r$, $$\valt(r) = \valt(\para(\alpha_1\seq r_1,...,\alpha_n\seq r_n)) = \val(\alpha')\circ \val(\cap(r'_1,...,r'_n))$$ if $b=\true$,
    where $(b', \alpha', (r'_1\sep ...\sep r'_n)) = \Merge(\alpha_1\seq r_1\sep...\sep \alpha_n\seq r_n)$.
    From Def.~\ref{def:Function Merge} it is easy to see that actually the values of $b'$ and $\alpha'$ are only related to $\alpha_1$,...,$\alpha_n$ and have nothing to do with $r_1$,...,$r_n$,
    therefore, we have $b' = b$ and $\alpha' = \alpha$ (, where in the rule $(\para, \mathit{mer})$ we have $(b, \alpha, (q'_1\sep...\sep q'_n)) = \Merge(\alpha_1\seq q_1\sep...\sep \alpha_n\seq q_n)$).
    So we have $\bigcup_{r\in \tau(\para(\alpha_1\seq q_1,...,\alpha_n\seq q_n))}\valt(r) = \bigcup_{r\in \tau(\para(\alpha_1\seq q_1,...,\alpha_n\seq q_n))} \val(\alpha)\circ \val(\para(r'_1,...,r'_n)) = \val(\alpha)\circ \bigcup_{r'_1\in \tau(q'_1),...,r'_n\in \tau(q'_n)} \val(\para(r'_1,...,r'_n)) = \val(\alpha)\circ \val(\para(q'_1,...,q'_n))$.

    It remains to show the soundness of rule $(\para,\mathit{dis})$. We have that $\val(\para(...,p\cup q,...)) = \bigcup_{r\in \tau(\para(...,p\cup q,...))}\valt(r) = \bigcup_{r\in \tau(\para(...,p,...))}\valt(r)\cup \bigcup_{r\in \tau(\para(...,q,...))}\valt(r) = \val(\para(...,p,...))\cup \val(\para(...,q,...)) = \val(\para(...,p,...)\cup \para(...,q,...))$.

\end{proof}

We introduce the level of a rewrite relation between two programs as the following definition.
\begin{mydef}[Level of a Rewrite Relation]
\label{def:Level of a Rewrite Relation}
    Given two SPs $p$, $q$ and a rewrite relation $p\red q$ between them, we define the level of the relation, denoted as $\lev(p\red q)$, as follows:
    \begin{enumerate}
    \item $\lev(p\red q) \dddef 1$ if $p\red q$ is from one of the rules of the types $(b)$ and $(c)$ in Table~\ref{table:Rules for Parallel SPs}.
    \item $\lev(p\red q)\dddef 1 + \lev(p'\red q')$ if there exist SPs $p'$, $q'$ and a program hole $r\{\place\}$ such that $p = r\{p'\}$, $q=r\{q'\}$ and $p'\red q'$.
    In other words, $p\red q$ is from rule $(r2)$.
    \end{enumerate}
\end{mydef}

\begin{lemma}
    \label{lemma:parallel-lemma}
    Given a parallel program $\para(...,p\{\place\},...)$ with a program hole $p\{\place\}$,
    if $p\{q\}\red p\{r\}$ with $q\red r$ and $\lev(q\red r) = 1$ for some programs $q$ and $r$,
    then
    \begin{equation}
    \label{equ:parallel-lemma-1}
    \val(\para(..., p\{q\},...)) = \val(\para(...,p\{r\},...)).
    \end{equation}
\end{lemma}

\begin{proof}
    We prove by analyzing two different cases of $\lev(q\red r) = 1$:
    \begin{enumerate}
    \item If $q\red r$ is from one of the rules of type $(b)$ and rule $(\para, \mathit{dis})$ (of type $(c)$), then
    $\tau(p\{q\}) = \tau(p\{r\})$, so the equation (\ref{equ:parallel-lemma-1}) holds obviously.
    \item If $q\red r$ is from other rules of type $(c)$ except rule $(\para, \mathit{dis})$, then
    the equation (\ref{equ:parallel-lemma-1}) holds.
    \end{enumerate}

    The case 1 can be proved by induction on the syntactic structure of $p_i\{\place\}$ for different rules. Here we only take rule $(\seq, \mathit{dis}1)$ as an example, other cases are similar.

    Let $q = u_1\seq (u_2\cup u_3)$ and $r = u_1\seq u_2\cup u_1\seq u_3$.
    If $p_i\{\place\} = \place$, we need to prove that $\tau(q) = \tau(r)$, which holds because $\tau(u_1\seq (u_2\cup u_3)) = \{t_1\link t_2\ |\ t_1\in \tau(u_1), t_2\in \tau(u_2\cup u_3)\} = \{t_1\link t_2\ |\ t_1\in \tau(u_1), t_2\in \tau(u_2)\}\cup \{t_1\link t_2\ |\ t_1\in \tau(u_1), t_2\in \tau(u_3)\} = \tau(u_1\seq u_2)\cup \tau(u_1\seq u_3) = \tau(u_1\seq u_2\cup u_1\seq u_3)$.

    If $p\{\place\} = p_1\{\place\}\seq p_2$, by induction hypothesis we have $\tau(p_1\{q\}) = \tau(p_1\{r\})$,
    so $\tau(p\{q\}) = \{t_1\link t_2\ |\ t_1\in \tau(p_1\{q\}), t_2\in \tau(p_2)\} = \{t_1\link t_2\ |\ t_1\in \tau(p_1\{r\}), t_2\in \tau(p_2)\} = \tau(p\{r\})$.

    The case for $p\{\place\} = p_1\seq p_2\{\place\}$ is similar.

    If $p\{\place\} = p_1\{\place\}\cup p_2$, by induction hypothesis we have $\tau(p_1\{q\}) = \tau(p_1\{r\})$,
    so $\tau(p\{q\}) = \{t\ |\ t\in \tau(p_1\{q\})\}\cup \{t\ |\ t\in \tau(p_2)\} = \{t\ |\ t\in \tau(p_1\{r\})\}\cup \{t\ |\ t\in \tau(p_2)\} = \tau(p\{r\})$.

    The case for $p\{\place\} = p_1\cup p_2\{\place\}$ is similar.

    If $p\{\place\} = (p_1\{\place\})^*$, by induction hypothesis we have $\tau(p_1\{q\}) = \tau(p_1\{r\})$,
    so $\tau(p\{q\}) = \bigcup^\infty_{n=0}\tau(p^n_1\{q\}) = \bigcup^\infty_{n=0}\{t_1\link...\link t_n\ |\ t_i\in \tau(p_1\{q\})\mbox{ for $1\le i\le n$}\} = \bigcup^\infty_{n=0}\{t_1\link...\link t_n\ |\ t_i\in \tau(p_1\{r\})\mbox{ for $1\le i\le n$}\} = \bigcup^\infty_{n=0}\tau(p^n_1\{r\}) = \tau(p\{r\})$.

    If $p\{\place\} = \para(p_1,..., p_i\{\place\}, ..., p_n)$, by induction hypothesis we have $\tau(p_i\{q\}) = \tau(p_i\{r\})$, so
    $\tau(p\{q\}) = \{\para(t_1,...,t_n)\ |\ t_1\in \tau(p_1),...,t_i\in \tau(p_i\{q\}),...,t_n\in \tau(p_n)\} = \{\para(t_1,...,t_n)\ |\ t_1\in \tau(p_1),...,t_i\in \tau(p_i\{r\}),...,t_n\in \tau(p_n)\} = \tau(p\{r\})$.

    To prove the case 2, we analyze the rules $(\para, \noth)$, $(\para, \halt)$, $(\para, \mathit{mer})$ and rule $(\para,\mathit{seq})$ separately.

    The cases for the rules $(\para, \noth)$, $(\para, \halt)$ and $(\para, \mathit{mer})$ can be proved according to the cases $1$, $2$ and $4$ of Def.~\ref{def:Valuation of the Trecs} respectively.
    Here we only consider the proof for rule $(\para, \mathit{mer})$ as an example, other cases for the rules $(\para, \noth)$ and $(\para, \halt)$ are similar.

    If $q\red r$ is from rule $(\para, \mathit{mer})$, let $q = \para(\alpha_1\seq q_1,..., \alpha_n\seq q_n)$ and $r = \alpha\seq \para(q'_1,...,q'_n)$ where $(b, \alpha, (q'_1\sep...\sep q'_n)) = \Merge(\alpha_1\seq q_1\sep...\sep \alpha_n\seq q_n)$ and $b=\true$.
    In the sets $\tau(p\{q\})$ and $\tau(p\{r\})$, it is easy to see that each trec $r$ either has no holes and $r\in \tau(p\{q\})\cap \tau(p\{r\})$, or
    has a hole in the same position of the hole $p\{\place\}$ and there exist two trecs $r' = \para(\alpha_1\seq u_1,...,\alpha_n\seq u_n) \in \tau(q)$ and $r'' = \alpha\seq \para(u'_1,...,u'_n)\in \tau(r)$ such that
    $r\{r'\} \in \tau(p\{q\})$, $r\{r''\}\in \tau(p\{r\})$ and $(b,\alpha, (u'_1\sep...\sep u'_n)) = \Merge(\alpha_1\seq u_1\sep...\sep\alpha_n\seq u_n)$ since the return values $b$ and $\alpha$ only depend on the events $\alpha_1$,...,$\alpha_n$.
    Let $r_1 = \para(...,r\{r'\},...)$ and $r_2 = \para(..., r\{r''\},...)$ be two trecs of $\para(...,p\{q\},...)$ and $\para(...,p\{r\},...)$ respectively,
    it remains to show that $\valt(r_1)=\valt(r_2)$.
    However, this is trivial according to the case 4 of Def.~\ref{def:Valuation of the Trecs}.
    Therefore, we have the equation (\ref{equ:parallel-lemma-1}) holds for rule $(\para, \mathit{mer})$.

\end{proof}

We show that Lemma~\ref{lemma:parallel-lemma} can be extended to the case when $\lev(q\red r)$ is an arbitrary number.

\begin{lemma}
    \ifx
    \label{lemma:parallel-lemma-2}
    Given a parallel program $\para(...,p\{\place\},...)$ with a program hole $p\{\place\}$,
    if $p\{q\}\red p\{r\}$ with $q\red r$ and $\lev(q\red r) = 1$ for some programs $q$ and $r$,
    then
    \begin{equation}
    \label{equ:parallel-lemma-2}
    \val(\para(..., p\{q\},...)) = \val(\para(...,p\{r\},...)).
    \end{equation}
    \fi
    \label{lemma:parallel-lemma-2}
    The equation (\ref{equ:parallel-lemma-1}) of Lemma.~\ref{lemma:parallel-lemma} holds for any $\lev(q\red r)$.
\end{lemma}

\begin{proof}
    We prove by induction on $\lev(q\red r)$, where Lemma.~\ref{lemma:parallel-lemma} has shown the base case.
    Now for any $n>1$, we assume that the equation (\ref{equ:parallel-lemma-1}) holds for $\lev(q\red r)<n$,
    next we prove that the equation (\ref{equ:parallel-lemma-1}) holds for $\lev(q\red r) = n$.

    Since $\lev(q\red r) >1$, we know that there exist a program $u$ and two programs $u_1$, $r_1$ such that $q = u\{q_1\}$, $r = u\{r_1\}$ and $q_1\red r_1$.
    Let $p\{q_1\}_1 = p\{u\{q_1\}\}$ and $p\{r_1\}_1 = p\{u\{r_1\}\}$ (where we use $p\{\place\}_1$ to distinguish itself from $p\{\place\}$),
    since $\lev(q_1\red r_1) < n$, by induction hypothesis, we immediately have $\val(\para(...,p\{q_1\}_1,...)) = \val(\para(...,p\{r_1\}_1,...))$.
    So $\val(\para(...,p\{q\},...)) = \val(\para(...,p\{r\},...))$ holds.

\end{proof}

\begin{prop}[Soundness of Rule $(r2)$ on level $1$]
    \label{prop:Soundness of the Rewrite Rule (r2) on level 1}
    Given a closed SP $p$ and a program hole $p\{\place\}$ of $p$, for any SPs $q$ and $r$ that satisfy $q\red r$ and $\lev(q\red r) = 1$, we have
    \begin{equation}
    \label{equ:Soundness of the Rewrite Rule (r2)}
    \val(p\{q\}) = \val(p\{r\}).
    \end{equation}
\end{prop}
\begin{proof}
    We proceed the proof by induction on the syntactic structure of $p$.

    If $p\{\place\} = \place$, i.e., $p\{\place\}$ is just a hole, then we need to prove $\val(q) = \val(r)$ for the closed programs $q$ and $r$, which are the direct results of Prop.~\ref{prop:Soundness of the Rewrite Rules for Closed Sequential Programs}, Prop.~\ref{prop:Soundness of the Rewrite Rules for Parallel Programs} and Prop.~\ref{???} for all the cases when $\lev(q\red r) = 1$.

    If $p\{\place\} = p_1\{\place\}\seq p_2$, by induction hypothesis we have that $\val(p_1\{q\}) = \val(p_1\{r\})$, hence $\val(p\{q\}) = \val(p_1\{q\})\circ \val(p_2) = \val(p_1\{r\})\circ \val(p_2) = \val(p\{r\})$.

    The case for $p\{\place\} = p_1\seq p_2\{\place\}$ is similar.

    If $p\{\place\} = p_1\{\place\}\cup p_2$, by induction hypothesis we have that $\val(p_1\{q\}) = \val(p_1\{r\})$, hence $\val(p\{q\}) = \val(p_1\{q\})\cup \val(p_2) = \val(p_1\{r\})\cup \val(p_2) = \val(p\{r\})$.

    The case for $p\{\place\} = p_1\cup p_2\{\place\}$ is similar.

    If $p\{\place\} = (p_1\{\place\})^*$, by induction hypothesis we have that $\val(p_1\{q\}) = \val(p_1\{r\})$, so $\val(p\{q\}) = \bigcup^\infty_{i=0}\val^n(p_1\{q\}) = \bigcup^\infty_{i=0}\val^n(p_1\{r\}) = \val(p\{r\})$.

    If $p\{\place\} = \para(p_1,...,p_i\{\place\},...,p_n)$, we directly obtain the result by Lemma~\ref{lemma:parallel-lemma}.
\end{proof}

With Prop.~\ref{prop:Soundness of the Rewrite Rule (r2) on level 1}, now we consider the soundness of rule $(r2)$ on an arbitrary level.

\begin{prop}[Soundness of Rule $(r2)$]
\label{prop:Soundness of the Rewrite Rule r2}
\ifx
    Given a closed SP $p$ and a program hole $p\{\place\}$ of $p$, for any SPs $q$ and $r$ that satisfy $q\red r$, we have
    \begin{equation}
    \label{equ:Soundness of the Rewrite Rule (r2)}
    \val(p\{q\}) = \val(p\{r\}).
    \end{equation}
\fi
    The equation (\ref{equ:Soundness of the Rewrite Rule (r2)}) of Prop.~\ref{prop:Soundness of the Rewrite Rule (r2) on level 1} holds for any $\lev(q\red r)$.
\end{prop}

\begin{proof}
    We prove by induction on $\lev(q\red r)$.
    By Prop.~\ref{prop:Soundness of the Rewrite Rule (r2) on level 1} we obtain the base case.
    Now for any $n>1$, we assume that the equation (\ref{equ:Soundness of the Rewrite Rule (r2)}) holds for $\lev(q\red r) < n$,
    next we show that the equation (\ref{equ:Soundness of the Rewrite Rule (r2)}) holds for $\lev(q\red r) = n$.

    We prove the proposition by induction on the structure of $p\{\place\}$.

    If $p\{\place\} = \place$, i.e., $p\{\place\}$ is just a hole, then we need to prove $\val(q) = \val(r)$ for the closed programs $q$ and $r$.
    Since $\lev(q\red r) > 1$, there exist a program $u$ and two programs $q_1$, $r_1$ such that $q = u\{q_1\}$, $r = u\{r_1\}$ and $q_1\red r_1$.
    Because $\lev(q_1\red r_1) < n$, by induction hypothesis we immediately obtain that $\val(q) = \val(r)$.

    The proofs for the inductive cases of $p\{\place\} = p_1\{\place\}\seq p_2$, $p\{\place\} = p_1\seq p_2\{\place\}$, $p\{\place\} = p_1\{\place\}\cup p_2$, $p\{\place\} = p_1\cup p_2\{\place\}$ and $p\{\place\} = (p_1\{\place\})^*$
    are similar to those in Prop.~\ref{prop:Soundness of the Rewrite Rule (r2) on level 1} and we omit them here.

    The case when $p\{\place\} = \para(p_1,...,p_i\{\place\},...,p_n)$ is straightforward by Lemma~\ref{lemma:parallel-lemma-2}.

\end{proof}

\ifx
\begin{prop}
Given a program $p$ and a program hole $p\{\place\}$ of $p$, for any programs $q$ and $r$, if $\val(q) = \val(r)$, then
$$\val(p\{q\}) = \val(p\{r\}).$$
\end{prop}

\begin{prop}
Given a formula $\phi$ and a program hole $\phi\{\place\}$ of $\phi$, for any programs $p$ and $q$, if $\val(p) = \val(q)$, then
$$\val(\phi\{p\}) = \val(\phi\{q\}).$$
\end{prop}
\fi

\begin{prop}[Soundness of Rule $(\para, \mathit{seq})$]
    \label{prop:Soundness of the Rewrite Rule of (para, seq)}
    Given a well-defined parallel program $\para(p_1,...,p_n)$,
    $\val(\para(p_1,...,p_n)) = \val(\ToSeq(\para(p_1,...,p_n))$.
\end{prop}

\begin{proof}
The soundness of rule $(\para, \mathit{seq})$ is directly from Arden's rule (Prop.~\ref{prop:Arden's Rule in SPs}) and the soundness of rule $(r2)$ and other rewrite rules of the types (b) and (c).
\end{proof}

\begin{prop}[Soundness of Rule $(r1)$]
    \label{prop:Soundness of the Rewrite Rule r1}
    The rewrite rule $(r1)$ is sound.
\end{prop}

\begin{proof}
    Since the programs $p$ and $q$ in rule $(r1)$ are closed programs, it is enough to prove that
    if $\val(p) = \val(q)$, then $\val(\phi\{p\}) = \val(\phi\{q\})$.

    We proceed by induction on the syntactic structure of $\phi\{\place\}$, where the base cases are $\phi\{\place\} = [r\{\place\}] \psi$ and $\phi\{\place\} = [r\{\place\}]\Box\phi$.

    For the base cases, we consider $\phi\{\place\} = [r\{\place\}] \psi$ for example, the case of $\phi\{\place\} = [r\{\place\}]\Box\psi$ is similar.
    If $\phi\{\place\} = [r\{\place\}] \psi$,
    by the soundness of rule $(r2)$ (Prop.~\ref{prop:Soundness of the Rewrite Rule r2}) and the rules of the types (b) and (c) (Prop.~\ref{prop:Soundness of the Rewrite Rules for Closed Sequential Programs}, \ref{prop:Soundness of the Rewrite Rules for Parallel Programs} and \ref{prop:Soundness of the Rewrite Rule of (para, seq)}), we have $\val(r\{p\}) = \val(r\{q\})$.
    By Def.~\ref{def:Valuation of sDTL Formulas}, we immediately obtain the result.

    The cases for $\phi\{\place\} = \neg \psi\{\place\}$, $\phi\{\place\} = \phi_1\{\place\}\wedge \phi_2$, $\phi\{\place\} = \phi_1\wedge \phi_2\{\place\}$ and $\forall x. \psi\{\place\}$ are similar, we only take
    $\phi\{\place\} = \neg \psi\{\place\}$ as an example.
    If $\phi\{\place\} = \neg \psi\{\place\}$, by induction hypothesis we have $\val(\psi\{p\}) = \val(\psi\{q\})$, then the result is straightforward by Def.~\ref{def:Valuation of sDTL Formulas}.
\end{proof}

\section{Proofs for the Relative Completeness of SDL}
\label{section:Proofs for the Relative Completeness of SDL}

In the proofs below, we often denote an AFOL formula $\phi$ as $\phi^\flat$ to distinguish it from an SDL formula.

The proof of Theorem~\ref{theorem:The ``Main Theorem''} mainly follows that of the ``main theorem'' in~\cite{Harel79}, but is augmented to fit the condition \rmn{4}.

\begin{proof}[The proof of Theorem~\ref{theorem:The ``Main Theorem''}]
    For an SDL formula $\phi$, we can convert $\phi$ into a semantical equivalent conjunctive normal form: $C_1\wedge ...\wedge C_n$.
Each clause $C_i$ is a disjunction of literals: $C_i = l_{i,1}\vee...\vee l_{i,m_i}$, where $l_{i,j}$ ($1\le i\le n$, $1\le j\le m_i$) is an atomic SDL formula, or its negation.
By the FOL rules, it is sufficient to prove that for each clause $C_i$, $\models C_i$ implies $\reld C_i$.
We proceed by induction on the sum $n$, of the number of the appearances of $[p]$ and $\la p\ra$ and the number of quantifiers $\forall x$ and $\exists x$ prefixed to dynamic formulas,
in $C_i$.

If $n=0$, there are no appearances of $[p]$ and $\la p\ra$ in $C_i$, so $C_i$ is an AFOL formula, thus we immediately obtain $\reld C_i$.

Suppose $n>0$, it is sufficient to consider the following cases:
$$C_i = \phi_1 \vee \mathop{op} \phi_2, C_i = \phi_1\vee [p] \Box\phi_2, C_i = \phi_1\vee \la p\ra \Diamond\phi_2$$
where $\mathop{op}\in \{[p],\la p\ra, \forall x, \exists x\}$.

If $C_i= \phi_1 \vee \mathop{op} \phi_2$, which is equivalent to $\neg \phi_1\to \mathop{op} \phi_2$, by the condition \rmn{1}, there exist two AFOL formulas $\psi^\flat_1$ and $\psi^\flat_2$ such that
$\models\psi^\flat_1\leftrightarrow \neg \phi_1$ and $\models \psi^\flat_2\leftrightarrow \phi_2$. Then we have $\models \psi^\flat_1\to \mathop{op} \psi^\flat_2$ holds.
By the condition \rmn{2} we have \begin{equation}\label{eq:M1}\reld\psi^\flat_1\to \mathop{op} \psi^\flat_2.\end{equation}
Since in $\psi^\flat_1\leftrightarrow \neg \phi_1$ and $\psi^\flat_2\leftrightarrow \phi_2$ the sum is strictly less than $n$, by inductive hypothesis we can get that
\begin{equation}\label{eq:M2}\reld\neg \phi_1\to \psi^\flat_1\end{equation} and $\reld\psi^\flat_2\to \phi_2$ hold.
By the condition \rmn{3} we know that \begin{equation}\label{eq:M3}\reld \mathop{op} \psi^\flat_2\to \mathop{op} \phi_2\end{equation} holds.
Based on (\ref{eq:M1}), (\ref{eq:M2}), (\ref{eq:M3}) and the FOL rules in Table~\ref{table:Rules of FOL} and \ref{table:Other Rules of FOL} we can conclude that $\reld\neg \phi_1\to \mathop{op} \phi_2$.

If $C_i = \phi_1\vee [p]\Box\phi_2$, which is equivalent to $\neg \phi_1\to [p]\Box\phi_2$, now we prove that $\reld \neg \phi_1\to [p]\Box\phi_2$.
By the condition \rmn{1} there exists an AFOL formula $\psi^\flat_1$ such that $\models \psi^\flat_1\leftrightarrow \neg \phi_1$.
Note that in $\psi^\flat_1\leftrightarrow \neg \phi_1$ the sum is strictly less than $n$, by inductive hypothesis we have $\reld \neg \phi_1\to \psi^\flat_1$.
On the other hand,
by the condition \rmn{4}, $\models \neg\phi_1\to [p]\Box\phi_2$ and $\models \psi^\flat_1\leftrightarrow \neg \phi_1$, we have that $\reld\psi^\flat_1\to [p]\Box\phi_2$ holds. Thus we have $\reld\neg\phi_1\to [p]\Box\phi_2$.

Similar for the case $C_i = \phi_1\vee \la p\ra \Diamond\phi_2$.

\end{proof}


\begin{proof}[The proof of Theorem~\ref{theorem:The ``Main Theorem''}\rmn{1}]
It is known that an FODL formula of the form $[p]\phi^\flat$ can be expressed as an AFOL formula (refer to page 326 of~\cite{Harel00}).
In SDL, it is easy to see that closed sequential SPs have the same expressiveness as the regular programs in FODL.
Because a closed sequential SP can be taken as a regular program by simply ignoring the differences between macro and micro events in this SP which only play their roles in parallel SPs.
For example, an SP $p= (x:=1\nex x>=1?\nex \epsilon \seq y:=x+2)$ can be taken as a regular program $p' = (x:=1\seq x>=1?\seq \true?\seq y:=x+2)$ by replacing all operators $\nex$ with the operator $\seq$ and the event $\epsilon$ with $\true?$.
$p$ and $p'$ have the same semantics.

Therefore,
\begin{equation}
\label{equ:MT2-1}
\begin{gathered}
\mbox{any SDL formula of the form $[p]\phi^\flat$, where $p$ is a sequential program,}\\
\mbox{can be expressed as an AFOL formula. }
\end{gathered}
\end{equation}
Based on this fact, we now show that any SDL formula is expressible in AFOL.
Actually, we only need to consider the case when all programs of an SDL formula are sequential, because a parallel SP is semantically equivalent to a sequential one according to the procedure $\ToSeq$.

Given an SDL formula $\phi$, in which all programs are sequential, we prove by induction on the number of the appearances of $[p]$ in $\phi$. The base case is trivial.
For the inductive cases, the only non-trivial cases are $\phi = [p]\psi$ and $\phi = [p]\Box\psi$, where by inductive hypothesis there exists an AFOL formula such that $\models \psi \leftrightarrow \psi^\flat_1$.
From the fact~\ref{equ:MT2-1} above, we already prove the case $[p]\psi^\flat_1$.
It remains to show that $[p]\Box\psi^\flat_1$ is expressible according to different cases of $p$.
Below we prove $[p]\Box\psi^\flat_1$ by induction on the structure of $p$.

\begin{itemize}
\item The base cases are $p=\noth$, $p=\halt$ and $p=\alpha$.
We only take $p=\alpha$ for example, other cases are similar.
By the soundness of rule $(\alpha, \Box\phi)$, we obtain that $\models [\alpha]\Box\psi^\flat_1\leftrightarrow (\psi^\flat_1\wedge [\alpha]\psi^\flat_1)$, where by~\ref{equ:MT2-1} $[\alpha]\psi^\flat_1$ is expressible in AFOL.
So $\psi^\flat_1\wedge [\alpha]\psi^\flat_1$ is also expressible in AFOL.

\item If $p=p_1\seq p_2$, by the soundness of rule $(\seq, \Box\phi)$, we have that $\models [p_1\seq p_2]\Box\psi^\flat_1\leftrightarrow ([p_1]\Box\psi^\flat_1\wedge [p_1][p_2]\Box\psi^\flat_1)$.
By inductive hypothesis, we have $[p_1]\Box\psi^\flat_1$ and $[p_2]\Box\psi^\flat_1$ are expressible in AFOL. So by~\ref{equ:MT2-1}, $[p_1][p_2]\Box\psi^\flat_1$ is also expressible in AFOL.
So $[p_1]\Box\psi^\flat_1\wedge [p_1][p_2]\Box\psi^\flat_1$ is expressible in AFOL.

\item If $p = p_1\cup p_2$, by the soundness of rule $(\cup)$, we have that $\models [p_1\cup p_2]\Box\psi^\flat_1\leftrightarrow ([p_1]\Box\psi^\flat_1\wedge [p_2]\Box\psi^\flat_2)$.
By inductive hypothesis, $[p_1]\Box\psi^\flat_1$ and $[p_2]\Box\psi^\flat_1$ are expressible in AFOL, so is $[p_1]\Box\psi^\flat_1\wedge [p_2]\Box\psi^\flat_2$.

\item If $p=q^*$, by the soundness of rule $(*, \Box\phi)$, we have that $\models [q^*]\Box\psi^\flat_1\leftrightarrow [q^*][q]\Box\psi^\flat_1$.
By inductive hypothesis, $[q]\Box\psi^\flat$ is expressible in AFOL.
By~\ref{equ:MT2-1} we then get that $[q^*][q]\Box\psi^\flat$ is expressible in AFOL.
\end{itemize}
\end{proof}

\begin{proof}[The proof of Theorem~\ref{theorem:The ``Main Theorem''}\rmn{2}]
We prove by induction on the syntactic structure of $p$.
In $\phi^\flat\to \mathop{op} \psi^\flat$, when $\mathop{op}$ is $\forall x$ or $\exists x$, the proof is trivial, because $\phi^\flat\to \mathop{op} \phi^\flat$ itself is an AFOL formula so there must be $\reld \phi^\flat\to \mathop{op} \psi^\flat$.

We first consider the case when $\mathop{op}$ is $[p]$.

\begin{itemize}
    \item The base cases are $p=\noth$, $p=\halt$ and $p=\alpha$, we only take $p=\alpha$ for example, other cases are similar.
    We prove by induction on the number $k$ of micro events in $\alpha$. If $k=1$, then $\alpha = \epsilon$. From $\models \phi^\flat\to [\epsilon]\psi^\flat$, by the soundness of rule $(\epsilon)$,
    we have $\models \phi^\flat \to \psi^\flat$.
    Since $\phi^\flat \to \psi^\flat$ is an AFOL formula, $\reld \phi^\flat \to \psi^\flat$ holds.
    By rule $(\epsilon)$ and the FOL rules in Table~\ref{table:Rules of FOL} and \ref{table:Other Rules of FOL}, we obtain that $\reld \phi^\flat\to [\epsilon]\psi^\flat$.
    Suppose $k>1$, $\alpha$ is either of the form $\psi_0?\nex \beta$ or $x:=e\nex \beta$. We only consider the case $\alpha = (x:=e\nex \beta)$, the other case is similar.
    From $\models \phi^\flat\to [x:=e\nex \beta]\psi^\flat$, by the soundness of rule $(x:=e)$ we have that $\models \phi^\flat\to ([\beta]\psi^\flat)[e/x]$.
    By inductive hypothesis, we know that $\reld \phi^\flat\to ([\beta]\psi^\flat)[e/x]$, therefore by rule $(x:=e)$ and the FOL rules we can derive $\reld \phi^\flat\to [x:=e\nex \beta]\psi^\flat$.

    \item If $p=p_1\seq p_2$, from $\models \phi^\flat\to [p_1;p_2]\psi^\flat$, by the soundness of rule $(\seq, \phi)$, we obtain $\models \phi^\flat\to [p_1][p_2]\psi^\flat$.
    By the condition \rmn{1}, there exists an AFOL formula $\psi^\flat_1$ such that $\models [p_2]\psi^\flat \leftrightarrow \psi^\flat_1$.
    Hence $\models \psi^\flat_1\to [p_2]\psi^\flat$ and $\models \phi^\flat\to [p_1]\psi^\flat_1$. By inductive hypothesis, we have $\reld \psi^\flat_1\to [p_2]\psi^\flat$ and $\reld \phi^\flat\to [p_1]\psi^\flat_1$.
    Applying rule $([], \mathit{gen})$ on $\reld \psi^\flat_1\to [p_2]\psi^\flat$, we obtain $\reld [p_1]\psi^\flat_1\to [p_1][p_2]\psi^\flat$.
    From $\reld \phi^\flat\to [p_1]\psi^\flat_1$ and $\reld [p_1]\psi^\flat_1\to [p_1][p_2]\psi^\flat$, by applying rule $(\seq, \phi)$ and the FOL rules in Table~\ref{table:Rules of FOL} and \ref{table:Other Rules of FOL},
    we can derive $\reld \phi^\flat\to [p_1\seq p_2]\psi^\flat$.

    \item If $p=p_1\cup p_2$, from $\models \phi^\flat\to [p_1\cup p_2]\psi^\flat$, by the soundness of rule $(\cup)$, there is $\models \phi^\flat\to ([p_1]\psi^\flat\wedge [p_2]\psi^\flat)$, which is equivalent to
    $\models \phi^\flat\to [p_1]\psi^\flat$ and $\models \phi^\flat\to [p_2]\psi^\flat$.
    By inductive hypothesis we have that $\reld \phi^\flat\to [p_1]\psi^\flat$ and $\reld \phi^\flat\to [p_2]\psi^\flat$.
    By rule $(\cup)$ and the FOL rules in Table~\ref{table:Rules of FOL} and \ref{table:Other Rules of FOL}, we have $\reld \phi^\flat\to [p_1\cup p_2]\psi^\flat$.

    \item If $p=q^*$, by the condition \rmn{1}, there exists an AFOL formula $\phi^\flat_1$ such that $\models [q^*]\psi^\flat\leftrightarrow \phi^\flat_1$.
    From $\models \phi^\flat\to [q^*]\psi^\flat$, we also have $\models \phi^\flat \to \phi^\flat_1$.
    From $\models [q^*]\psi^\flat\leftrightarrow \phi^\flat_1$, by the soundness of the rules $(*)$, $(\cup)$, $(\noth)$ and $(\seq, \phi)$, it is not hard to see that
    $\models \phi^\flat_1\leftrightarrow [q^*]\psi^\flat\leftrightarrow [\noth\cup q\seq q^*]\psi^\flat \leftrightarrow [\noth]\psi^\flat\wedge [q\seq q^*]\psi^\flat \leftrightarrow \psi^\flat\wedge [q][q^*]\psi^\flat \leftrightarrow \psi^\flat\wedge [q]\phi^\flat_1$.
    From these logical equivalences we can see that $\models \phi^\flat_1\to [q]\phi^\flat_1$ and $\models\phi^\flat_1\to \psi^\flat$.
    By inductive hypothesis, from $\models \phi^\flat \to \phi^\flat_1$, $\models \phi^\flat_1\to [q]\phi^\flat_1$ and $\models\phi^\flat_1\to \psi^\flat$, we get that
    $\reld \phi^\flat \to \phi^\flat_1$, $\reld \phi^\flat_1\to [q]\phi^\flat_1$ and $\reld\phi^\flat_1\to \psi^\flat$.
    By rule $([*])$ and the FOL rules, finally there is $\reld \phi^\flat \to [q^*]\psi^\flat$.

    \item If $p=\para(q_1,...,q_n)$, by applying rule $(\para, \mathit{seq})$, we can transform $p$ into a sequential program $p'$, i.e., $p\red p'$.
    By the soundness of rule $(\mathit{r1})$, from $\models \phi^\flat\to [p]\psi^\flat$, we can get that $\models \phi^\flat\to [p']\psi^\flat$.
    Since $p'$ is sequential, we can analyze it based on the cases given above.
    Using inductive hypothesis, we can prove that $\reld \phi^\flat\to [p']\psi^\flat$.
    By rule $(\mathit{r1})$ we can obtain $\reld \phi^\flat\to [p]\psi^\flat$.
\end{itemize}

For the case when $\mathop{op}$ is $\la p\ra$, the proofs for the cases $p=\noth$, $p=\halt$, $p=\alpha$, $p=p_1\seq p_2$, $p=p_1\cup p_2$ and $p=\para(q_1,...,q_n)$ are similar to the proofs above, since $\la p\ra\psi^\flat$ equals to $\neg[p]\neg \psi^\flat$ and the rules $(\epsilon), (x:=e), (\seq, \phi), (\cup), (*), (\cup), (\noth)$ used in the proofs above are bidirectional. (For rule $([],\mathit{gen})$, we have the rule $(\la\ra, \mathit{gen})$.)
In the following we only prove the case $p=q^*$.

\begin{itemize}
    \item If $p=q^*$, by the condition \rmn{1} and the way of expressing regular programs in AFOL in~\cite{???}, we know that for any $n$, $\la q^n\ra\psi^\flat$ can be expressed as an AFOL formula $\phi^\flat_1(n)$,
    i.e., $\models \la q^n\ra \psi^\flat\leftrightarrow \phi^\flat_1(n)$.
    From the semantics of $\la q^*\ra \psi^\flat$, it is easy to see that $\models \la q^*\ra \psi^\flat\leftrightarrow \exists n\ge 0. \phi^\flat_1(n)$.
    From $\models \phi^\flat\to \la q^*\ra\psi^\flat$, there is $\models \phi^\flat\to \exists n\ge 0. \phi^\flat_1(n)$.
    On the other hand, when $n>0$, by the soundness of rule $(\seq, \phi)$, we have $\models \phi^\flat_1(n)\leftrightarrow \la q^n\ra \psi^\flat\leftrightarrow \la q\seq q^{n-1}\ra\psi^\flat\leftrightarrow \la q\ra\la q^{n-1}\ra\psi^\flat\leftrightarrow \la q\ra \phi^\flat_1(n-1)$.
    From these logical equivalences we can get that $\models \phi^\flat_1(n)\to \la q\ra \phi^\flat_1(n-1)$.
    When $n=0$, from $\models \la q^n\ra \psi^\flat\leftrightarrow \phi^\flat_1(n)$, by the soundness of rule $(\epsilon)$, we have $\models \la q^0\ra\psi^\flat \leftrightarrow \la \noth\ra \psi^\flat\leftrightarrow \psi^\flat \leftrightarrow \phi^\flat_1(0)$. So $\models \phi^\flat_1(0)\to \psi^\flat$.
    By inductive hypothesis, from $\models \phi^\flat\to \exists n\ge 0. \phi^\flat_1(n)$, $\models \phi^\flat_1(n)\to \la q\ra \phi^\flat_1(n-1)$ and $\models \phi^\flat_1(0)\to \psi^\flat$,
    we have that $\reld \phi^\flat\to \exists n\ge 0. \phi^\flat_1(n)$, $\reld \phi^\flat_1(n)\to \la q\ra \phi^\flat_1(n-1)$ and $\reld \phi^\flat_1(0)\to \psi^\flat$.
    By rule $(\la*\ra)$ and the FOL rules in Table~\ref{table:Rules of FOL} and \ref{table:Other Rules of FOL}, we obtain $\reld \phi^\flat\to \la q^*\ra \psi^\flat$.
\end{itemize}
\end{proof}

\begin{proof}[The proof of Theorem~\ref{theorem:The ``Main Theorem''} \rmn{3}]
    When $\mathop{op}\in \{[p], \la p\ra\}$, the condition \rmn{3} is in fact stated as the rules $([], \mathit{gen})$ and $(\la\ra, \mathit{gen})$.
    So we only need to prove when $\mathop{op}\in \{\forall x, \exists x\}$.
    Below we only consider the case when $\mathop{op}$ is $\forall x$. The case when $\mathop{op}$ is $\exists x$ can be similarly obtained by using the dual rules of the rules used in the proof below.

Actually, using the FOL rules in Table~\ref{table:Rules of FOL} and \ref{table:Other Rules of FOL}, we can construct the following deductions:
\begin{center}
$\infer[^{(\to r)}]
    {\cdot\Rightarrow \forall x.\phi\to \forall x.\psi}
    {\infer[^{(\forall r)}]
        {\forall x.\phi\Rightarrow \forall x.\psi}
        {\infer[^{(\forall l)}]
            {\forall x.\phi\Rightarrow \psi[x'/x]}
            {\infer[^{(\to l)}]
                {\phi[x'/x]\Rightarrow \psi[x'/x]}
                {\cdot\Rightarrow \phi[x'/x]\to \psi[x'/x]}
            }
        }
    }
$
\end{center}
where $x'$ is a new variable w.r.t. $\phi$ and $\psi$.
\end{proof}

\begin{proof}[The proof of Theorem~\ref{theorem:The ``Main Theorem''} \rmn{4}]
    As the proof of Theorem~\ref{theorem:The ``Main Theorem''} \rmn{2}, we proceed by induction on the syntactic structure of $p$.
    Below we only consider the case $\phi^\flat\to [p]\Box\psi^\flat$.
    The proof of the case $\phi^\flat\to \la p\ra \Diamond\psi^\flat$ is similar by the relation between $\la p\ra \Diamond\psi^\flat$ and its dual form $[p]\Box\psi^\flat$.

    The cases for $p=\noth$, $p=\halt$, $p=p_1\cup p_2$ and $p=\para(q_1,...,q_n)$ are similar to the corresponding cases in the proof of Theorem~\ref{theorem:The ``Main Theorem''} \rmn{2} above.
    We omit them here.

    \begin{itemize}
    \item For the base case, we only consider $p=\alpha$.
    From $\models \phi^\flat \to [\alpha]\Box\psi^\flat$, by the soundness of rule $(\alpha, \Box\phi)$, there is $\models \phi^\flat\to (\psi^\flat\wedge [\alpha]\psi^\flat)$, which is equivalent to
    $\models \phi^\flat\to \psi^\flat$ and $\models \phi^\flat\to [\alpha]\psi^\flat$.
    Obviously there is $\reld \phi^\flat\to \psi^\flat$.
    By the condition \rmn{2}, we have $\reld \phi^\flat\to [\alpha]\psi^\flat$.
    Therefore, by the FOL rules in Table~\ref{table:Rules of FOL} and \ref{table:Other Rules of FOL} we have $\reld \phi^\flat \to (\psi^\flat\wedge [\alpha]\psi^\flat)$.
    By rule $(\alpha, \Box\phi)$ and the rules in FOL, we obtain that $\reld \phi^\flat\to [\alpha]\Box\psi^\flat$.

    \item If $p=p_1\seq p_2$, from $\models \phi^\flat \to [p_1\seq p_2]\Box\psi^\flat$, by the soundness of rule $(\seq, \Box\phi)$, there is $\models \phi^\flat\to ([p_1]\Box\psi^\flat\wedge [p_1][p_2]\Box\psi^\flat)$,
    which is equivalent to $\models \phi^\flat \to [p_1]\Box\psi^\flat$ and $\models \phi^\flat\to [p_1][p_2]\Box\psi^\flat$.
    From $\models \phi^\flat \to [p_1]\Box\psi^\flat$, by inductive hypothesis we can have $\reld \phi^\flat \to [p_1]\Box\psi^\flat$.
    According to the condition \rmn{1}, there is an AFOL formula $\psi^\flat_1$ such that $\models [p_2]\Box\psi^\flat \leftrightarrow \psi^\flat_1$.
    So $\models \phi^\flat\to [p_1]\psi^\flat_1$ and $\models \psi^\flat_1 \to [p_2]\Box\psi^\flat$.
    By inductive hypothesis, there are $\reld \phi^\flat\to [p_1]\psi^\flat_1$ and $\reld \psi^\flat_1 \to [p_2]\Box\psi^\flat$.
    From $\reld \psi^\flat_1 \to [p_2]\Box\psi^\flat$, by applying rule $([],\mathit{gen})$, we get that $\reld [p_1]\psi^\flat\to [p_1][p_2]\Box\psi^\flat$.
    By $\reld \phi^\flat\to [p_1]\psi^\flat_1$ and $\reld [p_1]\psi^\flat\to [p_1][p_2]\Box\psi^\flat$, we conclude that $\reld \phi^\flat\to [p_1][p_2]\Box\psi^\flat$.
    From $\reld \phi^\flat \to [p_1]\Box\psi^\flat$ and $\reld \phi^\flat\to [p_1][p_2]\Box\psi^\flat$, by applying rule $(\seq, \Box\phi)$ and other FOL rules, we can derive $\reld \phi^\flat\to [p_1\seq p_2]\Box\psi^\flat$.

    \item If $p=q^*$, from $\models \phi^\flat\to [q^*]\Box\psi^\flat$, by the soundness of rule $(*, \Box\phi)$, we have that $\models \phi^\flat\to [q^*][q]\Box\psi^\flat$.
    According to the condition \rmn{1}, there exists an AFOL formula $\psi^\flat_1$ such that $\models [q]\Box\psi^\flat\leftrightarrow \psi^\flat_1$.
    Hence there are $\models \phi^\flat\to [q^*]\psi^\flat_1$ and $\models \psi^\flat_1\to [q]\Box\psi^\flat$.
    By the condition \rmn{2}, from $\models \phi^\flat\to [q^*]\psi^\flat_1$ we have $\reld \phi^\flat\to [q^*]\psi^\flat_1$.
    From $\models \psi^\flat_1\to [q]\Box\psi^\flat$, by inductive hypothesis, we have $\reld \psi^\flat_1\to [q]\Box\psi^\flat$.
    Applying rule $([],\mathit{gen})$, there is $\reld [q^*]\psi^\flat_1\to [q^*][q]\Box\psi^\flat$.
    From $\reld \phi^\flat\to [q^*]\psi^\flat_1$ and $\reld [q^*]\psi^\flat_1\to [q^*][q]\Box\psi^\flat$, we conclude that $\reld \phi^\flat\to [q^*][q]\Box\psi^\flat$.
    Applying rule $(*, \Box\phi)$ and other rules in FOL, it is easy to see that $\reld \phi^\flat \to [q^*]\Box\psi^\flat$.

    \end{itemize}
\end{proof}

\end{document}